%% file: ho-journal2.tex
\documentclass{lmcs}
\pdfoutput=1

\usepackage{lastpage}
\lmcsdoi{16}{3}{11}
\lmcsheading{}{\pageref{LastPage}}{}{}%
{Apr.~18,~2014}{Aug.~20,~2020}{}

\usepackage{pstricks,pst-node,pst-tree}
\usepackage[utf8]{inputenc}
\usepackage{amssymb}
\usepackage{amsmath}
\usepackage{graphicx}
\usepackage{amsthm}
\usepackage{ifthen}
\usepackage{import}
\usepackage{color}

\newcommand{\calA}{{\mathcal A}}
\newcommand{\calB}{{\mathcal B}}
\newcommand{\calC}{{\mathcal C}}
\newcommand{\calD}{{\mathcal D}}
\newcommand{\calE}{{\mathcal E}}
\newcommand{\calP}{{\mathcal P}}

\newcommand{\calS}{{\mathcal S}}
\newcommand{\calT}{{\mathcal T}}
\newcommand{\calX}{{\mathcal X}}

\newcommand{\lab}[1]{\label{#1}}%\marg{#1}}
\newcommand{\subrun}[3]{#1{\restriction}_{#2,#3}}
\newcommand{\Rz}[1]{\langle #1\rangle}
\renewcommand{\epsilon}{\varepsilon}
\renewcommand{\phi}{\varphi}
\newcommand{\topp}{\mathsf{top}} %_\mathsf{p}}
\newcommand{\tops}{\mathsf{top}} %_\mathsf{s}}

\newcommand{\pop}{\mathsf{pop}}
\newcommand{\posl}{\mathsf{pos}{\downarrow}}
\newcommand{\poslinv}{\mathsf{pos}_+}
\newcommand{\posp}{\mathsf{p}_{+1}}
\newcommand{\push}{\mathsf{push}}
\renewcommand{\read}{\mathsf{read}}
\newcommand{\collapse}{\mathsf{collapse}}
\newcommand{\branch}{\mathsf{branch}}
\newcommand{\st}{\mathsf{st}}
\newcommand{\ass}{\mathsf{ass}}
\renewcommand{\red}{\mathsf{red}}
\newcommand{\pb}{\mathsf{pb}}
\newcommand{\rd}{\mathsf{rd}}
\newcommand{\conf}{\mathsf{conf}}
\newcommand{\pow}{\mathsf{pow}}
\newcommand{\low}{\mathsf{low}}
\newcommand{\high}{\mathsf{high}}
\newcommand{\len}{\mathsf{len}}
\newcommand{\pr}{\mathsf{pr}}
\newcommand{\np}{\mathsf{np}}
\newcommand{\type}{\mathsf{type}}
\newcommand{\stype}{\mathsf{stype}}
\newcommand{\hist}{\mathsf{hist}}
\newcommand{\up}{\mathsf{up}}
\newcommand{\ret}{\mathsf{ret}}
\newcommand{\rank}{\mathit{rank}}
\newcommand{\stars}{\mathit{stars}}
\newcommand{\Nat}{\mathbb{N}}
\renewcommand{\mid}{\colon\,}
\newcommand{\sbf}{\mathbf{s}}
\renewcommand{\tt}{\mathbf{t}}
\newcommand{\uu}{\mathbf{u}}
\newcommand{\vv}{\mathbf{v}}
\newcommand{\RR}{\mathbf{R}}
\renewcommand{\SS}{\mathbf{S}}
\newcommand{\TT}{\mathbf{T}}
\newcommand{\depth}{\mathsf{depth}}
\newcommand{\prz}[2]{[#1,#2]}

\definecolor{darkgreen}{RGB}{0,191,0}
\providecommand{\new}[1]{#1}
\providecommand{\old}[1]{}
\newcommand{\dtempty}{\mathsf{empty},}
\newcommand{\dtread}{\read,}
\newcommand{\dtpop}{\pop,}
\newcommand{\dtpush}{\push,}

\begin{document}

\author{Pawe\l\ Parys}
\address{University of Warsaw, ul.~Banacha 2, 02-097 Warszawa, Poland}
\email{parys@mimuw.edu.pl}
\thanks{Work supported by the National Science Center (decision DEC-2012/07/D/ST6/02443).
The author holds a post-doctoral position supported by Warsaw Center of Mathematics and Computer Science.}

\title{On the Expressive Power of\texorpdfstring{\\~}{}
  Higher-Order Pushdown Systems}

\begin{abstract}
	We show that deterministic collapsible pushdown automata of second order can recognize
	a language that is not recognizable by any deterministic higher-order pushdown automaton (without collapse) of any order.
	This implies that there exists a tree generated by a second order collapsible pushdown system (equivalently, by a recursion scheme of second order) that is not generated by any deterministic higher-order pushdown
	system (without collapse) of any order (equivalently, by any safe recursion scheme of any order).
	As a side effect, we present a pumping lemma for deterministic higher-order pushdown automata, which potentially can be useful for other
	applications.
\end{abstract}

\titlecomment{This is a full version of our conference paper \cite{ho-new}}

\maketitle

\section{Introduction}

Already in the 70's, Maslov \cite{Mas74,Mas76} generalized the concept of pushdown automata to higher-order pushdown automata ($n$-PDA)
by allowing the stack to contain other stacks rather than just atomic elements.
In the last decade, renewed interest in these automata has arisen. 
They are now studied not only as acceptors of string languages, but also as generators of graphs and trees.
It was an interesting problem whether the class of trees generated by $n$-PDA coincides with the class of trees generated by order-$n$ recursion schemes.
Knapik, Niwiński, and Urzyczyn \cite{easy-trees} showed something similar but different: that this class coincides
with the class of trees generated by \emph{safe} order-$n$ recursion schemes (safety is a syntactic restriction on the recursion scheme),
and Caucal \cite{Caucal02} gave another characterization: trees of order $n+1$ are obtained from trees of order $n$
by an MSO-interpretation of a graph, followed by application of unfolding.

Driven by the question whether safety implies a semantical restriction to recursion schemes
Hague, Murawski, Ong, and Serre \cite{collapsible} extended the model of $n$-PDA to order-$n$ collapsible pushdown automata ($n$-CPDA) by
introducing a new stack operation called collapse,
and proved that the class of trees generated by $n$-CPDA coincides with the class of trees generated by order-$n$ recursion schemes
(earlier, Knapik, Niwiński, Urzyczyn, and Walukiewicz \cite{panic} introduced panic automata, a model equivalent to $2$-CPDA).
Let us mention that these trees have decidable MSO theory \cite{ong-lics}, 
and that higher-order recursion schemes have close connections with verification of some real life higher-order programs \cite{Kobayashi09}.

Nevertheless, it was still an open question whether these two hierarchies of trees are possibly the same hierarchy?
This problem was stated in Knapik et al.~\cite{easy-trees} 
and repeated in other papers concerning higher-order pushdown automata \cite{panic,AehligMO05,ong-lics,collapsible}.
A partial answer to this question was given in our previous paper \cite{parys-panic}: 
there is a tree generated by a $2$-CPDA that is not generated by any $2$-PDA.
We prove the following stronger property.

\begin{thm}\lab{thm:main-tree}
	There is a tree generated by a $2$-CPDA (equivalently, by a recursion scheme of order $2$) that is not generated by any $n$-PDA, for any $n$
	(equivalently, by any safe recursion scheme of any order).
\end{thm}

This confirms that the correspondence between higher-order recursion schemes and higher-order pushdown automata is not perfect.
The tree used in Theorem~\ref{thm:main-tree} (after some adaptations) comes from Knapik et al.~\cite{easy-trees} and from that time was conjectured to be a good example.

In this paper we work with PDA that recognize words instead of generating trees.
While in general PDA used to recognize word languages can be nondeterministic, 
trees generated by PDA closely correspond to word languages recognized by deterministic PDA.
Technically, we prove the following theorem, from which Theorem~\ref{thm:main-tree} follows (it is shown in Section~\ref{sec:words-trees} how these theorems are related).

\begin{thm}\lab{thm:main-det}
	There is a language recognized by a deterministic $2$-CPDA that is not recognized by any deterministic $n$-PDA, for any $n$.
\end{thm}

As a side effect, in Section~\ref{sec:pumping} we present a pumping lemma for higher-order pushdown automata.
Although its formulation is not very natural, we believe it may be useful for some other applications.
The lemma is similar to the pumping lemma from our another paper \cite{parys-pumping}; see Section~\ref{sec:pumping} for some comments.
Earlier, several pumping lemmas related to the second order of the pushdown hierarchy were proposed \cite{hayashi-pumping,gilman-pumping,kartzow-pumping}.

This paper is an extended version of our conference paper \cite{ho-new}.
The proof of Theorem~\ref{thm:main-tree} goes along the same line, but with essential differences in details.
The part about types (Section~\ref{sec:types}) was simplified slightly, in the cost of complicating other parts 
(which was necessary since Theorem~\ref{thm:types} is now proven in a weaker form than in the conference paper).

\subsection{Related Work}

One may ask a similar question for word languages instead of trees: is there a language recognized by a CPDA that is not recognized by any (nondeterministic) PDA?
This is an independent problem.
The answer is known only for order $2$ and is opposite:
one can see that in $2$-CPDA the collapse operation can be simulated by nondeterminism, 
hence $2$-PDA and $2$-CPDA recognize the same languages \cite{AehligMO05}.
It is also an open question whether all word languages recognized by CPDA are context-sensitive.

We have shown \cite{collapse-data} that the collapse operation increases the expressive power of deterministic higher-order pushdown automata with data.
In this model of automata each letter from the input word is equipped by a data value, which comes from an infinite set;
these data values can be stored on the stack and compared with other data values.
In such a setting the proof becomes easier than in the no-data case considered in this paper.

One can consider configuration graphs of $n$-PDA and $n$-CPDA, and their $\epsilon$-closures.
We know \cite{collapsible} that there is a $2$-CPDA whose configuration graph has undecidable MSO theory, 
hence which is not a configuration graph of an $n$-PDA, nor an $\epsilon$-closure of such,
as they all have decidable MSO theories.

Engelfriet \cite{Engelfriet91} showed that the hierarchies of word languages and of trees generated by PDA are strict 
(that is, for each $n$ there is a language recognized by an $n$-PDA that is not recognized by any $(n-1)$-PDA, and similarly for trees).
As observed by Hau{\ss}ner and Kartzow \cite{HeussnerKartzow}, his proof works equally well for these hierarchies for CPDA, once we know that the reachability problem for $n$-CPDA is $(n-1)$-EXPTIME complete
(which follows from Kobayashi and Ong \cite{emptiness-n-1-exptime}).

\section{Preliminaries}

%While in general PDA used to recognize word languages can be nondeterministic, 
%trees generated by PDA closely correspond to word languages recognized by deterministic PDA.
%In fact in the whole proof we use word-recognizing deterministic PDA.

For natural numbers $a$, $b$, where $b\geq a-1$, by $\prz{a}{b}$ we denote the set $\{a,\dots,b\}$ (which is empty if $b=a-1$).

In the whole paper, the letter $n$ is used exclusively for the order of pushdown automata, which is usually assumed to be fixed and known implicitly.

We now define \emph{stacks of order $k$} (\emph{$k$-stacks} for short).
Traditionally, a $0$-stack is just a single symbol, and a $k$-stack for $k\geq 1$ is a (possibly empty) sequence of nonempty $(k-1)$-stacks.
However, having a $k$-stack that is a part of an $r$-stack for $k<r$, it is convenient to know where this $k$-stack is located in the $r$-stack.
For this reason, we equip every element of a stack by its position, written as a vector of natural numbers.
Thus, for a fixed alphabet $\Gamma$ (of \emph{stack symbols}), a stack of order $0$ is a pair $(\gamma,x)$, 
where $\gamma\in\Gamma$ and $x=(x_n,x_{n-1},\dots,x_1)$ is a vector of $n$ positive integers, called a \emph{position}.
Then, for $k\in\prz{1}{n}$ we define $k$-stacks by induction:
a $k$-stack is a list $[s_1,s_2,\dots,s_m]$ of nonempty $(k-1)$-stacks (where, by convection, all $0$-stacks are nonempty) 
for which there exist numbers $x_n,x_{n-1},\dots,x_{k+1}$ such that, for $i\in\prz{1}{m}$, all positions in $s_i$ are of the form $(x_n,x_{n-1},\dots,x_{k+1},i,y_{k-1},y_{k-2},\dots,y_1)$.
By $\Gamma^k_*$ and $\Gamma^k_+$ we denote the the set of order-$k$ stacks, and the set of nonempty order-$k$ stack, respectively, where $k\in\prz{0}{n}$.
The top of a stack is on the right.

For example, when we have a $3$-stack $s$, and $n=5$, then the second $0$-stack of the third $1$-stack (counting from the bottom) of the bottommost $2$-stack of $s$
is of the form $(\gamma,(x_5,x_4,1,3,2))$, where $x_5$ and $x_4$ say where $s$ is located in an imaginary $5$-stack; the numbers $x_5$ and $x_4$ should be the same in the whole $s$.

For a $k$-stack $s^k$, where $k\in\prz{0}{n-1}$, let $\posp(s^k)$ be the $k$-stack obtained from $s^k$ by increasing the $(n-k)$-th coordinate of all its positions by $1$.
For example $\posp((\gamma,(2,3)))=(\gamma,(2,4))$, and $\posp([(\gamma,(2,1)),(\gamma,(2,2))])=[(\gamma,(3,1)),(\gamma,(3,2))]$.

Let us emphasize that when for two $k$-stacks $s^k$, $t^k$ we write $s^k=t^k$, we mean that not only their contents are equal, 
but also positions contained in their $0$-stacks are equal; thus, when $s^k$ and $t^k$ come from the same $n$-stack, 
this actually means that $s^k$ and $t^k$ refer to the same $k$-stack.

While comparing two stacks, we sometimes need to ignore positions contained in their $0$-stacks, and compare only their contents.
For a $k$-stack $s^k$, let \emph{positionless stack} $\posl(s^k)$ be the list of lists of ... of lists of stack symbols obtained from $s^k$ by removing positions from all $0$-stacks.
We say that two $k$-stacks $s^k, t^k$ are \emph{positionless-equal}, denoted $s^k\cong t^k$, when $\posl(s^k)=\posl(t^k)$.
When $s^n_-$ is a positionless $n$-stack, there is a unique $n$-stack $s^n$ such that $s^n_-=\posl(s^n)$; we write $\poslinv(s^n_-)$ for $s^n$.

The \emph{size} of a $k$-stack $s^k$, denoted $|s^k|$, is the number of $(k-1)$-stacks it contains.
When $s^k=[s_1,s_2,\dots,s_m]\in\Gamma^k_*$, and $s^{k-1}\in\Gamma^{k-1}_+$, and $[s_1,s_2,\dots,s_m,s^{k-1}]$ is a valid $k$-stack,
we denote this $k$-stack by $s^k:s^{k-1}$.
The operator ``$:$'' is assumed to be right associative (i.e., e.g., $s^2:s^1:s^0=s^2:(s^1:s^0)$).
When $0\leq k\leq r$, and $s^r=t^r:t^{r-1}:\dots:t^k\in\Gamma^r_+$, by $\tops^k(s^r)$ we denote the topmost $k$-stack of $s^r$, that is, $t^k$.
We use the name \emph{positionless topmost $k$-stack} for $\posl(\topp^k(\cdot))$.

When $\Gamma$ is fixed, the \emph{stack operations} of order $k\geq 1$ are $\pop^k$ and $\push^k_\gamma$ for each $\gamma\in\Gamma$.
We can apply them to a nonempty $r$-stack for $r\geq k$, which gives the following:
\begin{itemize}
\item	$\pop^k(s^r:s^{r-1}:\dots:s^k:s^{k-1})=s^r:s^{r-1}:\dots:s^k$, that is, we remove the topmost $(k-1)$-stack;
	it is defined only when the topmost $k$-stack contains at least two $(k-1)$-stacks;
\item	$\push^k_\gamma(s^r:s^{r-1}:\dots:s^0)=s^r:s^{r-1}:\dots:s^{k+1}:(s^k:s^{k-1}:\dots:s^0):\posp(s^{k-1}:s^{k-2}:\dots:s^1:(\gamma,x))$ for $s^0=(\gamma',x)$,
	that is, we duplicate the topmost $(k-1)$-stack, and then we replace the topmost stack symbol by $\gamma$, adjusting appropriately all positions.\footnote{
	 	In the classical definition the topmost symbol can be changed only when $k=1$ (for $k\geq 2$ it required that $\gamma=\gamma'$).
	 	We make this (unimportant) extension to have a uniform definition of $\push^k$ for all $k$.
	 }
\end{itemize}

A \emph{deterministic word-recognizing pushdown automaton of order $n$} (\emph{$n$-DPDA} for short) is a tuple $(A,\Gamma,\gamma_I,Q,q_I,F,\delta)$ 
where $A$ is an input alphabet, $\Gamma$ is a stack alphabet, $\gamma_I\in\Gamma$ is an initial stack symbol, 
$Q$ is a set of states, $q_I\in Q$ is an initial state, $F\subseteq Q$ is a set of accepting states,
and $\delta$ is a transition function that maps every element of $Q\times\Gamma$ into one of the following objects:
\begin{itemize}
\item	$\read(\vec{q})$, where $\vec{q}:A\to Q$ is an injective function, or
\item	$(q,op)$, where $q\in Q$ and $op$ is a stack operation of order at most $n$.
\end{itemize}

A \emph{configuration} of $\calA$ consists of a state and of a nonempty $n$-stack, that is, it is an element of $Q\times\Gamma_+^n$.
The \emph{initial} configuration consists of the initial state $q_I$ and of the $n$-stack containing only one $0$-stack, 
enclosing the initial stack symbol $\gamma_I$.
We use the notation $\pi_i((p_1,\dots,p_k))=p_i$; in particular for a configuration $c$, $\pi_1(c)$ denotes its state, and $\pi_2(c)$ its stack.
Additionally, for a set $X$ of tuples we define $\pi_i(X)$ to be $\{\pi_i(p)\mid p\in X\}$.
In order to shorten the notation, for a configuration $c$ we sometimes write $\tops^k(c)$ or $\pop^k(c)$ for $\tops^k(\pi_2(c))$ or $\pop^k(\pi_2(c))$, respectively.

We use a shorthand $\delta(c)$ for a configuration $c$ to denote $\delta(\pi_1(c),\posl(\tops^0(c)))$.
A configuration $d$ is a \emph{successor} of a configuration $c$, if
\begin{itemize}
\item	$\delta(c)=\read(\vec{q})$, and $d=(\vec{q}(a),\pi_2(c))$ for some $a\in A$, or
\item	$\delta(c)=(q,op)$, and $d=(q,op(\pi_2(c)))$.
\end{itemize}
Notice that a configuration $c$ has 
\begin{itemize}
\item	$|A|$ successors, if the transition is $\read(\vec{q})$;
\item	no successors, if the operation is $\pop^k$ but there is only one $(k-1)$-stack on the topmost $k$-stack;
\item	one successor, otherwise.
\end{itemize}

Next, we define a \emph{run} of $\calA$. 
For $0\leq i \leq m$, let $c_i$ be a configuration.
A run $R$ from $c_0$ to $c_m$ is a sequence $c_0,c_1,\dots,c_m$ such that, for each $i\in\prz{1}{m}$, $c_i$ is a successor of $c_{i-1}$.
We set $R(i)=c_i$ and call $\lvert R\rvert=m$ the \emph{length of $R$}. 
The \emph{subrun} $\subrun{R}{i}{j}$ is $c_i,c_{i+1},\dots,c_j$. 
For runs $R,S$ with $R(\lvert R \rvert)=S(0)$, we write $R\circ S$ for the \emph{composition} of $R$ and $S$ that is defined as expected. 
Sometimes we also consider infinite runs, such that the sequence $c_0,c_1,c_2,\dots$ is infinite.
However, unless stated explicitly, a run is finite.

The \emph{word read by a run} is a word over the input alphabet $A$.
For a run from a configuration $c$ to its successor $d$, it is the empty word if the transition between them is of the form $(q,op)$.
If the transition is $\read(\vec{q})$, this is the one-letter word consisting of the letter $a$ for which $\pi_1(d)=\vec{q}(a)$ (this letter is determined uniquely, as $\vec{q}$ is
injective).
For a longer run $R$ this is defined as the concatenation of the words read by the subruns $\subrun{R}{i-1}{i}$ for $i\in\prz{1}{|R|}$.
A run is \emph{accepting} if it ends in a configuration whose state is accepting.
A word $w$ is \emph{accepted by $\calA$} if it is read by some accepting run starting in the initial configuration.
%A word $w$ is \emph{accepted by $\calA$} if it is read by some run from the initial configuration to a configuration with accepting state.
The \emph{language recognized by $\calA$} is the set of words accepted by $\calA$.

\subsection{Collapsible $2$-DPDA}

In Section~\ref{sec:lang} we also use deterministic collapsible pushdown automata of order $2$ ($2$-DCPDA for short).
Such automata are defined like $2$-DPDA, with the following differences.
A $0$-stack contains now three parts: a symbol from $\Gamma$, a position, and a natural number, but
still only the symbol (together with a state) is used to determine which transition is performed from a configuration.
The $\push^1_\gamma$ operation sets the number in the topmost $0$-stack to the current size of the $2$-stack (while $\push^2_\gamma$ does not modify these numbers).
We have a new stack operation $\collapse$.
Its result $\collapse(s)$ is obtained from $s$ by removing its topmost $1$-stacks, so that only $k-1$ of them are left, where
$k$ is the number stored in $\tops^0(s)$ (intuitively, we remove all $1$-stacks on which the topmost $0$-stack is present).

\section{Relation between Word Languages and Trees}\label{sec:words-trees}

In this section we describe how word languages recognized by DPDA are related to trees generated by PDA.
Before seeing how Theorem~\ref{thm:main-det} implies Theorem~\ref{thm:main-tree}, we need to define how $n$-PDA are used to generate trees.
We consider ranked, potentially infinite trees.
Beside of the input alphabet $A$ we have a function $\rank\colon A\to\Nat$;
a tree node labelled by some $a\in A$ has always $\rank(a)$ children.

Automata used to generate trees are defined like DPDA or DCPDA (in particular they are deterministic as well), 
with the difference that they do not have the set of accepting states, 
and that instead of the $\read(\vec{q})$ transitions, there are $\branch(a,q_1,q_2,\dots,q_{\rank(a)})$ transitions, 
for $a\in A$, and for pairwise distinct states $q_1,q_2,\dots,q_{\rank(a)}\in Q$.
If the transition from $c$ is $\delta(c)=\branch(a,q_1,q_2,\dots,q_{\rank(a)})$, in a successor $d$ of $c$ we have $\pi_2(d)=\pi_2(c)$ and $\pi_1(d)=q_i$
for some $i\in\prz{1}{\rank(a)}$ (in particular $c$ has no successors if $\rank(a)=0$).

Let $T(\calA)$ be the set of all configurations $c$ of $\calA$ reachable from the initial one, such that a $\branch$ transition should be performed from $c$.
If there is a configuration of $\calA$ reachable from the initial one, from which there is no run to a configuration from $T(\calA)$,
by definition $\calA$ does not generate any tree.
Otherwise, a \emph{tree generated by $\calA$} has runs from the initial configuration to a configuration from $T(\calA)$ as its nodes.
A node $R$ is labelled by $a\in A$ such that $\delta(R(|R|))=\branch(a,q_1,q_2,\dots,q_{\rank(a)})$.
A node $S$ is its $i$-th child ($1\leq i\leq \rank(a)$), if $S$ is the composition of $R$ and a run $S'$ that uses a $\branch$ transition only in its first
transition, and for which $\pi_1(S'(1))=q_i$.
Notice that the graph obtained this way is really an $A$-labelled ranked tree.

We now see how Theorem~\ref{thm:main-tree} follows from Theorem~\ref{thm:main-det}.
Let $L\subseteq A^*$ be the language recognized by a $2$-DCPDA $\calA$ that is not recognized by any $n$-DPDA, for any $n$ ($L$ exists by Theorem~\ref{thm:main-det}).
First, we transform $\calA$ into a $2$-DCPDA $\calB$, recognizing $L$ as well, such that each configuration of $\calB$ reachable from the initial one has a successor.
Observe that the only reason why in $\calA$ there may be configurations with no successors is that it wants to empty a stack using a $\pop$ operation.
To avoid such situations, $\calB$ should have some bottom-of-stack marker $\bot$ on the bottom of each $1$-stack, and on the bottom of the $2$-stack (a $1$-stack containing only the $\bot$ marker).
Thus, $\calB$ starts with the $\bot$ marker as the initial stack symbol, performs $\push^2_{\bot}$ and $\push^1_{\gamma_I}$, placing the original initial stack symbol $\gamma_I$.
Then, whenever $\calA$ blocks because it wants to empty a stack, in $\calB$ the bottom-of-stack marker is uncovered; 
in such a situation $\calB$ starts some loop with no accepting state.
There is also a technical detail, that a $\pop$ operation that would block $\calA$, in $\calB$ can enter an accepting state;
to overcome this problem, every $\pop$ operation ending in an accepting state should first end in some auxiliary, not accepting state, from which (if the bottom-of-stack marker is not seen) the
accepting state is reached.

Next, we create a tree-generating $2$-CPDA $\calC$, which generates a tree over the alphabet $B=\{{X},{Y},{Z}\}$, where $\rank({X})=|A|$ 
and $\rank({Y})=\rank({Z})=1$.
It is obtained from $\calB$ in two steps. 
First, we replace each transition $\read(\vec{q})$ of $\calB$ by the transition $\branch({X},\vec{q}(a_1),\vec{q}(a_2),\dots,\vec{q}(a_{|A|}))$, where $A=\{a_1,\dots,a_{|A|}\}$.
Then, in each transition we replace the resulting state $q$ by a fresh auxiliary state $\overline{q}$, and from $\overline{q}$ (for any topmost stack symbol) we perform
transition $\branch({Y},q)$ if $q$ was accepting, or transition $\branch({Z},q)$ if $q$ was not accepting
(this way, after each step of the original automaton, we perform a transition $\branch({Y},\cdot)$ or $\branch({Z},\cdot)$).
Notice that from each configuration of $\calC$ reachable from the initial one, there exists a run to a configuration from $T(\calC)$, as required by the definition of a
tree-generating CPDA.
Let $t_\calC$ be the tree generated by $\calC$.

Finally, suppose that $t_\calC$ can also be generated by some $n$-PDA $\calD$ (without collapse).
From $\calD$ we create a word-recognizing $n$-DPDA $\calE$.
We replace each transition of the form $\branch({X},q_1,q_2,\dots,q_{|A|})$ of $\calD$ by the transition $\read(\vec{q})$, where $\vec{q}(a_i)=q_i$.
We replace each transition $\branch({Y},q)$ of $\calD$ by the transition $(p,\push^1_\gamma)$ for a fresh accepting state $p$ and some stack symbol $\gamma$;
from $(p,\gamma)$ we perform the transition $(q,\pop^1)$ (thus, we replace $\branch({Y},q)$ by a pass through an accepting state).
The same for a $\branch({Z},q)$ transition, but the fresh state $p$ is not accepting.

Notice that $\calE$ recognizes $L$; this contradicts our assumptions about $L$, so $t_\calC$ is not generated by any $n$-PDA.
Indeed, take any word $w\in L$.
We have an accepting run of $\calB$ that reads $w$ and starts in the initial configuration.
This run corresponds to a run of $\calC$, that is, to a path $p$ in $t_\calC$ from the root to a ${Y}$-labelled node.
Letters of $w$ tell us which child the path $p$ chooses in ${X}$-labelled nodes: if $i$-th letter of $w$ is $a_j$, 
then from the $i$-th ${X}$-labelled node of $p$, the path continues to the $j$-th child.
This path $p$ corresponds also to a run of $\calD$, so to a run of $\calE$.
This run starts in the initial configuration, ends with an accepting state, and reads $w$; thus, $\calE$ accepts $w$.
Similarly, each word accepted by $\calE$ is also accepted by $\calB$.

We also recall that a tree is generated by a recursion scheme of order $2$ if and only if it is generated by a $2$-CPDA \cite{collapsible}, 
and that a tree is generated by a safe recursion scheme of order $n$ if and only if it is generated by an $n$-PDA \cite{easy-trees};
this implies the ``equivalently'' parts of Theorem~\ref{thm:main-tree}.

\section{The Separating Language}\lab{sec:lang}

In this section we define a language $U$ that can be recognized by a $2$-DCPDA, but not by any $n$-DPDA, for any $n$.
It is a language over the alphabet $A=\{[,],\star,\sharp\}$.
For a word $w\in \{[,],\star\}^*$ we define $\stars(w)$.
Whenever in some prefix of $w$ there are more closing brackets than opening brackets, $\stars(w)=0$.
Also when in the whole $w$ we have the same number of opening and closing brackets, $\stars(w)=0$.
Otherwise, let $\stars(w)$ be the number of stars in $w$ before the last opening bracket that is not closed.
Let $U$ be the set of words $w\sharp^{\stars(w)+1}$, for any $w\in\{[,],\star\}^*$ 
(i.e., these are words $w$ consisting of brackets and stars, followed by $\stars(w)+1$ sharp symbols).

It is known that languages similar to $U$ can be recognized by a $2$-DCPDA (cf., e.g., Aehlig, de Miranda, and Ong~\cite{AehligMO05}), but for completeness we briefly show it below.
The $2$-DCPDA uses three stack symbols: $X$ (used to mark the bottom of $1$-stacks), $Y$ (used to count brackets), 
$Z$ (used to mark the bottommost $1$-stack).
The initial symbol is $X$.
The automaton first pushes $Z$, makes a copy of the $1$-stack (i.e., it performs $\push^2_{Z}$), and pops $Z$ 
(hence the first $1$-stack is marked with $Z$, unlike any other $1$-stack used later).
Then, for an opening bracket we push $Y$,
for a closing bracket we pop $Y$,
and for a star we perform $\push^2_{\gamma}$ (where $\gamma$ is the topmost stack symbol).
Hence for each star we have a $1$-stack and on the last $1$-stack we have as many $Y$ symbols as the number of currently open brackets.
If for a closing bracket the topmost symbol is $X$, it means that in the word read so far we have more closing brackets than opening brackets;
in this case we should accept suffixes of the form $\{[,],\star\}^*\sharp$, which is easy.

Finally, the $\sharp$ symbol is read.
If the topmost symbol is $X$, we have read as many opening brackets as closing brackets, hence we should accept one $\sharp$ symbol.
Otherwise, the topmost $Y$ symbol corresponds to the last opening bracket that is not closed.
We execute the $\collapse$ operation.
It leaves the $1$-stacks created by the stars read before this bracket, except one (plus the first $1$-stack). 
Thus, the number of $1$-stacks is precisely equal to $\stars(w)$.
Now we should read as many $\sharp$ symbols as we have $1$-stacks, plus one (after each $\sharp$ symbol we perform $\pop^2$), 
and then accept.

In the remaining part of the paper we prove that any $n$-DPDA cannot recognize $U$;
in particular all automata appearing in the following sections do not use collapse.

\section{Overview of the Proof}

Before we start the real proof, in this section we present its general structure, on the intuitive level.
Let us first see why $U$ cannot be recognized by any $1$-DPDA $\calA$.
Consider the input word
\begin{align*}
	w_1=[\star^{n_1}[\star^{n_2}\dots[\star^{n_N}[\star^{m_{N+1}}]\star^{m_N}]\dots\star^{m_1}]\star^{m_0}[
\end{align*}
(where each bracket is matched, except the last opening bracket).
Notice that $\stars(w_1)$ equals the sum of all $n_{i}$ and $m_i$, so $\calA$, after reading $w_1$, has to store all these numbers in its stack.
Thus, it first stores the number $n_1$ on the stack (by repeating some stack symbol $n_1$ times),
then it can mark that there was an opening bracket, then it stores $n_2$, and so on (see Figure~\ref{fig-stack-w1});
none of these numbers can be removed later.
\begin{figure*}
\begin{center}
	\import{pics/}{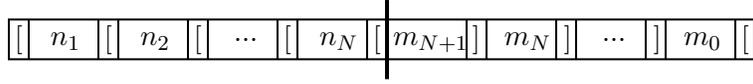}
\end{center}
\caption{The stack of a $1$-DPDA after reading the word $w_1$}
\label{fig-stack-w1}
\end{figure*}
Now consider the prefix $w_{1,i}$ of $w_1$ that ends just after the $i$-th closing bracket.
Since $\calA$ is deterministic, the stack at the end of $w_{1,i}$ looks similar: it is just shorter, but for sure it ends to the right of the vertical line, 
which denotes the stack size after the last opening bracket.
We see that $\stars(w_{1,i})=n_1+\dots+n_{N-i}$.
Thus, when $\calA$ sees a $\sharp$ after $w_{1,i}$, it has to remove (ignore) the numbers above $n_{N-i}$, and sum the rest.
In particular it passes the vertical line in some state $q_i$.
We see that for each $i$, at the moment of crossing this line, the stack is the same (everything to the right of the line is removed), only the state $q_i$ can differ.
So in fact each $q_i$ has to be different, since for each $i$ we expect a different behavior.
This is a contradiction when $N$ is greater than the number of states.

It follows that $\calA$ is of order at least $2$, and while reading $w_1$ at some moment a push of order $2$ has to be performed,
where in the topmost $1$-stack we don't remember some of the numbers $n_i$ or $m_i$
(for example, in order to recognize $w_1$, after each $]$ we can copy the topmost $1$-stack, 
and remove a fragment of its copy, so that the matching opening bracket is on the top).
But now we can consider the word
\begin{align*}
	w_2=w_1\star^{n'_1}w_1\star^{n'_2}\dots w_1\star^{n'_N}w_1\star^{m'_{N+1}}]\star^{m'_N}]\dots\star^{m'_1}]\star^{m'_0}[\,,
\end{align*}
where the numbers $n_i, m_i$ in each copy of $w_1$ are independent (so in fact each $w_1$ is a different word).
Notice that each $w_1$ ends by an unmatched opening bracket; they are matched by the closing brackets at the end of $w_2$.
We can now almost repeat the previous reasoning.
First, $\stars(w_2)$ equals the sum of all numbers, so they all have to be kept on the stack.
Then, we draw a line after reading the last $w_1$ (that is, separating the $1$-stacks created before that moment from those created later).
By the order-$1$ argument, some number from each $w_1$ is not present in the topmost $1$-stack after reading this $w_1$,
so it cannot be present above the line.
Next, for each $i$ we try to end the word already after the $i$-th closing bracket (among those at the end of $w_2$, not those inside words $w_1$).
When we have a $\sharp$ after each of these prefixes, we have to go below the line and behave differently
(include a different subset of those values which are not present above the line),
so we have to cross the line in different states.
This is again a contradiction when $N$ is greater than the number of states.
By induction we can continue like this, and nesting the words $w_n$ again we can show that for each order of the DPDA there is a problem.

Although the above idea of the proof looks simple, formalizing it is not straightforward.
We have to deal with the following issues:
\begin{enumerate}
\item	Above we have argued why a $1$-DPDA cannot deal correctly with the word $w_1$.
	But in fact we should consider any $n$-DPDA, and prove that it is impossible that it stores all numbers from $w_1$ inside one $1$-stack.
	Then there arises a problem: when crossing ``the line'' it is no longer true that the stack can only be of one form.
	Indeed, the topmost $1$-stack has one fixed form, but we can cross the line in a copy of this $1$-stack, with anything below this $1$-stack.
	We can even cross the line multiple times, in several copies of the $1$-stack.
	Thus, it is no longer true that the number of states gives the number of ways in which we can visit a substack.
	The ways of visiting a substack are described by types of stacks and by types of sequences of configurations, defined in Section~\ref{sec:types}.
	The key point is that there are finitely many types for a fixed DPDA.
\item	Where exactly is a number stored in a stack?
	And, where exactly ``the line'' should be placed?
	This is not sharp, since a DPDA may delay some stack operations by keeping information in its state, 
	as well as it may temporarily create some fancy redundant structures on the stack, which are removed later in the run.
	To deal with this issue, in Section~\ref{sec:milestone} we define milestone configurations.
	Intuitively, these are configurations in which no additional garbage is present on the stack.
\item	Finally, why it would be wrong when, while reading the $\sharp$ symbols, the automaton did not visit a place where there is stored a number that is a part of $\stars(\cdot)$?
	Maybe, accidentally, this number is equal to some other amount in the stack.
	Or maybe it was propagated to some other region on the stack by some involved manipulations.
	To overcome this difficulty, in Section~\ref{sec:pumping} we prove a pumping lemma.
	It allows to change any of the numbers in the input word,
	without altering too much the whole stack.
	If some number (included in $\stars(\cdot)$) is changed, the DPDA has to enter the part of the stack changed by the pumping lemma;
	otherwise it would incorrectly accept after the same number of the $\sharp$ symbols for two words with different $\stars(\cdot)$.
\end{enumerate}

\section{The History Function and Special Runs}\lab{sec:basic}

We begin this section by defining the history function.
Then we define two classes of runs that are particularly interesting for us, namely $k$-upper runs and $k$-returns.

For any run $R$ and any $k$-stack $s^k$ of $R(|R|)$, where $k\in\prz{0}{n}$, we define a $k$-stack $\hist(R,s^k)$.
Intuitively, $\hist(R,s^k)$ is the (unique) $k$-stack of $R(0)$, which evolved to the $k$-stack $s^k$ in $R(|R|)$.
Formally, we define $\hist(R,s^k)$ by induction on the length of $R$, starting with the case of $k=0$.
When $|R|=0$, we take $\hist(R,s^0)=s^0$.
Consider now a longer run $R=S\circ T$ with $|T|=1$.
We take $\hist(R,s^0)=\hist(S,s^0)$ if the last transition of $R$ is $\read$ or performs $\pop$,  
as well as if the transition performs $\push^r_{\gamma}$ and $s^0$ is not in the topmost $(r-1)$-stack of $R(|R|)$.
If the last transition of $R$ performs $\push^r_{\gamma}$ and $s^0$ is in the topmost $(r-1)$-stack of $R(|R|)$, 
then $\hist(R,s^0)=\hist(S,t^0)$, where $t^0$ is equal to $s^0$ with the $(n-r+1)$-th coordinate of its position decreased by $1$ 
(i.e.,~$t^0$ is the $0$-stack of $T(0)$ from which $s^0$ was obtained as a copy).
Notice that (for technical convenience) $\hist$ works in this way also for the topmost $0$-stack, 
although the content of the topmost $0$-stack changes during the $\push^r_\gamma$ operation.
For $k>0$, we define $\hist(R,s^k)$ to be the $k$-stack of $R(0)$ containing $\hist(R,s^0)$ for all $0$-stacks $s^0$ in $s^k$ 
(observe that when $s^0$, $t^0$ are two $0$-stacks in $s^k$, the $0$-stacks $\hist(R,s^0)$ and $\hist(R,t^0)$ are in the same $k$-stack).

It is important to notice that whenever $R=S\circ T$, then $\hist(S,\hist(T,s^k))=\hist(R,s^k)$.
In the sequel we extensively use this property, which we call \emph{compositionality of histories}.

For $k\in\prz{0}{n}$, we say that a run $R$ is \emph{$k$-upper} if $\hist(R,\topp^k(R(|R|)))=\topp^k(R(0))$;
let $\up^k$ be the set of all such runs.
%For convenience let $\up^{-1}$ contain runs of length $0$; 
%observe that $\up^n$ contains all runs. 
Intuitively, a run $R$ is $k$-upper when %the topmost $k$-stack of $R(|R|)$ ``comes from'' the topmost $k$-stack of $R(0)$, in the sense that 
the topmost $k$-stack of $R(|R|)$ is a copy of the topmost $k$-stack of $R(0)$, but possibly some changes were made to it.
Notice that $\up^n$ contains all runs, $\up^k\subseteq \up^l$ for $k\leq l$, 
and for a run $R\circ S$ with $S\in \up^k$ it holds $R\in \up^k\iff R\circ S\in \up^k$ (the last property is by compositionality of histories).

For $k\in\prz{1}{n}$, a run $R$ is a \emph{$k$-return} if
\begin{itemize}
\item	$\hist(R,\topp^{k-1}(R(|R|)))=\topp^{k-1}(\pop^k(R(0)))$, and
\item	$\subrun{R}{i}{|R|}\not\in \up^{k-1}$ for all $i\in\prz{0}{|R|-1}$.
\end{itemize}
Let $\ret^k$ be the set of $k$-returns.
Observe that $\ret^k\subseteq \up^k$.
Intuitively, $R$ is a $k$-return when the topmost $k$-stack of $R(|R|)$ is obtained from the topmost $k$-stack of $R(0)$ by removing its topmost $(k-1)$-stack 
(but not only in the sense of contents, but we require that really it was obtained this way).

\begin{exa}
	Consider a $2$-DPDA, and its run $R$ of length $6$ in which $\posl(\pi_2(R(0)))=[[a,b],[c,d]]$, and in which the operations between consecutive configurations are
	\begin{align*}
		\push^2_e\,,\ \pop^1\,,\ \pop^2\,,\ \pop^1\,,\ \push^1_d\,,\ \pop^1\,.
	\end{align*}
	Recall that our definition is that a $\push$ of any order can change the topmost stack symbol.
	The contents of the stacks of the configurations in the run, and subruns being $k$-upper runs and $k$-returns are presented in Table~\ref{tab:example}.
	\begin{table}
		\caption{Stack contents of the example run, and subruns being $k$-upper runs and $k$-returns}
		\label{tab:example}
		\vspace*{-3ex}
		\begin{align*}
			\begin{array}{c|l|l|l|l|l}
				j
				 &\posl(\pi_2(R(j)))
				 &i\mid\subrun{R}{i}{j}\in \up^0
				 &i\mid\subrun{R}{i}{j}\in \up^1
				 &i\mid\subrun{R}{i}{j}\in \ret^1
				 &i\mid\subrun{R}{i}{j}\in \ret^2\\
				\hline
				0&[[a,b],[c,d]]&0&0&-&-\\
				1&[[a,b],[c,d],[c,e]]&0,1&0,1&-&-\\
				2&[[a,b],[c,d],[c]]&2&0,1,2&0,1&-\\
				3&[[a,b],[c,d]]&0,3&0,3&-&1,2\\
				4&[[a,b],[c]]&4&0,3,4&0,3&-\\
				5&[[a,b],[c,d]]&4,5&0,3,4,5&-&-\\
				6&[[a,b],[c]]&4,6&0,3,4,5,6&5&-
			\end{array}
		\end{align*}
	\end{table}
	Notice that $R$ is not a $1$-return.
	We have $\hist(\subrun{R}{0}{5}, (d,(2,2)))=(c,(2,1))$.
\end{exa}

\subsection{Basic Properties of Runs}

We now state several easy propositions, which are useful later, and also give more intuition about the above definitions.

\begin{prop}\lab{prop:pre-k}
	Let $R$ be a $k$-upper run (where $k\in\prz{0}{n}$) such that $\subrun{R}{i}{|R|}\not\in \up^k$ for each $i\in\prz{1}{|R|-1}$.
	Then either
	\begin{itemize}
	\item $\tops^k(R(0))\cong\tops^k(R(|R|))$; 
		additionally for every $0$-stack $s^0$ in $\tops^k(R(|R|))$, 
		$\hist(R,s^0)$ is the corresponding $0$-stack in $\tops^k(R(0))$, or
	\item $|R|=1$ and the only transition of $R$ performs $\pop^r$ for $r\leq k$, or $\push^r_\gamma$ for $r\leq k$.
	\end{itemize}
\end{prop}

\proof
	For $|R|\leq 1$ we immediately fall into one of the possibilities.
	Otherwise, we look at the history of the topmost $k$-stack of $R(|R|)$. 
	It is covered by the first operation of $R$, and then it is not the topmost $k$-stack until $R(|R|)$.
	Thus, it remains unchanged (we have the first possibility).
\qed

Next, we give four propositions about $k$-upper runs and $k$-returns.

\begin{prop}\lab{prop:size2pre}
	Let $R$ be a $k$-upper run, where $k\in\prz{1}{n}$.
	Then $R$ is $(k-1)$-upper if and only if $|\tops^k(R(0))|\leq\tops^k(R(i))|$ for each $i\in\prz{0}{|R|}$ such that $\subrun{R}{i}{|R|}\in\up^k$.
\end{prop}

\begin{prop}\lab{prop:pre-jest-fajne}
	Let $S\circ T$ be a $(k-1)$-upper run in which $T$ is $k$-upper, where $k\in\prz{1}{n}$.
	Then $S$ is $(k-1)$-upper.
\end{prop}

\begin{prop}\lab{prop:pomp-bez-zmian}
	Let $R$ be a run that is not $(k-1)$-upper, where $k\in\prz{1}{n}$.
	Suppose that $\subrun{R}{0}{j}$ is $(k-1)$-upper for the greatest index $j\in\prz{0}{|R|-1}$ such that $\subrun{R}{j}{|R|}$ is $k$-upper (in particular such an index $j$ exists).
	Then $R$ is a $k$-return.
\end{prop}

\begin{prop}\lab{prop:return-jest-fajny}
	Let $R$ be a $k$-return, where $k\in\prz{1}{n}$.
	Then $\pop^k(\tops^k(R(0)))\cong\tops^k(R(|R|))$.
	Additionally for every $0$-stack $s^0$ in $\tops^k(R(|R|))$, $\hist(R,s^0)$ 
	is the corresponding $0$-stack in $\pop^k(\tops^k(R(0)))$.
\end{prop}

\proof[Proof of Propositions~\ref{prop:size2pre}-\ref{prop:return-jest-fajny}]
	Recall that a $(k-1)$-upper run, a composition of two $k$-upper runs, and a $k$-return are special cases of $k$-upper runs.
	Thus in all four propositions we have a $k$-upper run $R$, where $k\in\prz{1}{n}$ (where for Proposition~\ref{prop:pre-jest-fajne} we take $R=S\circ T$).
	Let $X$ denote the set of those indices $i\in\prz{0}{|R|}$ for which $\subrun{R}{i}{|R|}$ is $k$-upper.
	Notice that $0\in X$ and $|R|\in X$.
	For $i\in X$, let $r_i=|\tops^k(R(i))|$, and let $s_i(r)$ be the $r$-th $(k-1)$-stack 
	in $\tops^k(R(i))$ (for $r\in\prz{1}{r_i}$).
	We claim that for all $b,e\in X$ with $b\leq e$, and for each $r\in\prz{1}{r_{e}}$,
	\begin{align}
		\hist(\subrun{R}{b}{e},s_e(r))=s_b\big(\min(\{r\}\cup\{r_l\mid l\in X\land b\leq l<e\})\big)\,.\label{eq:hist}
	\end{align}
	\begin{figure*}
		\begin{center}
			\includegraphics[scale=0.7]{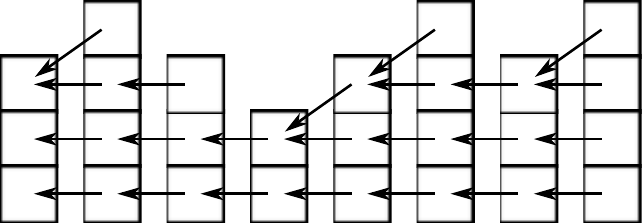}
			\hspace{.8cm}
			\includegraphics[scale=0.7]{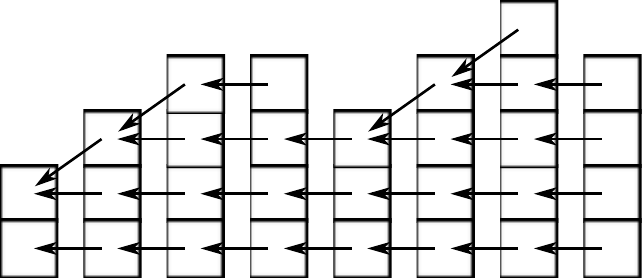}
			\hspace{.8cm}
			\includegraphics[scale=0.7]{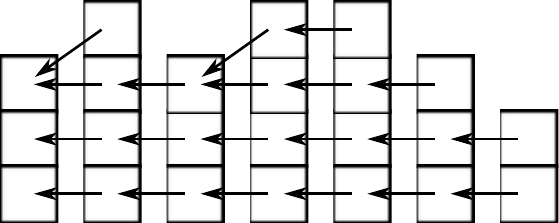}
		\end{center}
		\caption{Illustrations to Equality~\eqref{eq:hist} and Propositions~\ref{prop:size2pre}-\ref{prop:return-jest-fajny}.
			Particular columns represent a $k$-stack in consecutive configurations of a run.
			Arrows show the value of the $\hist$ function.
			The run on the first diagram is not $(k-1)$-upper.
			The run on the second diagram is $(k-1)$-upper; 
			notice that every its prefix is $(k-1)$-upper, assuming that (in the final configuration of the prefix) the illustrated $k$-stack is the topmost one (cf.~Proposition~\ref{prop:pre-jest-fajne}).
			The run on the last diagram is a $k$-return.}
		\label{fig:eq-6-1}
	\end{figure*}%
	See Figure~\ref{fig:eq-6-1} for an illustration.
	We prove Equality~\eqref{eq:hist} by induction on $e-b$.
	For $e=b$ it is true.
	For the induction step consider the smallest $e'\in X$ that is greater than $e$.
	Notice that for all $l\in[e+1,e'-1]$ necessarily $\subrun{R}{l}{e'}\not\in\up^k$ (since $\subrun{R}{l}{e'}\in\up^k$ implies that $\subrun{R}{l}{|R|}\in\up^k$, that is, that $l\in X$), 
	so the subrun $\subrun{R}{e}{e'}$ is in one of the forms described by Proposition~\ref{prop:pre-k}.
	For both of them we see that for each $r\in\prz{1}{r_{e'}}$,
	\begin{align*}
		\hist(\subrun{R}{e}{e'},s_{e'}(r))=s_e(\min\{r,r_{e}\})\,.
	\end{align*}
	Together with the induction assumption for $b,e$, this implies Equality~\eqref{eq:hist} for $b,e'$.

	We also claim that for all $b,e\in X$ with $b\leq e$,
	\begin{align}
		\subrun{R}{b}{e}\in\up^{k-1} \quad\Leftrightarrow\quad r_b\leq r_i\mbox{ for each }i\in X\cap[b,e]\,.\label{eq:hist2}
	\end{align}
	Indeed, $\subrun{R}{b}{e}$ is $(k-1)$-upper 
	if and only if $\hist(\subrun{R}{b}{e},s_{e}(r_{e}))=s_{b}(r_{b})$,
	and, as we see from Equality~\eqref{eq:hist}, the latter holds if and only if $r_{b}\leq r_i$ for each $i\in X\cap[b,e]$.
	
	Proposition~\ref{prop:size2pre} follows directly from Equivalence~\eqref{eq:hist2} used with $b=0$ and $e=|R|$.

	In order to prove Proposition~\ref{prop:pre-jest-fajne}, we suppose that $R=S\circ T$, that $R$ is $(k-1)$-upper, and that $T$ is $k$-upper.
	Using Equivalence~\eqref{eq:hist2} with $b=0$ and $e=|R|$, we obtain that $r_0\leq r_i$ for each $i\in X$.
	Since $T$ is $k$-upper, $|S|\in X$.
	Thus, we can use Equivalence~\eqref{eq:hist2} with $b=0$ and $e=|S|$; it tells us that $S$ is $(k-1)$-upper, as required.

	Consider now a situation as in Proposition~\ref{prop:pomp-bez-zmian}, namely,
	let $R\not\in\up^{k-1}$, let $j=\max(X\cap\prz{0}{|R|-1})$, and let $\subrun{R}{0}{j}\in\up^{k-1}$.
	If $j<|R|-1$, then from Proposition~\ref{prop:pre-k} applied to $\subrun{R}{j}{|R|}$ we obtain that $\hist(\subrun{R}{j}{|R|},\topp^{k-1}(R(|R|)))=\topp^{k-1}(R(j))$, that is,
	that $\subrun{R}{j}{|R|}$ is $(k-1)$-upper.
	Since a composition of $(k-1)$-upper runs is $(k-1)$-upper, this contradicts the assumptions that $\subrun{R}{0}{j}$ is $(k-1)$-upper but $R$ is not.
	Thus, $j=|R|-1$.
	By Equivalence~\eqref{eq:hist2}, the assumptions $\subrun{R}{0}{|R|-1}\in \up^{k-1}$ and $R\not\in\up^{k-1}$ imply that $r_0\leq r_i$ for each $i\in X\setminus\{|R|\}$, 
	but not for $i=|R|$.
	It follows that $r_0=r_{|R|-1}=r_{|R|}+1$, since $|r_{|R|-1}-r_{|R|}|\leq 1$.
	From Equality~\eqref{eq:hist} we deduce that $\hist(R,\topp^{k-1}(R(|R|)))=\topp^{k-1}(\pop^k(R(0)))$, 
	and from Equivalence~\eqref{eq:hist2} that $\subrun{R}{i}{|R|}\not\in\up^{k-1}$ for $i\in X\cap[0,|R|-1]$.
	For $i\not\in X$, we also have that $\subrun{R}{i}{|R|}\not\in\up^k\supseteq\up^{k-1}$ by definition of $X$.
	Thus $R$ is a $k$-return.

	Finally, suppose that $R$ is a $k$-return.
	By definition, this implies that $\subrun{R}{i}{|R|}$ is not $(k-1)$-upper for any $i\in X\setminus\{|R|\}$, 
	so, by Equivalence~\eqref{eq:hist2}, for every $i\in X\setminus\{|R|\}$ there is some $j\in X$ such that $j>i$ and $r_i>r_j$.
	By transitivity, it actually holds that $r_i>r_{|R|}$ for each $i\in X\setminus\{|R|\}$.
	Thus, Equality~\eqref{eq:hist} implies that $\hist(R,s_{|R|}(r))=s_0(r)$ for each $r\in\prz{1}{r_{|R|}}$;
	moreover, $\hist(\subrun{R}{i}{|R|},s_{|R|}(r))\neq\topp^{k-1}(R(i))$ 
	for all $i<|R|$, which implies that $s_{|R|}(r)$ is an unmodified copy of $s_0(r)$ (a $(k-1)$-stack can be modified only when it becomes the topmost $(k-1)$-stack).
	It means that $\tops^k(R(|R|))$ consists of the $r_{|R|}$ bottommost $(k-1)$-stacks of $\tops^k(R(0))$,
	also in the sense of the history function.
	By the definition of a $k$-return, $\hist(R,s_{|R|}(r_{|R|}))=s_0(r_0-1)$, so $r_{|R|}=r_0-1$.
\qed

We now give one more proposition, which is an immediate consequence of Proposition~\ref{prop:return-jest-fajny}.

\begin{prop}\lab{prop:push-return-jest-fajny}
	Let $R$ be a run such that its first transition performs $\push^k_\gamma$, and $\subrun{R}{1}{|R|}$ is a $k$-return, where $k\in\prz{1}{n}$.
	Then $\tops^k(R(0))\cong\tops^k(R(|R|))$.
	Additionally, for every $0$-stack $s^0$ in $\tops^k(R(|R|))$, 
	$\hist(R,s^0)$ is the corresponding $0$-stack in $\tops^k(R(0))$.
\qed\end{prop}

\subsection{Characterization of Returns and Upper Runs}

Next we give two propositions, which describe possible forms of upper runs and returns.

\begin{prop}\lab{prop:upper}
	A run $R$ is $k$-upper (where $k\in\prz{0}{n}$) if and only if
	\begin{enumerate}
	\item	$|R|=0$, or
	\item	$|R|=1$, and the only transition of $R$ is $\read$, or it performs $\push^r_\gamma$ for any $r$, or $\pop^r$ for $r\leq k$, or
	\item	the first transition of $R$ performs $\push^r_\gamma$ for $r\geq k+1$, and $\subrun{R}{1}{|R|}$ is an $r$-return, or
	\item	$R$ is a composition of two nonempty $k$-upper runs.
	\end{enumerate}
\end{prop}

\newcommand{\spl}{x}

\begin{proof}
	The right-to-left implication is almost immediate; in Case (3) we use Proposition~\ref{prop:push-return-jest-fajny}.
	
	Concentrate on the left-to-right implication.
	If $|R|=0$, then we have Case (1).
	Suppose that $|R|\geq 1$.
	Notice that the first transition, between $R(0)$ and $R(1)$, cannot perform $\pop^r$ for $r\geq k+1$, as such an operation removes the 
	topmost $k$-stack of $R(0)$, which contradicts the assumption that $R$ is $k$-upper.
	Thus, if $|R|=1$, then we have Case (2).
	Suppose that $|R|\geq 2$.
	If the first transition is $\read$, or performs $\pop^r$ for $r\leq k$, or $\push^r_\gamma$ for $r\leq k$,
	then both $\subrun{R}{0}{1}$ and $\subrun{R}{1}{|R|}$ are $k$-upper; we have Case (4).
	We can do the same when the operation is $\push^r_\gamma$ for $r\geq k+1$ and $\subrun{R}{1}{|R|}$ is $k$-upper.

	The remaining case is that the first operation is $\push^r_\gamma$ for $r\geq k+1$ and $\subrun{R}{1}{|R|}$ is not $k$-upper.
	Notice that $\hist(\subrun{R}{0}{1},s^k)=\topp^k(R(0))$ holds only for two $k$-stack of $R(1)$: for $s^k=\topp^k(R(1))$ and for $s^k=\topp^k(\pop^r(R(1)))$.
	So, because $R$ is $k$-upper and $\subrun{R}{1}{|R|}$ is not $k$-upper, 
	which by definition means that $\hist(R,\topp^k(R(|R|)))=\topp^k(R(0))$ and $\hist(\subrun{R}{1}{|R|},\topp^k(R(|R|)))\neq\topp^k(R(1))$, 
	it has to be $\hist(\subrun{R}{1}{|R|},\topp^k(R(|R|)))=\topp^k(\pop^r(R(1)))$.
	Thus, also $\hist(\subrun{R}{1}{|R|},\topp^{r-1}(R(|R|)))=\topp^{r-1}(\pop^r(R(1)))$.
	Let $\spl$ be the smallest positive index for which $\subrun{R}{\spl}{|R|}$ is $(r-1)$-upper.
	Then $\hist(\subrun{R}{1}{\spl},\topp^{r-1}(R(\spl)))=\topp^{r-1}(\pop^r(R(1)))$
	(by compositionality of histories, because $\hist(\subrun{R}{\spl}{|R|},\topp^{r-1}(R(|R|)))=\topp^{r-1}(R(\spl))$ and $\hist(\subrun{R}{1}{|R|},\topp^{r-1}(R(|R|)))=\topp^{r-1}(\pop^r(R(1)))$), 
	and there is no $i\in\prz{1}{\spl-1}$ such that $\subrun{R}{i}{\spl}$ is $(r-1)$-upper
	(because $\subrun{R}{i}{\spl}\in\up^{r-1}$ and $\subrun{R}{\spl}{|R|}\in\up^{r-1}$ would imply that $\subrun{R}{i}{|R|}\in\up^{r-1}$).
	Thus, $\subrun{R}{1}{\spl}$ is an $r$-return.
	The knowledge at this point of the proof is summarized in Figure~\ref{fig:prop-6-7}.
 
	\begin{figure*}
		\begin{center}
			\begin{minipage}{.45\textwidth}
				\begin{center}
					\import{pics/}{prop-6-7.pdf_tex_ok}
					\caption{Illustration for the proof of Proposition~\ref{prop:upper}}
					\label{fig:prop-6-7}
				\end{center}
			\end{minipage}\hspace{0.08\textwidth}%
			\begin{minipage}{.45\textwidth}
				\begin{center}
					\import{pics/}{prop-6-8.pdf_tex_ok}
					\caption{Illustration for the proof of Proposition~\ref{prop:return}}
					\label{fig:prop-6-8}
				\end{center}
			\end{minipage}
		\end{center}
	\end{figure*}
	
	If $\spl=|R|$, then we have Case (3).
	For the remaining part of the proof suppose that $\spl<|R|$.
	Let $s^k=\hist(\subrun{R}{\spl}{|R|},\topp^k(R(|R|)))$.
	Because $\subrun{R}{\spl}{|R|}\in\up^{r-1}$ and $\topp^k(R(|R|))$ is in $\topp^{r-1}(R(|R|))$ (recall that $k\leq r-1$), we have that $s^k$ is in $\tops^{r-1}(R(\spl))$.
	On the other hand, because $R\in\up^k$, by compositionality of histories we know that $\hist(\subrun{R}{0}{\spl},s^k)=\topp^k(R(0))$.
	Proposition~\ref{prop:push-return-jest-fajny} applied to $\subrun{R}{0}{\spl}$ (its first operation is $\push^r_\gamma$, and $\subrun{R}{1}{\spl}$ is an $r$-return)
	implies that $s^k=\topp^k(R(\spl))$, that is, that $\subrun{R}{0}{|R|}$ and $\subrun{R}{\spl}{|R|}$ are $k$-upper.
	Thus, we have Case (4).
\end{proof}

\begin{prop}\lab{prop:return}
	A run $R$ is an $r$-return (where $r\in\prz{1}{n}$) if and only if
	\begin{enumerate}
	\item	$|R|=1$, and the only transition of $R$ performs $\pop^r$, or
	\item	the first transition of $R$ is $\read$, or it performs $\pop^k$ for $k<r$, or $\push^k_\gamma$ for $k\neq r$, and $\subrun{R}{1}{|R|}$ is an $r$-return, or
	\item	the first transition of $R$ performs $\push^k_\gamma$ for $k\geq r$, and $\subrun{R}{1}{|R|}$ is a composition of a $k$-return and an $r$-return.
	\end{enumerate}
\end{prop}

\begin{proof}
	Let us analyze the right-to-left implication, which is easier.
	Case (1) is trivial.
	In Case (2) we observe that
	$\hist(\subrun{R}{0}{1},\topp^{r-1}(\pop^r(R(1))))=\topp^{r-1}(\pop^r(R(0)))$ 
	(it is important that $k\neq r$ in the case of $\push^k_\gamma$), and hence 
	\begin{align*}
		\hist(R,\topp^{r-1}(R(|R|)))&=\hist(\subrun{R}{0}{1},\hist(\subrun{R}{1}{|R|},\topp^{r-1}(R(|R|))))\\
			&=\hist(\subrun{R}{0}{1},\topp^{r-1}(\pop^r(R(1))))=\topp^{r-1}(\pop^r(R(0)))\,.
	\end{align*}
	In particular, this implies that $R$ is not $(r-1)$-upper;
	moreover $\subrun{R}{i}{|R|}$ is not $(r-1)$-upper for $i\in[1,|R|-1]$ because $\subrun{R}{1}{|R|}$ is an $r$-return.
	Thus, $R$ is an $r$-return.
	In Case (3), let $\spl-1$ be the length of the first return (so the $k$-return ends in $R(\spl)$).
	The situation is depicted in Figure~\ref{fig:prop-6-8}.
	Recall that $k\geq r$.
	By Proposition~\ref{prop:push-return-jest-fajny},
	$\hist(\subrun{R}{0}{\spl},\topp^{r-1}(\pop^r(R(\spl))))=\topp^{r-1}(\pop^r(R(0)))$.
	Since $\hist(\subrun{R}{\spl}{|R|},\topp^{r-1}(R(|R|)))=\topp^{r-1}(\pop^r(R(\spl)))$, we conclude that $\hist(R,\topp^{r-1}(R(|R|)))=\topp^{r-1}(\pop^r(R(0)))$.
	This in particular implies that $R$ is not $(r-1)$-upper.
	Because $\subrun{R}{\spl}{|R|}$ is an $r$-return, $\subrun{R}{i}{|R|}$ cannot be $(r-1)$-upper for $i\in\prz{\spl}{|R|-1}$.
	If $\subrun{R}{i}{|R|}$ was $(r-1)$-upper for some $i\in\prz{1}{\spl-1}$, 
	then $\hist(\subrun{R}{i}{\spl},\topp^{r-1}(\pop^r(R(\spl))))=\topp^{r-1}(R(i))$.
	This would imply that $\subrun{R}{i}{\spl}$ is $(k-1)$-upper (both for $k>r$ and $k=r$), 
	which is impossible, because $\subrun{R}{1}{\spl}$ is a $k$-return.
	We conclude that $R$ is an $r$-return.

	Concentrate now on the left-to-right implication.
	Before starting the proof, notice that in order to prove that $\subrun{R}{\spl}{|R|}\in\ret^r$ for some $\spl\in\prz{0}{|R|}$, 
	it is enough to check that $\hist(\subrun{R}{\spl}{|R|},\topp^{r-1}(R(|R|)))=\topp^{r-1}(\pop^r(R(\spl)))$:
	the condition that $\subrun{R}{i}{|R|}\not\in\up^{r-1}$ for all $i\in[\spl,|R|-1]$ is ensured by the fact that $R$ itself is an $r$-return.
	
	Of course $|R|\geq 1$.
	Because $R$ is an $r$-return, 
	\begin{align}
		\hist(\subrun{R}{0}{1},\hist(\subrun{R}{1}{|R|},\topp^{r-1}(R(|R|))))&=\hist(R,\topp^{r-1}(R(|R|)))\nonumber\\
			&=\topp^{r-1}(\pop^r(R(0)))\,.\label{eq:6-8}
	\end{align}
	
	Observe that the first operation, between $R(0)$ and $R(1)$, cannot be $\pop^k$ for $k\geq r+1$, as after such an operation 
	there would be no $(r-1)$-stack $s^{r-1}$ of $R(1)$ such that $\hist(\subrun{R}{0}{1},s^{r-1})=\topp^{r-1}(\pop^r(R(0)))$, 
	which contradicts Equality~\eqref{eq:6-8}.
	
	Suppose that the first operation of $R$ is $\pop^r$.
	In this situation, the only $(r-1)$-stack $s^{r-1}$ of $R(1)$ such that $\hist(\subrun{R}{0}{1},s^{r-1})=\topp^{r-1}(\pop^r(R(0)))$ is 
	$s^{r-1}=\topp^{r-1}(R(1))$, 
	and thus we have $\hist(\subrun{R}{1}{|R|},\topp^{r-1}(R(|R|)))=\topp^{r-1}(R(1))$ by Equality~\eqref{eq:6-8}. This means that $\subrun{R}{1}{|R|}$ is $(r-1)$-upper.
	A nonempty suffix of an $r$-return cannot be $(r-1)$-upper, so $|R|=1$; we have Case (1).

	Next, suppose that the first operation is $\read$, or $\pop^k$ for $k\leq r-1$, or $\push^k_\gamma$ for $k\leq r-1$.
	In this situation, the only $(r-1)$-stack $s^{r-1}$ of $R(1)$ such that $\hist(\subrun{R}{0}{1},s^{r-1})=\topp^{r-1}(\pop^r(R(0)))$ 
	is $s^{r-1}=\topp^{r-1}(\pop^r(R(1)))$, 
	and thus $\hist(\subrun{R}{1}{|R|},\topp^{r-1}(R(|R|)))=\topp^{r-1}(\pop^r(R(1)))$, by Equality~\eqref{eq:6-8}. 
	In consequence, $\subrun{R}{1}{|R|}$ is an $r$-return;
	we have Case (2).	
	
	Finally, suppose that the first operation of $R$ is $\push^k_\gamma$ for $k\geq r$.
	If $k>r$, then there are two $(r-1)$-stacks $s^{r-1}$ of $R(1)$ such that $\hist(\subrun{R}{0}{1},s^{r-1})=\topp^{r-1}(\pop^r(R(0)))$, 
	namely $s^{r-1}=\topp^{r-1}(\pop^r(R(1)))$ and $s^{r-1}=\topp^{r-1}(\pop^r(\pop^k(R(1))))$.
	If $k=r$, only the latter possibility remains: $s^{r-1}=\topp^{r-1}(\pop^r(\pop^k(R(1))))$.
	By Equality~\eqref{eq:6-8}, $\hist(\subrun{R}{1}{|R|},\topp^{r-1}(R(|R|)))$ has to be one of these two $(r-1)$-stacks.

	Suppose first that $\hist(\subrun{R}{1}{|R|},\topp^{r-1}(R(|R|)))=\topp^{r-1}(\pop^r(R(1)))$ and $k>r$.
	Then $\subrun{R}{1}{|R|}$ is an $r$-return;
	we have Case (2). 

	The opposite possibility is that $\hist(\subrun{R}{1}{|R|},\topp^{r-1}(R(|R|)))=\topp^{r-1}(\pop^r(\pop^k(R(1))))$.
	Because $\topp^{r-1}(R(|R|))$ and $\topp^{r-1}(\pop^r(\pop^k(R(1))))$ are in $\topp^k(R(|R|))$ and $\topp^k(R(1))$, respectively (recall that $k\geq r$),
	this implies that $\subrun{R}{1}{|R|}\in\up^k$.
	Let $\spl$ be the smallest positive index such that $\subrun{R}{1}{\spl}\not\in\up^{k-1}$ and $\subrun{R}{\spl}{|R|}\in\up^k$
	(it exists: in the worst case we can take $\spl=|R|$, since $\subrun{R}{1}{|R|}\not\in\up^{k-1}$ and $\subrun{R}{|R|}{|R|}\in\up^k$).
	Because $\subrun{R}{1}{|R|}$ and $\subrun{R}{\spl}{|R|}$ are $k$-upper, also $\subrun{R}{1}{\spl}$ is $k$-upper.
	Moreover, because $\subrun{R}{1}{\spl}\not\in\up^{k-1}$ and $\subrun{R}{1}{1}\in\up^{k-1}$, necessarily $\spl>1$.
	Let also $j$ be the greatest index in $\prz{1}{\spl-1}$ such that $\subrun{R}{j}{\spl}\in\up^k$ (it exists, because $\subrun{R}{1}{\spl}\in\up^k$ and $1\in\prz{1}{\spl-1}$).
	Then $\subrun{R}{j}{|R|}$ (a composition of two $k$-upper runs) is $k$-upper, and thus $\subrun{R}{1}{j}$ is not $(k-1)$-upper, by minimality of $\spl$.
	In such a situation, Proposition~\ref{prop:pomp-bez-zmian} implies that $\subrun{R}{1}{\spl}$ is a $k$-return.
	Let $t^{r-1}=\hist(\subrun{R}{\spl}{|R|},\topp^{r-1}(R(|R|)))$.
	Because $\subrun{R}{\spl}{|R|}\in\up^k$ and $\topp^{r-1}(R(|R|))$ is in $\topp^k(R(|R|))$ (since $k\geq r$), we have that $t^{r-1}$ is in $\tops^k(R(\spl))$.
	On the other hand, because $R\in\ret^r$, by compositionality of histories we know that $\hist(\subrun{R}{0}{\spl},t^{r-1})=\topp^{r-1}(\pop^r(R(0)))$.
	Proposition~\ref{prop:push-return-jest-fajny} applied to $\subrun{R}{0}{\spl}$ (its first operation is $\push^k_\gamma$, and $\subrun{R}{1}{\spl}$ is a $k$-return)
	implies that $t^{r-1}=\topp^{r-1}(\pop^r(R(\spl)))$, that is, that $\subrun{R}{\spl}{|R|}\in\ret^r$; we have Case (3).
\end{proof}

\section{Types and Sequence Equivalence}\lab{sec:types}

In this section we assign to each configuration a type from a finite set.
The slogan is that configurations with the same positionless topmost $k$-stacks and the same type are starting points of similar $k$-upper runs.
We start by an example.

\begin{exa}\lab{ex:types}
	Consider a $3$-DPDA that (while being in some state) can perform the following $1$-upper run:
	it executes $\pop^1$, $\push^3$, and then it starts analyzing the topmost $2$-stack using $\pop^1$ and $\pop^2$;
	when a $0$-stack containing a fixed stack symbol $a$ is found, the automaton performs $\pop^3$; the run ends in the same state as it begins.
	As an effect of this run, only the topmost $0$-stack is removed, so this is indeed a $1$-upper run.
	Notice that it can be executed only when the topmost $2$-stack contains the $a$ symbol, and can be repeated as long as the topmost $1$-stack is nonempty.
	Consider now two configuration of this $3$-DPDA, having the same positionless topmost $1$-stack.
	If additionally the topmost $2$-stacks of both configurations contain the $a$ symbol, 
	then from each of them we can start the $1$-upper run described above, and repeat it the same number of times.
	
	Because a $1$-upper run can arbitrarily modify the topmost $1$-stack, we consider configurations having the same positionless topmost $1$-stack.
	On the other hand, we summarize the rest of the stack in a small piece of information, called a type.
	In this example we only need to know whether there is the $a$ symbol in the topmost $2$-stack (below the topmost $1$-stack).
	In general, whenever a $3$-DPDA removes the topmost $1$-stack and starts analyzing the stack below, 
	next it has to remove the whole topmost $2$-stack (since we consider a $1$-upper run).
	Thus, for each entering state (i.e.,~the state when removing the topmost $1$-stack) we only need to know the exit state (i.e.,~the state when removing the topmost $2$-stack).
	For higher orders the situation is slightly more complicated, but similar.
\end{exa}

There is also a second goal of this section.
Suppose that we have a sequence of configurations, all having the same positionless topmost $k$-stack and the same type.
Then, as said above, from each of them we can execute a similar $k$-upper run.
Typically, these $k$-upper runs are prefixes of some accepting runs.
We want to determine whether such accepting runs can read an unbounded number of $\sharp$ symbols, or not.
(For technical reasons, we consider $n$-returns instead of accepting runs.)

For this section we fix an $n$-DPDA $\calA$ with stack alphabet $\Gamma$, state set $Q$, and input alphabet $A$ that contains a distinguished symbol $\sharp$.
Moreover, we fix a morphism $\phi\colon A^*\to M$ into a finite monoid $M$.
For a run $R$ reading a word $w$, by $\phi(R)$ we denote $\phi(w)$, and by $\sharp(R)$ we denote the number of sharps in $w$.
The goal of the morphism is to describe when two upper runs read a similar word:
we want to distinguish input words evaluating to different elements of $M$.

Recall that when both $R\circ S$ and $S$ are $k$-upper runs, then $R$ is $k$-upper as well.
It follows that any nonempty $k$-upper run $R$ can be uniquely represented as a composition of the maximal number of nonempty $k$-upper runs $R_1\circ\dots\circ R_r$:
we keep on cutting off minimal suffixes that are $k$-upper (notice that infixes or even prefixes of $R_i$ can be $k$-upper, but suffixes are not).
We compare $k$-upper runs using the following definition of being $(k,\phi)$-parallel.

\begin{defi}
	Let $R=R_1\circ\dots\circ R_r$ and $S=S_1\circ\dots\circ S_s$ be $k$-upper runs decomposed into the maximal number of nonempty $k$-upper runs.
	We say that $R$ and $S$ are \emph{$(k,\phi)$-parallel}
	when $r=s$, 
	and for each $i\in\prz{1}{r}$ it holds that $\phi(R_i)=\phi(S_i)$ and $\tops^k(R_i(0))\cong\tops^k(S_i(0))$, 
	as well as $\tops^k(R(|R|))\cong\tops^k(S(|S|))$.
	In particular, two runs $R$, $S$ of length $0$ are $(k,\phi)$-parallel when $\tops^k(R(0))\cong\tops^k(S(0))$.
	When saying that two runs are $(k,\phi)$-parallel we implicitly mean that they are $k$-upper.
\end{defi}

We claim that if runs $R$ and $S$ are $(k,\phi)$-parallel, and $R$ is divided in any way into $k$-upper runs $R=R_1'\circ\cdots\circ R_m'$,
then $S$ can be as well divided into $k$-upper runs $S=S_1'\circ\cdots\circ S_m'$ such that 
for each $i\in\prz{1}{m}$ it holds that $\phi(R_i')=\phi(S_i')$ and $\tops^k(R_i'(0))\cong\tops^k(S_i'(0))$, 
as well as $\tops^k(R(|R_{}|))\cong\tops^k(S(|S|))$.
Indeed, on the one hand, each nonempty $R_i'$ can be further subdivided into $k$-upper runs of the finest decomposition.
On the other hand, for each empty $R_i'$ we can insert an empty $S_i'$ into the sequence for $S$.

As already mentioned, to each configuration $c$ we assign its \emph{$(\calA,\phi)$-type} (simply called \emph{type} when $\calA$ and $\phi$ are fixed), which comes from a finite set.
Before giving a definition, we state two theorems, which describe required properties of our types.

\begin{thm}\lab{thm:types}
	Let $R$ be a $k$-upper run, where $k\in\prz{0}{n}$, 
	and let $c$ be a configuration having the same $(\calA,\phi)$-type and the same positionless topmost $k$-stack as $R(0)$. 
	Then from $c$ we can start a run that is $(k,\phi)$-parallel to $R$.
\end{thm}

In addition to types, we also define an equivalence relation over infinite sequences of configurations of $\calA$, called $(\calA,\phi)$-sequence-equivalence, 
which has finitely many equivalence classes.
The goal is to specify whether the number of $\sharp$ symbols read by a run constructed in Theorem~\ref{thm:types} is big or small.
However, instead of having ``big'' and ``small'' numbers, we say whether their sequence is bounded or unbounded.
This is made precise in the following theorem.

\begin{thm}\lab{thm:stypes}
	Let $R\circ R'$ be a run in which $R$ is $k$-upper and $R'$ is an $n$-return, where $k\in\prz{0}{n}$.
	Let $c_1,c_2,\dots$ and $d_1,d_2,\dots$ be infinite sequences of configurations that are $(\calA,\phi)$-sequence-equivalent, and
	in which all configurations have the same $(\calA,\phi)$-type and the same positionless topmost $k$-stack as $R(0)$.
	Then for each $i$ there exist runs $S_i\circ S_i'$ from $c_i$, and $T_i\circ T_i'$ from $d_i$ in which $S_i$ and $T_i$ are $(k,\phi)$-parallel to $R$, 
	and $S_i'$ and $T_i'$ are $n$-returns such that $\phi(S_i')=\phi(T_i')=\phi(R')$, and such that
	the sequences $\sharp(S_1\circ S_1'),\sharp(S_2\circ S_2'),\dots$ and $\sharp(T_1\circ T_1'),\sharp(T_2\circ T_2'),\dots$ are either both bounded or both unbounded.
\end{thm}

Let us mention briefly how this theorem is used in Section~\ref{sec:ostatni}.
We consider there a configuration $c$ reached after reading a complicated word, containing some blocks of stars, separated by some brackets.
Using a pumping lemma developed in Section~\ref{sec:pumping}, we increase the number of stars read in one of such blocks, 
obtaining configurations $c_1, c_2, \dots$ at the end of the run (where consecutive configurations are reached after reading more and more stars in the considered block).
It is ensured that all $c_i$ have the same $(\calA,\phi)$-type, and the same positionless topmost $k$-stack.
Likewise, we increase the number of stars read in some other block of stars, obtaining configurations $d_1, d_2, \dots$.
Having more blocks of stars than classes of the $(\calA,\phi)$-sequence-equivalence relation 
we can ensure that the sequences $c_1,c_2,\dots$ and $d_1,d_2,\dots$ are $(\calA,\phi)$-sequence-equivalent, by the pigeonhole principle.
Theorem~\ref{thm:stypes} says that the two sequences of configurations cannot be distinguished by runs (of a specific form) starting in these configurations:
the runs contain sharps corresponding to stars either from both considered blocks of stars
(and then both sequences  $\sharp(S_1\circ S_1'),\sharp(S_2\circ S_2'),\dots$ and $\sharp(T_1\circ T_1'),\sharp(T_2\circ T_2'),\dots$ are unbounded)
or from none of them (and then both these sequences are bounded).

The $n$-returns in Theorem~\ref{thm:stypes} should be understood as accepting runs.
Indeed, in Section~\ref{sec:ostatni} we increase by $1$ the order of an arbitrary $(n-1)$-DPDA, and we add a $\pop^n$ operation just before reaching an accepting state;
after such a modification, a run is accepting if and only if it is an $n$-return.
This trick is performed only for uniformity of presentation: instead of considering accepting runs as a separate concept, we see them as a special case of returns (and returns are used anyway).

One may be puzzled by the fact that Theorem~\ref{thm:types} talks about a $k$-upper run, while Theorem~\ref{thm:stypes} about a $k$-upper run composed with an $n$-return.
This difference is application-driven: the first theorem needs to be used without an $n$-return, while the second one with an $n$-return.
In fact Theorem~\ref{thm:types} is true also with an $n$-return at the end, and Theorem~\ref{thm:stypes} also without an $n$-return.

\begin{exa}
	Consider the $3$-DPDA and the $1$-upper run from Example~\ref{ex:types}, with the difference that now
	whenever a $b$ symbol is removed from the stack during the analysis of the topmost $2$-stack, the DPDA reads the $\sharp$ symbol from the input.
	Additionally, suppose that when the topmost $1$-stack becomes empty (a bottom-of-stack symbol is uncovered), the DPDA performs $\pop^3$; this $\pop^3$ serves as the $3$-return $R'$.
	Then basically we need two equivalence classes of sequences of configurations (recall that only for sequences with the same positionless topmost $1$-stack the relation is meaningful): 
	one where the number of $b$ symbols in the topmost $2$-stacks in the configurations is bounded, and one where this number is unbounded.
	Depending on this fact, the runs read either a bounded or an unbounded number of sharps.
	Of course in general we need more classes than just two (``bounded'' and ``unbounded''), 
	because, for example, another $1$-upper run (having a different image under $\phi$) might read one sharp per each $c$ symbol found on the stack (instead of the $b$ symbols).
\end{exa}

The rest of this section is devoted to defining types and sequence-equivalence, and proving Theorems~\ref{thm:types} and \ref{thm:stypes}.
This is independent from the rest of the paper.

\subsection{Definition of Types}

The types considered here are similar to stack automata of Broadbent, Carayol, Hague, and Serre~\cite{saturation}, 
as well as to intersection types of Kobayashi~\cite{Kobayashi09}.
Notice, however, that we extend them by a productive/nonproductive flag, which is not present there.
This flag is essential for our proof, since we want to estimate the number of $\sharp$ symbols read by a run, not just to determine existence of some kind of runs.
On the other hand, in the conference version of the current paper \cite{ho-new} we were using types that were directly describing returns (while here returns correspond to using an assumption);
these types were more complicated.

\subsubsection*{Run Descriptors}

\makebox[0cm]{\ }%    <--- przekolorowanie nagłówka, razem z kropką
We label stacks by \emph{run descriptors}.
To label a $k$-stack $s^k$, where $k\in\prz{0}{n}$, we can use a run descriptor from a set $\calT^k$.
The sets $\calT^k$ are defined inductively as follows:
\begin{align*}
	\calT^k &= Q\times\calP(M\times\calT^n)\times\calP(M\times\calT^{n-1})\times\dots\times\calP(M\times\calT^{k+1})\times\{\np,\pr\}\,,
\end{align*}
where $\calP(X)$ denotes the power set of $X$.
We use lowercase Greek letters ($\sigma,\tau,\dots$) to denote elements of $\calT^k$, uppercase Greek letters ($\Psi,\Phi,\dots$) to denote subsets of $M\times\calT^k$,
and uppercase Greek letters with a tilde ($\widetilde\Psi,\widetilde\Phi,\dots$) to denote subsets of $\calT^k$; to all of them we often attach $k$ in superscript.

A run descriptor in $\calT^k$ is of the form $\sigma=(p,\Psi^n,\Psi^{n-1},\dots,\Psi^{k+1},f)$.
Its first coordinate, $p$, is called the state of $\sigma$.
The sets $\Psi^i$, for $i\in\prz{k+1}{n}$, are called \emph{assumption sets} of $\sigma$, and are denoted $\ass^i(\sigma)$.
The last coordinate, $f$, is called a \emph{productivity flag} of $\sigma$.
When $f=\np$, we say that $\sigma$ is \emph{nonproductive}; otherwise, it is \emph{productive}.
By $\calT_\np$ and $\calT_\pr$ we denote the subsets of $\bigcup_{k\in\prz{0}{n}}\calT^k$ containing only nonproductive and productive run descriptors, respectively.

A run descriptor $\sigma$ assigned to some $k$-stack $s^k$ describes a run that starts in a configuration with state $p$ and topmost $k$-stack $s^k$.
The run descriptor ``can be used'' only when the stack $t^n:t^{n-1}:\dots:t^{k+1}:s^k$ in this configuration is such that for each $i\in\prz{k+1}{n}$ to the $i$-stack $t^i$ we have assigned $\pi_2(\Psi^i)$.
An assumption $(m,\tau)\in\Psi^i$ is used when (a copy of) the stack $t^i$ becomes uncovered.
The run descriptor $\tau$ describes a run from such a configuration $d$;
this run is a suffix of the run from $c=(p,t^n:t^{n-1}:\dots:t^{k+1}:s^k)$.
The run from $c$ to $d$, which uncovers $t^i$, is an $i$-return.
The monoid element $m$ describes the word $w$ read by the return: $m=\phi(w)$.

Beside of the state $p$, and the assumption sets, in $\sigma$ we also have a productivity flag.
Roughly speaking, the run descriptor $\sigma$ is productive if $s^k$ is itself responsible for reading some $\sharp$ symbols.
It means that either some reading of a $\sharp$ symbol is performed ``inside $s^k$'', or some productive run descriptor 
(coming from some assumption set $\Psi^i$) is used at least twice as an assumption
(the latter also increases the number of $\sharp$ symbols read, since some reading described by this productive assumption is repeated).
Thanks to the productivity flag, we can estimate the number of $\sharp$ symbols read, by calculating the number of productive run descriptors used.

One may wonder which runs have a description by a run descriptor.
The answer is that all runs: we do not restrict ourselves to any specific kind of runs at this point.

We now give more intuitions on run descriptors, in particular cases.
Run descriptors in $\calT^n$, assigned to stacks $s^n$ of the maximal order $n$, are simply of the form $(p,f)$.
When the starting state $p$ is fixed, we only have two run descriptors: $(p,\np)$ and $(p,\pr)$.
The former describes runs from $(p,s^n)$ that do not read any $\sharp$ symbols, while the latter those that do read some $\sharp$ symbols.

A run descriptor in $\calT^{n-1}$ is of the form $(p,\Psi^n,f)$.
When assigned to a stack $s^{n-1}$, it describes a run $R$ from a configuration of the form $c=(p,t^n:s^{n-1})$.
It is possible that $R$ never visits $t^n$, and only builds on top of $s^{n-1}$ (i.e., $R$ is $(n-1)$-upper).
In this situation, the set of assumptions $\Psi^n$ is empty, and the flag $f$ simply says whether $R$ reads some $\sharp$ symbols.
The opposite case is that $R$ uncovers the stack $t^n$ in some configuration $d=(q,t^n)$,
that is, that some its prefix $\subrun{R}{0}{i}$ is an $n$-return.
In this situation $\Psi^n=\{(\phi(\subrun{R}{0}{i}),\tau)\}$, where $\tau$ describes the suffix $\subrun{R}{i}{|R|}$.
Because we are considering the highest order, when $t^n$ is uncovered in $R(i)$ there are no other copies of $t^n$.
This means that only a single assumption may be used for $t^n$ (i.e., $|\Psi^n|\leq 1$), and this assumption is used only once.
The flag $f$ simply says whether the prefix $\subrun{R}{0}{i}$ reads some $\sharp$ symbol.

For run descriptors in $\calT^{n-2}$ the situation becomes more interesting.
We explain this by means of an example.

\begin{exa}
	Consider a run $R$ such that
	\begin{itemize}
	\item	$R(0)=(p_0,t^n:t^{n-1}:s^{n-2})$,
	\item	$\subrun{R}{0}{i}$ is an $(n-1)$-return with $R(i)=(p_1,(t^n:t^{n-1}:u^{n-2}):\posp(t^{n-1}))$ (notice that $\subrun{R}{0}{i}$ performs some $\push^n_\gamma$ without a corresponding $\pop^n$);
	\item	$\subrun{R}{i}{j}$ is an $n$-return with $R(j)=(p_2,t^n:t^{n-1}:u^{n-2})$;
	\item	$\subrun{R}{j}{k}$ is an $(n-1)$-return with $R(k)=(p_3,t^n:t^{n-1})$;
	\item	$\subrun{R}{k}{l}$ is an $n$-return with $R(l)=(p_4,t^n)$.
	\end{itemize}
	Let us see how such a run is described by a run descriptor $\sigma=(p_0,\Psi^n,\Psi^{n-1},f)\in\calT^{n-2}$,
	which can be assigned to the $(n-2)$-stack $s^{n-2}$.
	Necessarily $\Psi^n$ is a singleton containing $(\phi(\subrun{R}{0}{l}),\tau_4)$ for $\tau_4$ describing the suffix $\subrun{R}{l}{|R|}$.
	The set $\Psi^{n-1}$ contains in general two elements: $(\phi(\subrun{R}{0}{i}),\tau_1)$ for $\tau_1$ describing the suffix $\subrun{R}{i}{|R|}$,
	and $(\phi(\subrun{R}{0}{k}),\tau_3)$ for $\tau_3$ describing the suffix $\subrun{R}{k}{|R|}$.
	If $\tau_1\neq\tau_3$, then $f$ says whether $\subrun{R}{0}{i}$ or $\subrun{R}{j}{k}$ reads some $\sharp$ symbol
	(and the flags in $\tau_1,\tau_3,\tau_4$ are responsible for the subruns $\subrun{R}{i}{k}$, $\subrun{R}{k}{l}$, and $\subrun{R}{l}{|R|}$, respectively).
	In may also happen that $\tau_1=\tau_3$. %$(\phi(\subrun{R}{0}{i}),\tau_1)=(\phi(\subrun{R}{0}{k}),\tau_3)$.
	In this situation, we say that the run descriptor $\tau_1$ is used twice as an assumption.
	Then $f=\pr$ if $\subrun{R}{0}{i}$ or $\subrun{R}{j}{k}$ reads some $\sharp$ symbol, but also when $\tau_1$ is productive 
	(i.e., when a productive run descriptor is used more than once as an assumption).
	The intuition for this is that now while looking at run descriptors in $\pi_2(\Psi^{n-1})$ it is not visible that there are two subruns $\subrun{R}{i}{j}$ and $\subrun{R}{k}{l}$ reading $\sharp$ symbols,
	as they both correspond to the same assumption;
	by setting $f=\pr$ we reflect the fact that $R$ reads more $\sharp$ symbols than in the situation when every assumption would be used only once.\footnote{
		\new{One can imagine two possible definitions of ``the same assumption'': we may compare either only run descriptors (in our case, $\tau_1=\tau_3$),
		or pairs consisting of a monoid element and a run descriptor (in our case, $(\phi(\subrun{R}{0}{i}),\tau_1)=(\phi(\subrun{R}{0}{k}),\tau_3)$).
		At first glance both definitions look equally good, or the latter definition seems to be more natural than the former.
		It turns out, however, that the latter definition is problematic.
		The difficulty is that for pairs that are originally different, $(m_1,\sigma_1)\neq(m_2,\sigma)$, 
		it may happen that after} \new{multiplying them by a monoid element they become equal, $(m\cdot m_1,\sigma_1)=(m\cdot m_2,\sigma)$.
		We prefer to avoid this, and hence we stick to the former definition.}}
\end{exa}

We remark that a run descriptor $\sigma=(p,\Psi^n,\Psi^{n-1},\dots,\Psi^{k+1},f)$ should not be seen as a classical implication of the form
``if the stacks below the topmost $k$-stack satisfy assumptions $\Psi^n,\Psi^{n-1},\dots,\Psi^{k+1}$, then there exists a run satisfying some properties''.
It is much closer to an implication in a linear logic.
Indeed, the conclusion of the implication is trivial, as it, basically, says only that there exists some (arbitrary) run.
The interesting information about the run is contained in the sets of assumptions:
by specifying an assumption, we say that there is a suffix of the run corresponding to this assumption.
In particular, it is essential that there are no redundant assumptions (every assumption ``has to be used'', at least once).
Moreover, the information whether some assumptions are used more than once is also recorded, in the productivity flag.

\subsubsection*{Composers}

\makebox[0cm]{\ }%    <--- przekolorowanie nagłówka, razem z kropką
For $m\in M$ and $\Psi\subseteq M\times\calT^k$ we use the notation $m\circ\Psi$ for $\{(m\cdot m',\sigma)\mid(m',\sigma)\in\Psi\}$.
Given a run descriptor $\sigma=(p,\Psi^n,\Psi^{n-1},\dots,\Psi^{l+1},f)\in\calT^l$, for $k\in\prz{l}{n}$
by $\red^{k}(\sigma)$ we denote the ``reduced'' run descriptor $(p,\Psi^n,\Psi^{n-1},\dots,\Psi^{k+1},g)\in\calT^{k}$ in which 
\begin{align*}
	g=\np \mbox{}\Leftrightarrow ( f=\np\mbox{, and }\pi_2(\Psi^{i})\subseteq\calT_\np\mbox{ for each }i\in\prz{l+1}{k})\,.
\end{align*}
The following proposition is a direct consequence of the definition.

\begin{prop}\lab{prop:assoc-red}
	For $0\leq l\leq j\leq k\leq n$ and $\sigma\in\calT^l$ it holds that $\red^k(\red^j(\sigma))=\red^k(\sigma)$.
\qed\end{prop}

We now define composers, which are used to compose run descriptors corresponding to smaller stacks into run descriptors corresponding to greater stacks.

\begin{defi}\lab{def:composer}
	Consider a tuple $(\Phi^k,\Phi^{k-1}\dots,\Phi^l;\Psi^k;f)$, where $0\leq l\leq k\leq n$, 
	$\Phi^i\subseteq M\times\calT^i$ for each $i\in\prz{l}{k}$, $\Psi^k\subseteq M\times\calT^k$, and $f\in\{\np,\pr\}$.
	Such a tuple is called a \emph{composer} if
	\begin{enumerate}[label=(C\arabic*)]
	\item\lab{pkt:compo-ass}
		$\Phi^i=\bigcup\{m\circ\ass^i(\sigma)\mid(m,\sigma)\in\Phi^l\}$ for each $i\in\prz{l+1}{k}$,
	\item\lab{pkt:compo-red}
		$\Psi^k=\{(m,\red^k(\sigma))\mid(m,\sigma)\in\Phi^l\}$,
	\item\lab{pkt:compo-inj}
		$|\pi_2(\Psi^k)|=|\pi_2(\Phi^l)|$ (which means that each $\sigma\in\pi_2(\Phi^l)$ gives a different $\red^k(\sigma)$), and
	\item\lab{pkt:compo-flag}
		$f=\np$ if and only if
		$\pi_2(\ass^i(\sigma))\cap\pi_2(\ass^i(\tau))\subseteq\calT_\np$ for each $i\in\prz{l+1}{k}$ and each $\sigma,\tau\in\pi_2(\Phi^l)$ such that $\sigma\neq\tau$.
	\end{enumerate}
\end{defi}

Suppose that we have a $k$-stack $t^k=s^k:s^{k-1}:\dots:s^l$, where, for $i\in[l,k]$, elements of $\Phi^i$ ``are assigned'' to $s^i$.
Intuitively, a tuple $(\Phi^k,\Phi^{k-1}\dots,\Phi^l;\Psi^k;f)$ is a composer, if in such a situation the set $\Psi^k$ can be assigned to the whole stack $t^k$.
As we see, the definition does not depend on the stacks that are actually composed, only on the sets $\Phi^i$ assigned to these stacks.
One can think about $\Phi^k,\Phi^{k-1}\dots,\Phi^l$ as about inputs to the composer, and about $\Psi^k$ and $f$ as outputs.
Nevertheless, already $\Phi^l$ determines all remaining coordinates of the composer (when $l$ and $k$ are fixed).
We remark that not every set $\Phi^l\subseteq M\times\calT^l$ can be used in a composer, due to Condition~\ref{pkt:compo-inj} of the definition.

The definition says that run descriptors assigned to $t^k$ are of the form $\red^k(\sigma)$ for $\sigma$ assigned to $s^l$ 
(the reason is that if the topmost $k$-stack of a configuration is $t^k$, then its topmost $l$-stack is $s^l$).
Moreover, the assumptions of $\sigma$ that are contained in $\ass^i(\sigma)$ for $i\in\prz{l+1}{k}$ have to be realized by the stacks $s^i$.
We notice that the run descriptor $\red^k(\sigma)$ is productive when $\sigma$ is productive or some of the assumptions in $\ass^i(\sigma)$ for $i\in\prz{l+1}{k}$ is productive
(cf.~the definition of $\red^k(\sigma)$);
in other words, in a run corresponding to the run descriptor, the part corresponding to the stack $t^k$ is productive when some part corresponding to $s^i$ for some $i\in\prz{l}{k}$ is productive.

The composer itself also has a productivity flag $f$.
The intuition is that we set this flag to $\pr$ if the runs described by elements of $\Psi^k$ read more $\sharp$ symbols than those described by elements of $\Phi^k,\Phi^{k-1}\dots,\Phi^l$, in total.
This is the case when a productive run descriptor (coming from some $\Phi^i$ for $i\in\prz{l+1}{k}$) is used as an assumption for more than one element of $\Phi^l$.
While summing over all run descriptors from all $\Phi^i$, such a run descriptor is added only once, but it contributes to more than one run descriptor in $\Psi^k$.

Another important issue is that in $\Phi^k,\Phi^{k-1}\dots,\Phi^l$ we only have elements that really contribute while constructing elements of $\Psi^k$.
Simultaneously, we require that every run descriptor in $\Psi^k$ has exactly one realization by run descriptors from $\Phi^k,\Phi^{k-1}\dots,\Phi^l$
(i.e., it is obtained by reducing exactly one run descriptor from $\Phi^l$).
The justification for both these properties is the same: we want to ``approximate'' the number of $\sharp$ symbols read by runs described by elements of $\Psi^k$ 
while looking at the number of $\sharp$ symbols read by runs described by elements of $\Phi^k,\Phi^{k-1}\dots,\Phi^l$.
Any redundant run descriptors in $\Phi^k,\Phi^{k-1}\dots,\Phi^l$ would bias our calculations;
multiple decompositions of a single element of $\Psi^k$ would also bias our calculations.

In the sets $\Phi^k,\Phi^{k-1}\dots,\Phi^l$ and $\Psi^k$ we also have monoid elements, not only run descriptors.
Intuitively, a pair $(m,\sigma)\in\Psi^k$ (or $(m,\sigma)\in\Phi^i$) corresponds to a run consisting of two parts:
the first part reads a word evaluating to $m$ and uncovers the stack $t^k$ ($s^i$, respectively);
the second part starts when the topmost $k$-stack is $t^k$ (the topmost $i$-stack is $s^i$), and it is described by $\sigma$.
The monoid elements in $\Psi^k$ are the same as in $\Phi^l$, since uncovering $t^k$ means uncovering $s^l$.
On the other hand, elements in $\Phi^i$ for $i\in\prz{l+1}{k}$ are obtained as the composition of $m$ coming from $(m,\sigma)\in\Phi^l$ (describing the word read before uncovering $s^l$) 
and of $m'$ coming from a particular assumption $(m',\tau)\in\ass^i(\sigma)$ (describing the word read after uncovering $s^l$, but before uncovering $s^i$).

\begin{exa}\lab{ex:composer}
	Suppose that $n=2$.
	Let $\tau_\np\in\calT^1\cap\calT_\np$ and $\tau_\pr\in\calT^1\cap\calT_\pr$.
	Consider the following elements of $\calT^0$:
	\begin{align*}
		\sigma_1&=(p,\Psi^2,\{(m_1,\tau_\np),(m_2,\tau_\pr)\},\np)\,,&
		\sigma_2&=(p,\Psi^2,\{(m_3,\tau_\np)\},\np)\,,\displaybreak[0]\\
		\sigma_1'&=(p,\Psi^2,\{(m_3,\tau_\np)\},\pr)\,,&
		\sigma_3&=(q,\Phi^2,\{(m_4,\tau_\np),(m_5,\tau_\pr)\},\pr)\,,
	\end{align*}
	and the following elements of $\calT^1$:
	\begin{align*}
		\xi_1&=(p,\Psi^2,\pr)\,,&
		\xi_2&=(p,\Psi^2,\np)\,,&
		\xi_3&=(q,\Phi^2,\pr)\,.
	\end{align*}
	It holds that $\red^1(\sigma_i)=\xi_i$ for $i\in\{1,2,3\}$, and $\red^1(\sigma_1')=\xi_1$.
	We have composers resulting in a single pair from $M\times\calT^1$, like
	\begin{align*}
		&(\{(m_6\cdot m_1,\tau_\np),(m_6\cdot m_2,\tau_\pr)\},\{(m_6,\sigma_1)\};\{(m_6,\xi_1)\};\np)\,,&&\mbox{or}\displaybreak[0]\\
		&(\{(m_6\cdot m_3,\tau_\np)\},\{(m_6,\sigma_1')\};\{(m_6,\xi_1)\};\np)\,.
	\end{align*}
	Notice that these two composers have the same ``output set''.
	We may also repeat the same run descriptor with multiple monoid elements:
	\begin{align*}
		(\{(m_6\cdot m_3,\tau_\np),(m_7\cdot m_3,\tau_\np)\},\{(m_6,\sigma_1'),(m_7,\sigma_1')\};\{(m_6,\xi_1),(m_7,\xi_1)\};\np)\,.
	\end{align*}
	Here, the situation that $m_6\cdot m_3=m_7\cdot m_3$ is allowed as well (and then the first set has only a single element).
	Observe that that all the above composers are nonproductive, even though they involve productive run descriptors.
	Next, we can also have composers involving multiple run descriptors, like
	\begin{align*}
		(\{(m_6\cdot m_3,\tau_\np)\},\{(m_6,\sigma_1'),(m_6,\sigma_2)\};\{(m_6,\xi_1),(m_6,\xi_2)\};\np)\,.
	\end{align*}
	This composer is again nonproductive: 
	although there is a run descriptor $\tau_\np$ that appears in assumption sets of both $\sigma_1'$ and $\sigma_2$, this run descriptor $\tau_\np$ is nonproductive.
	But if a productive run descriptor $\tau_\pr$ appears in assumption sets of two run descriptors (like for $\sigma_1$ and $\sigma_3$), we have a productive composer:
	\begin{align*}
		(\{(m_6\cdot m_1,\tau_\np),(m_7\cdot m_4,\tau_\np),(m_6\cdot m_2,\tau_\pr),(m_7\cdot m_5,\tau_\pr)\},\{(m_6,\sigma_1),(m_7,\sigma_3)\};&\\
			&\hspace{-7em}\{(m_6,\xi_1),(m_7,\xi_3)\};\pr)\,.
	\end{align*}
	As negative examples, we have the following tuples, which are not composers:
	\begin{align*}
		&(\{(m_6\cdot m_3,\tau_\np),(m_6\cdot m_2,\tau_\pr)\},\{(m_6,\sigma_1')\};\{(m_6,\xi_1)\};\np)\,,&&\mbox{and}\displaybreak[0]\\
		&(\{(m_6\cdot m_1,\tau_\np),(m_6\cdot m_2,\tau_\pr),(m_7\cdot m_3,\tau_\np)\},\{(m_6,\sigma_1),(m_7,\sigma_1')\};\\
			&&&\hspace{-4em}\{(m_6,\xi_1),(m_7,\xi_1)\};\np)\,.
	\end{align*}
	In the first tuple, the first set contains a redundant pair $(m_6\cdot m_2,\tau_\pr)$, which is forbidden.
	In the second tuple, we have simultaneously two realizations of $\xi_1$: one using $\sigma_1$, and one using $\sigma_1'$; this is forbidden as well.
\end{exa}

In the next proposition we notice that composers are associative.

\begin{prop}\label{prop:assoc-compo}
	Let $0\leq l\leq j\leq k\leq n$.
	%, let $\Phi^i\in M\times\calT^i$ for $i\in\prz{m,k}$, let $\Psi^k\in\calT^k$, and let $f\in\{\pr,\np\}$.
	For all fixed $\Phi^k,\Phi^{k-1},\dots,\Phi^l,\Psi^k$ and $f\in\{\np,\pr\}$, the following two conditions are equivalent:
	\begin{itemize}
	\item	$(\Phi^k,\Phi^{k-1},\dots,\Phi^l;\Psi^k;f)$ is a composer, and
	\item	$(\Phi^j,\Phi^{j-1},\dots,\Phi^l;\Psi^j;f_1)$ and $(\Phi^k,\Phi^{k-1},\dots,\Phi^{j+1},\Psi^j;\Psi^k;f_2)$ are composers
		for some $\Psi^j$ and $f_1,f_2$ such that $(f=\np)\Leftrightarrow(f_1=f_2=\np)$.
	\end{itemize}
\end{prop}

\begin{proof}
	We use symbols $X$ and $X_1, X_2$ for the tuples from the first and second item, respectively.
	First, notice that the definition of a composer instantiated either for $X$, or for both $X_1$ and $X_2$ simultaneously, 
	implies that $\Phi^i\subseteq M\times\calT^i$ for $i\in\prz{l}{k}$, and $\Psi^k\subseteq M\times\calT^k$.
	While proving the right-to-left implication, Condition~\ref{pkt:compo-red} of the definition of a composer instantiated for $X_1$ implies that
	$\Psi^j=\{(m,\red^j(\sigma))\mid(m,\sigma)\in\Phi^l\}$.
	While proving the opposite implication, we choose $\Psi^j$ in this way; the effect is that $X_1$ satisfies Condition~\ref{pkt:compo-red}.
	
	We see that $\ass^i(\red^j(\sigma))=\ass^i(\sigma)$ for $i\in\prz{j+1}{k}$ and $\sigma\in\calT^l$,
	so Condition~\ref{pkt:compo-ass} for $X$ is equivalent to the conjunction of Condition~\ref{pkt:compo-ass} for $X_1$ and $X_2$.
	
	By Proposition~\ref{prop:assoc-red} we have that $\red^k(\red^j(\sigma))=\red^k(\sigma)$ for $\sigma\in\calT^l$, which implies that
	Condition~\ref{pkt:compo-red} instantiated for $X$ is equivalent to Condition~\ref{pkt:compo-red} instantiated for $X_2$.
	
	Notice that $|\pi_2(\Psi^k)|\leq|\pi_2(\Psi^j)|\leq|\pi_2(\Phi^l)|$, 
	because $\pi_2(\Psi^k)=\{\red^k(\sigma)\mid\sigma\in\pi_2(\Psi^j)\}$ and $\pi_2(\Psi^j)=\{\red^j(\sigma)\mid\sigma\in\pi_2(\Phi^l)\}$.
	Thus $|\pi_2(\Psi^k)|=|\pi_2(\Phi^l)|$ if and only if $|\pi_2(\Psi^k)|=|\pi_2(\Psi^j)|=|\pi_2(\Phi^l)|$;
	we obtain equivalence for Condition~\ref{pkt:compo-inj}.
	
	Finally, Condition~\ref{pkt:compo-flag} for $X_1$ says that $f_1=\np$ if and only if 
	$\pi_2(\ass^i(\sigma))\cap\pi_2(\ass^i(\tau))\subseteq\calT_\np$ for each $i\in\prz{l+1}{j}$ and each $\sigma,\tau\in\pi_2(\Phi^l)$ such that $\sigma\neq\tau$;
	otherwise $f_1=\pr$.
	For $X_2$ and $f_2$ we have the same with $i\in\prz{j+1}{k}$ (here we use again the fact that $\ass^i(\red^j(\sigma))=\ass^i(\sigma)$).
	Thus, while proving the right-to-left implication we have these equivalences for $f_1$ and $f_2$; they imply that $(f=\np)\Leftrightarrow(f_1=f_2=\np)$.
	While proving the left-to-right implication we define $f_1$ and $f_2$ so that this is satisfied;
	then we have Condition~\ref{pkt:compo-flag} for $X_1$ and $X_2$, and the equivalence $(f=\np)\Leftrightarrow(f_1=f_2=\np)$.
\end{proof}

\subsubsection*{Derivation Trees}

\makebox[0cm]{\ }%    <--- przekolorowanie nagłówka, razem z kropką
Next, we say when a run descriptor $\sigma$ from $\calT^0$ can be assigned to a stack symbol $\gamma$.
To this end, we define how statements of the form $\gamma\vdash\sigma$ can be derived.
Such a statement means that from a configuration with topmost stack symbol $\gamma$ one can start a run described by the run descriptor $\sigma$ 
(assuming that the stacks below $\gamma$ satisfy the assumptions of $\sigma$).
Actually, it is not enough to define when $\gamma\vdash\sigma$ is true;
we explicitly need to handle derivation trees justifying such statements.
Thus, in Definition~\ref{def:types} we define the notion of a \emph{derivation tree for $\gamma\vdash\sigma$}.
Having such a derivation tree $D$, $\gamma\vdash\sigma$ is called the \emph{conclusion} of $D$,
and $\sigma$ is called the run descriptor of $D$ and is denoted $\rd(D)$.

\begin{defi}\lab{def:types}
	We define the set of \emph{derivation trees} as the smallest set satisfying the following conditions.
	Let $p$ be a state, and $\gamma$ a stack symbol.
	\begin{enumerate}
	\item\lab{pkt:dt-empty}
		A triple $(\dtempty\gamma,p)$ is a derivation tree for $\gamma\vdash(p,\emptyset,\emptyset,\dots,\emptyset,\np)$.
	\item\lab{pkt:dt-read}
		Suppose that $\delta(\gamma,p)=\read(\vec{q})$ and that $D'$ is a derivation tree for $\gamma\vdash\tau$, where the state of $\tau$ is $\vec{q}(a)$ for some $a\in A$.
		Denote $\Phi^i=\phi(a)\circ\ass^i(\tau)$ for $i\in\prz{1}{n}$.
		Then $(\dtread p,D')$ is a derivation tree for $\gamma\vdash(p,\Phi^n,\Phi^{n-1},\dots,\Phi^1,f)$, where $f=\np$ if and only if $\tau\in\calT_\np$ and $a\neq\sharp$.
	\item\lab{pkt:dt-pop}
		Suppose that $\delta(\gamma,p)=(q,\pop^k)$, and that $\tau^k\in\calT^k$ is a run descriptor with state $q$.
		Then $(\dtpop\gamma,p,\tau^k)$ is a derivation tree for
		\begin{align*}
			\gamma\vdash(p,\ass^n(\tau^k),\ass^{n-1}(\tau^k),\dots,\ass^{k+1}(\tau^k),\{(\mathbf{1}_M,\tau^k)\},\emptyset,\dots,\emptyset,\np)\,.
		\end{align*}
	\item\lab{pkt:dt-push}
		Suppose that $\delta(\gamma,p)=(q,\push^k_\alpha)$, and that $D'$ is a derivation tree for $\alpha\vdash\tau$, where the state of $\tau$ is $q$.
		Denote $\Psi^i=\ass^i(\tau)$ for $i\in\prz{1}{n}$.
		Moreover, suppose that $(\Phi^k,\Phi^{k-1},\dots,\Phi^0;\Psi^k;f)$ is a composer,
		and $\mathfrak{D}$ is a set of derivation trees, all having the stack element $\gamma$ in their conclusion, and such that $\{\rd(E)\mid E\in\mathfrak{D}\}=\pi_2(\Phi^0)$ and $|\mathfrak{D}|=|\pi_2(\Phi^0)|$.
		Let
		\begin{align*}
			\Upsilon^i=\left\{\begin{array}{ll}
				\Psi^i  & \mbox{for }i\in\prz{k+1}{n}\,,\\
				\Phi^i  & \mbox{for }i=k\,,\\
				\Psi^i\cup\Phi^i  & \mbox{for }i\in\prz{1}{k-1}\,.
			\end{array}\right.
		\end{align*}
		Then $(\dtpush\gamma,p,D',\mathfrak{D})$ is a derivation tree for $\gamma\vdash(p,\Upsilon^n,\Upsilon^{n-1},\dots,\Upsilon^1,g)$,
		where $g=\np$ if and only if $f=\np$, and $\{\tau\}\cup\pi_2(\Phi^0)\subseteq\calT_\np$,
		and $\pi_2(\Psi^i)\cap\pi_2(\Phi^i)\subseteq\calT_\np$ for each $i\in\prz{1}{k-1}$.
	\end{enumerate}
\end{defi}

The \emph{depth} of a derivation tree $D$, denoted $\depth(D)$, is defined naturally: it is $0$ in Cases~(\ref{pkt:dt-empty}) and (\ref{pkt:dt-pop}),
$1+\depth(D')$ in Case~(\ref{pkt:dt-read}), and $1+\max(\depth(D'),\max_{E\in\mathfrak{D}}\depth(E))$ in Case~(\ref{pkt:dt-push}).

Notice that a derivation tree $D$ determines its conclusion.
This can be seen by induction on the structure of $D$.
In Cases~(\ref{pkt:dt-empty}) and~(\ref{pkt:dt-pop}) this is immediate.
In Case~(\ref{pkt:dt-read}), $\gamma$ and $\tau$ are determined by $D'$, and the letter $a$ is determined by $p$ and by the state of $\tau$
(recall that $\vec{q}$ is required to be injective, by the definition of the $\read$ operation);
this already fixes the whole conclusion of $D$. Case~(\ref{pkt:dt-push}) is the most complicated one.
First, we can see that $\pi_2(\Phi^0)$ and all $\Psi^i$ are fixed by $D'$ and $\mathfrak{D}$. 
Then, by the definition of a composer, we have that the set $\Phi^0$ has to contain those pairs $(m,\tau)\in M\times\pi_2(\Phi^0)$ for which $(m,\red^k(\tau))\in\Psi^k$.
The set $\Phi^0$ fixes the rest of the composer, and thus the whole conclusion of $D$.

We now comment on the intuitions staying behind Definition~\ref{def:types}.
Let $D$ be a derivation tree for $\gamma\vdash\sigma$, where the state of $\sigma$ is $p$.
We should have in mind a run $R$ with $R(0)=(p,s^n:s^{n-1}:\dots:s^1:u^0_\gamma)$, where $\posl(u^0_\gamma)=\gamma$.
Basically, the run descriptor $\sigma$ describes such a run, and the derivation tree $D$ specifies parts of this run for which the $0$-stack $u^0_\gamma$ is responsible.

Case~(\ref{pkt:dt-empty}) corresponds to an empty run ($|R|=0$).
Thus, the assumption sets of $\sigma$ are empty (stacks $s^i$ are never uncovered), and the run descriptor is nonproductive.

In the remaining cases, nonempty runs are considered.
Case~(\ref{pkt:dt-read}) talks about a run $R$ starting with a $\read$ operation.
Say that this operation reads a symbol $a$, and ends in a state $\vec{q}(a)$.
The derivation tree $D'$ (and, in particular, its run descriptor $\tau$) describes the suffix $\subrun{R}{1}{|R|}$.
Thus, we require that the state in $\tau$ is $\vec{q}(a)$.
Because $\subrun{R}{0}{1}$ does not modify the stack, assumptions from $\ass^i(\tau)$, referring to $\tops^i(\pop^i(R(1)))$,
refer simultaneously to $\tops^i(\pop^i(R(0)))$ (for $i\in\prz{1}{n}$).
This is expressed by the fact that as $\ass^i(\sigma)$ we almost take assumptions from $\ass^i(\tau)$.
The only difference is that we multiply the monoid element in every pair by $\phi(a)$ on the left.
This is because monoid elements in $\ass^i(\tau)$ correspond to some prefixes of $\subrun{R}{1}{|R|}$, 
while monoid elements in $\ass^i(\sigma)$ should talk about prefixes of the whole $R$;
the latter prefixes additionally read the symbol $a$ at the very beginning.
Our new run descriptor $\sigma$ is productive either when the first operation reads the $\sharp$ symbol,
or when the run descriptor $\tau$, describing the rest of the run, is productive.
Recall that $\vec{q}$ is required to be injective (cf.\ the definition of the $\read$ operation), so seeing the state $\vec{q}(a)$ we can determine the symbol $a$ that was read.

In Case~(\ref{pkt:dt-pop}) the first operation of $R$ is $\pop^k$, which leads to a configuration of the form $(q,s^n:s^{n-1}:\dots:s^k)$.
The suffix $\subrun{R}{1}{|R|}$ is described by a run descriptor $\tau^k$.
In particular, the state in $\tau^k$ should be $q$.
For $i\in\prz{1}{k-1}$ the $i$-stack $s^i$ is never uncovered (it is destroyed by the $\pop^k$ operation), so the assumption set $\ass^i(\sigma)$ is set to $\emptyset$.
The $k$-stack $s^k$ is uncovered after the first operation, so we put $\tau^k$ to the assumption set $\ass^k(\sigma)$, 
together with the neutral element of the monoid (because the word read by $\subrun{R}{0}{1}$ is empty).
This is the only pair in $\ass^k(\sigma)$, because before $R(1)$ no copies of $s^k$ are created, so $s^k$ cannot be uncovered again.
For $i\in\prz{k+1}{n}$, the assumption set $\ass^i(\tau^k)$ talks about uncovering $s^i$, and hence it is taken as $\ass^i(\sigma)$ 
(the monoid elements remain unchanged, because $\subrun{R}{0}{1}$ does not read any symbols).
Notice that, unlike for the $\read$ operation, we do not include in $D$ any derivation tree talking about $\tau^k$.
This is because in $D$ we only describe the part of $R$ for which the $0$-stack $u^0_\gamma$ is responsible: 
on the one hand, after the $\read$ operation, the $0$-stack $u^0_\gamma$ is still on the top of the stack and is responsible for continuing the run;
on the other hand, after the $\pop^k$ operation, $u^0_\gamma$ is no longer present on the stack.
The run descriptor $\sigma$ is nonproductive: $u^0_\gamma$ is not responsible for reading any $\sharp$ symbol, 
and every assumption of $\sigma$ is used only once (for $i\in\prz{k+1}{n}$, we simply pass every assumption from $\ass^i(\sigma)$ to $\tau^k$;
independently, $\tau^k$ may use these assumptions more than once, but this is a responsibility of $\tau^k$, not of $\sigma$).

Finally, we have Case~(\ref{pkt:dt-push}), where the first operation of $R$ is $\push^k_\alpha$, leading to a configuration of the form 
$(q,s^n:s^{n-1}:\dots:s^{k+1}:t^k:\posp(s^{k-1}:s^{k-2}:\dots:s^1:u^0_\alpha))$ with $t^k=s^k:s^{k-1}:\dots:s^1:u^0_\gamma$ and with $\posl(u^0_\alpha)=\alpha$.
A run descriptor $\tau$ (having state $q$) describes the suffix $\subrun{R}{1}{|R|}$.
In $D$ we include a derivation tree $D'$ for $\alpha\vdash\tau$, talking about the parts of $\subrun{R}{1}{|R|}$ for which the new topmost $0$-stack $u^0_\alpha$ is responsible,
because $u^0_\gamma$ is responsible for creating $u^0_\alpha$.
For $i\in\prz{1}{k-1}\cup\prz{k+1}{n}$, the assumption set $\ass^i(\tau)$ refers to the $i$-stack $s^i$, 
and hence we take assumptions from $\ass^i(\tau)$ to $\ass^i(\sigma)$.
The assumption set $\ass^k(\tau)$, in turn, refers to the $k$-stack $t^k$, not directly to $s^k$.
Because of that, we need a composer to decompose the assumption set $\ass^k(\tau)$ to sets $\Phi^k,\Phi^{k-1},\dots,\Phi^0$, referring to stacks $s^k,s^{k-1},\dots,s^1,u^0_\gamma$.
Then, for $i\in\prz{1}{k}$, assumptions from $\Phi^i$ are taken to $\ass^i(\sigma)$.
Elements of $\Phi^0$ refer to our $0$-stack $u^0_\gamma$, and hence should be realized by our derivation tree.
Thus, for every run descriptor in $\Phi^0$ we provide a derivation tree (in the set $\mathfrak{D}$).
We should set $\sigma$ to be productive if $u^0_\gamma$ is responsible for reading some $\sharp$ symbols, 
or when some productive run descriptor from an assumption set is used more than once.
Maybe this happens inside $\tau$ or inside some run descriptor from $\Phi^0$; when any of them is productive, we set $\sigma$ to be productive.
But it may also happen that a productive run descriptor is used as an assumption by $\tau$ and by some run descriptor from $\Phi^0$,
or that a productive run descriptor is used as an assumption by multiple run descriptors from $\Phi^0$ (in the latter case, the composer is productive);
in both these situation we also set $\sigma$ to be productive.

As already said, a derivation tree $D$ is not a representation of the whole run $R$, it only talks about parts of $R$ for which the topmost $0$-stack $u^0_\gamma$ is responsible.
It is not, however, a precise representation even of these parts.
This can be seen in Case~(\ref{pkt:dt-push}) of the definition.
It is possible that $\tau$ uses some run descriptor $\xi$ from the assumption set $\ass^k(\tau)$ more than once.
This means that we have multiple suffixes of the run $R$ described by $\xi$.
But in $\mathfrak{D}$ we include only one derivation tree talking about $\xi$, corresponding to one of the suffixes of $R$.
Thus, in a sense, a derivation tree for $\gamma\vdash\sigma$ is a proof that a run described by $\sigma$ exists, rather than a way of representing any such run.

\newcommand{\nonempty}{\mathsf{ne}}
\begin{exa}\lab{exa:der-tree}
	Consider the $2$-DPDA $\calA_1$ depicted below; its stack alphabet is $\{\gamma\}$, and input alphabet $A=\{a,b,\sharp\}$.
	% state set $\{q_i\mid i\in\prz{1}{7}\}$, and transitions
	\begin{center}
		\import{pics/}{example-7-12.pdf_tex_ok}
	\end{center}
	%\begin{align*}
	%	\delta(q_1,\gamma)&=(q_2,\push^1_\gamma)\,,&&\delta(q_2,\gamma)=(q_3,\push^2_\gamma)\,,&&\delta(q_3,\gamma)=(q_4,\pop^1)\,,\\
	%	\delta(q_4,\gamma)&=\read(\vec{q})\,,&&\makebox[0cm][l]{where $\vec{q}(a)=q_5$, $\vec{q}(b)=q_6$, $\vec{q}(\sharp)=q_7$,\qquad and}\\
	%	\delta(q_i,\gamma)&=(q_3,\pop_2)&&\mbox{for }i\in\{5,6,7\}\,.
	%\end{align*}
	The arrow from $q_1$ to $q_2$ denotes that $\delta(q_1,\gamma)=(q_2,\push^1_\gamma)$, the arrow from $q_4$ to $q_5$ denotes that $\delta(q_4)=\read(\vec{q})$ with $\vec{q}(a)=q_5$, and so on.
	Take also a monoid $M=\{1,\nonempty\}$ with $1\cdot 1=1$ and $1\cdot\nonempty=\nonempty\cdot 1=\nonempty\cdot\nonempty=\nonempty$, 
	and a morphism $\phi:A^*\to M$ that maps the empty word to $1$ and all nonempty words to $\nonempty$.

%	Denote
%	\begin{align*}
%		\sigma_i=(q_i,\emptyset,\emptyset,\np)&&\mbox{for }i\in\prz{1}{7}\,.
%	\end{align*}
%	These run descriptors have empty assumption sets, hence they describe runs that never uncover stacks that are below the topmost $0$-stack
%	(in other words, these are runs whose every prefix is $0$-upper).
%	We can derive $a\vdash\sigma_i$ using, for example, the derivation tree $E_i=(\dtempty\gamma,q_i)$, which corresponds to the empty run.
	
	Consider run descriptors
	\begin{align*}
		\sigma_i&=(q_i,\{(1,(q_3,\pr))\},\emptyset,\np)&&\mbox{for }i\in\{5,6,7\}\,,\mbox{ and}\displaybreak[0]\\
		\sigma_{4,f}&=(q_4,\{(\nonempty,(q_3,\pr))\},\emptyset,f)&&\mbox{for }f\in\{\np,\pr\}\,.
	\end{align*}
	They describe runs that are compositions of a $2$-return, and of a run that starts in the state $q_3$ and reads some $\sharp$ symbols (is productive).
	For $i\in\{5,6,7\}$, the $2$-return should not read anything.
	Thus, in this case $D_i=(\dtpop\gamma,q_i,(q_3,\pr))$ is a derivation tree for $\gamma\vdash\sigma_i$.
	On the other hand, for $i=4$ the $2$-return should read a nonempty word, which should contain some $\sharp$ symbol when $f=\pr$.
	We can derive $\gamma\vdash\sigma_{4,\pr}$ by $D_{4,\sharp}=(\dtread q_4,D_7)$, and $\gamma\vdash\sigma_{4,\np}$ by $D_{4,a}=(\dtread q_4,D_5)$, as well as by $D_{4,b}=(\dtread q_4,D_6)$.
	
	For $f\in\{\np,\pr\}$, denote $\sigma_{4,f}^1=\red^1(\sigma_{4,f})$, that is, $\sigma_{4,f}^1=(q_4,\{(\nonempty,(q_3,\pr))\},f)$.
	The next run descriptors that we consider are
	\begin{align*}
		\sigma_{3,f}=(q_3,\{(\nonempty,(q_3,\pr))\},\{(1,\sigma_{4,f}^1)\},\np)&&\mbox{for }f\in\{\np,\pr\}\,.
	\end{align*}
	They describe runs whose some suffix (that starts before reading anything, and after a $1$-return) is described by $\sigma_{4,f}^1$, 
	and some other suffix (that starts after reading a nonempty word, and after a $2$-return) is described by $(q_3,\pr)$.
	We can derive $\gamma\vdash\sigma_{3,f}$ by $D_{3,f}=(\dtpop\gamma,q_3,\sigma_{4,f}^1)$, for $f\in\{\np,\pr\}$.
	Notice that the derivation trees $D_{3,f}$ (and the former ones as well) specify precisely only a prefix of a run; 
	more precisely, the fragment of the run that corresponds to the topmost $0$-stack, before using an assumption.
	It is not checked that the assumption can be fulfilled by any stack that would be placed below the topmost $0$-stack.
	We can equally well have a derivation tree like $(\dtpop\gamma,q_3,\xi_4)$ with $\xi_4=(q_4,\{(1,(q_1,\pr))\},\pr)$, 
	which derives $\gamma\vdash(q_3,\{(1,(q_1,\pr))\},\{(1,\xi_4)\},\np)$, although it is impossible to deliver $\xi_4$ by any $1$-stack 
	(because, e.g., there is no $2$-return that ends in the state $q_1$).
	
	Next, for $f_1,f_2\in\{\np,\pr\}$ consider run descriptors
	\begin{align*}
		&\sigma_{2,f_1,f_2}=(q_2,\{(\nonempty,(q_3,\pr))\},\{(1,\sigma_{4,f_1}^1),(\nonempty,\sigma_{4,f_2}^1)\},f_3)\,,&&\mbox{where }f_3=\pr\Leftrightarrow f_1=f_2=\pr\,.
	\end{align*}
	We can derive $\gamma\vdash\sigma_{2,f_1,f_2}$ by $D_{2,f_1,f_2}=(\dtpush\gamma,q_2,D_{3,f_1},\{D_{3,f_2}\})$.
	The derivation tree has to specify fragments of the run for which the topmost $0$-stack is responsible.
	When applied to a stack $s^2:s^1:s^0$, the $\push^2_\gamma$ operation (resulting in $(s^2:s^1:s^0):\posp(s^1:s^0)$) splits the topmost $0$-stack $s^0$ into two copies;
	$D_{2,f_1,f_2}$ says that, with the new copy of $s^0$, we should continue according to the derivation tree $D_{3,f_1}$.
	The run descriptor $\rd(D_{3,f_1})$ (i.e., $\sigma_{3,f_1}$) says that the $1$-stack $s^1$ will be used with run descriptor $\sigma_{4,f_1}^1$ (before reading anything),
	and that the $2$-stack $s^2:s^1:s^0$ will be used with run descriptor $(q_3,\pr)$ (after reading a nonempty word).
	The definition of our derivation tree involves a composer
	\begin{align*}
		(\{(\nonempty,(q_3,\pr))\}, \{(\nonempty,\sigma_{4,f_2}^1)\}, \{(\nonempty,\sigma_{3,f_2})\};\{(\nonempty,(q_3,\pr))\};\np)\,.
	\end{align*}
	It says that in order to provide $(q_3,\pr)$ by the $2$-stack $s^2:s^1:s^0$,
	we have to provide $\sigma_{3,f_2}$ by $s^0$, and $\sigma_{4,f_2}^1$ by $s^1$, and $(q_3,\pr)$ by $s^2$.
	Because $(q_3,\pr)$ for $s^2:s^1:s^0$ will be used after reading a nonempty word (the pair in the output set of the composer is $(\nonempty,(q_3,\pr))$),
	also the run descriptors for $s^0, s^1, s^2$ will be used after reading a nonempty word.
	The derivation tree $D_{2,f_1,f_2}$ says also how the lower copy of $s^0$ provides $\sigma_{3,f_2}$; this is described by the derivation tree $D_{3,f_2}$.
	In the order-$1$ assumption set of $\sigma_{2,f_1,f_2}$ we put both assumptions, $(1,\sigma_{4,f_1}^1)$ and $(\nonempty,\sigma_{4,f_2}^1)$;
	the former needs to be provided by the new copy of $s^1$, and the latter by the lower copy of $s^1$.
	Notice that $\sigma_{2,\pr,\pr}$ is productive, because for $f_1=f_2=\pr$ the same (productive) run descriptor $\sigma_{4,\pr}^1$ is used for $s^1$ twice.
	
	Consider also run descriptors 
	\begin{align*}
		\sigma_{1,f}=(q_1,\{(\nonempty,(q_3,\pr))\},\emptyset,f)&&\mbox{for }f\in\{\np,\pr\}\,.
	\end{align*}
	We can derive $\gamma\vdash\sigma_{1,\np}$ by $D_{1,a,a}=(\dtpush\gamma,q_1,D_{2,\np,\np},\{D_{4,a}\})$.
	After performing $\push^1_\gamma$ from $s^2:s^1:s^0$, we obtain $s^2:(s^1:s^0):\posp(s^0)$.
	The derivation tree $D_{2,\np,\np}$ specifies the behavior of the new topmost $0$-stack.
	The composer 
	\begin{align*}
		(\emptyset,\{(1,\sigma_{4,\np}),(\nonempty,\sigma_{4,\np})\};\{(1,\sigma_{4,\np}^1),(\nonempty,\sigma_{4,\np}^1)\};\np)
	\end{align*}
	specifies how the assumptions of $\rd(D_{2,\np,\np})$ (i.e., $\sigma_{2,\np,\np}$) can be realized by the $1$-stack $s^1:s^0$.
	It says that (the lower copy of) $s^0$ should provide $\sigma_{4,\np}$, and the derivation tree $D_{4,a}$ specifies how it is provided.
	The whole $D_{1,a,a}$ corresponds to a run that reads the letter $a$ twice.
	Similarly, $\gamma\vdash\sigma_{1,\np}$ can be derived by $D_{1,b,b}=(\dtpush\gamma,q_1,D_{2,\np,\np},\{D_{4,b}\})$,
	which corresponds to a run that reads the letter $b$ twice.
	On the other hand, there is no derivation tree corresponding to a run that first reads the letter $a$, and then the letter $b$ (or vice versa).
	The reason is that we have to provide exactly one realization of $\sigma_{4,\np}$ for $s^0$; this can be either $D_{4,a}$, or $D_{4,b}$, but not both.
	The statement $\gamma\vdash\sigma_{1,\pr}$ can be derived by $D_{1,\sharp,\sharp}=(\dtpush\gamma,q_1,D_{2,\pr,\pr},\{D_{4,\sharp}\})$,
	by $D_{1,\sharp,a}=(\dtpush\gamma,q_1,D_{2,\pr,\np},\{D_{4,\sharp},D_{4,a}\})$,
	by $D_{1,a,\sharp}=(\dtpush\gamma,q_1,D_{2,\np,\pr},\{D_{4,\sharp},D_{4,a}\})$,
	and by similar derivation trees using the letter $b$ instead of $a$.
	Notice that here we have sets with two derivation trees, $D_{4,\sharp}$ and $D_{4,a}$;
	this is because one provides a realization for $\sigma_{4,\pr}$, and the other for $\sigma_{4,\np}$.

	Finally, consider run descriptors
	\begin{align*}
		\tau_7&=(q_7,\emptyset,\emptyset,\np)\,,\displaybreak[0]\\
		\tau_4&=(q_4,\emptyset,\emptyset,\pr)\,,&&\mbox{and}\displaybreak[0]\\
		\tau_3&=(q_3,\emptyset,\{(1,\tau_4^1)\},\np)&&\mbox{with }\tau_4^1=\red^1(\tau_4)=(q_4,\emptyset,\pr)\,.
	\end{align*}
	Run descriptors $\tau_4$ and $\tau_7$ have empty assumption sets, hence they describe runs that never uncover stacks that are below the topmost $0$-stack
	(in other words, these are runs whose every prefix is $0$-upper).
	We can derive $\gamma\vdash\tau_7$ by $E_7=(\dtempty\gamma,q_7)$, which corresponds to the empty run;
	$\gamma\vdash\tau_4$ by $E_4=(\dtread q_4,E_7)$, which corresponds to the run that reads $\sharp$ and stops;
	and $\gamma\vdash\tau_3$ by $E_3=(\dtpop\gamma,q_3,\tau_4^1)$.
\end{exa}

\subsubsection*{Annotated Stacks}

A derivation tree provides an information about a part of a run $R$ for which a particular $0$-stack is responsible.
In order to describe the whole $R$, we have to specify derivation trees for all $0$-stacks in $R(0)$.
To the end, we annotate stacks using sets of derivation trees.

An \emph{annotated $k$-stack} is a positionless\footnote{%
	\new{It turns out that while considering annotated $k$-stacks we do not need positions, 
	so for notational convenience we assume that annotated stacks are positionless (i.e., their $0$-stacks do not contain positions, conversely to non-annotated stacks).}}
$k$-stack over an extended alphabet, whose elements are pairs $(\gamma,\mathfrak{D})$, where $\gamma\in\Gamma$ 
and $\mathfrak{D}$ is a set of derivation trees having conclusions with the stack symbol $\gamma$, and different run descriptors
(that is, $\rd(D)=\rd(E)$ for $D,E\in\mathfrak{D}$ implies $D=E$).
Annotated stacks are denoted using boldface letters, often with their order written in the superscript: $\sbf^0$, $\tt^5$, etc.
The projection of each letter in an annotated $k$-stack $\sbf^k$ to the $\Gamma$ coordinate is denoted by $\st(\sbf^k)$.

We also define the type of an annotated $k$-stack, which is a subset of $\calT^k$:
\begin{gather*}
		\type((\gamma,\mathfrak{D}))=\{\rd(D)\mid D\in\mathfrak{D}\}\,,\qquad\qquad\type([\,])=\emptyset\,,\displaybreak[0]\\
		\type(\sbf^k:\sbf^{k-1})=\{\red^k(\sigma)\mid\sigma\in\type(\sbf^{k-1})\}\,.
\end{gather*}

We always want to annotate stacks in a consistent way.
Intuitively, when a run descriptor assigned to some stack element requires some assumptions, 
then the part of the stack that is below has to deliver annotations fulfilling these assumptions.
Simultaneously, all annotations have to be useful: 
they cannot provide derivation trees for run descriptors that do not appear as assumptions of annotations assigned higher in the stack.
To formalize this, we define below when an annotated $k$-stack $\sbf^k$ is \emph{well-formed}.

\begin{defi}\lab{def:well-formed}
	Each annotated $0$-stack, and the empty annotated $k$-stack for each $k\geq 1$, are always well-formed.
	An annotated $k$-stack $\sbf^k:\sbf^{k-1}$ is well-formed if both $\sbf^k$ and $\sbf^{k-1}$ are well-formed, 
	and $\type(\sbf^k)=\bigcup_{}\{\pi_2(\ass^k(\sigma))\mid\sigma\in\type(\sbf^{k-1})\}$, and $|\type(\sbf^k:\sbf^{k-1})|=|\type(\sbf^{k-1})|$.
\end{defi}

In the sequel, generally we only consider well-formed annotated stacks (except for some moments when we first define an annotated stack, and then we prove that it is well-formed).

Notice that beside of the condition mentioned earlier (saying that $\type(\sbf^k)$ provides exactly assumptions for run descriptors in $\type(\sbf^{k-1})$),
we also have the second condition, $|\type(\sbf^k:\sbf^{k-1})|=|\type(\sbf^{k-1})|$, 
saying that every run descriptor in $\type(\sbf^k:\sbf^{k-1})$ has exactly one realization by a composition of run descriptors from $\type(\sbf^k)$ and $\type(\sbf^{k-1})$.
Both conditions have the same goal, which is also the same as for analogous conditions in the definition of a composer (Definition~\ref{def:composer}):
we want to estimate the number of $\sharp$ symbols read by a run by looking at the number of productive run descriptors associated to $0$-stacks in an annotated stacks.
Two realizations of the same run descriptor, as well as realizations of a run descriptor not appearing in an assumption set, would bias these calculations.

The definition of types and well-formedness connects only the type of $\sbf^k:\sbf^{k-1}$ with the types of $\sbf^k$ and $\sbf^{k-1}$, 
but similar conditions can be written for a stack of the form $\sbf^k:\sbf^{k-1}:\dots:\sbf^l$.

\begin{prop}\lab{prop:well-formed-multi}
	Let $0\leq l\leq k\leq n$, and let $\sbf=\sbf^k:\sbf^{k-1}:\dots:\sbf^l$ be an annotated $k$-stack in which each $\sbf^i$ is well-formed.
	Then,
	\begin{enumerate}[label=\new{(T\arabic*)}]
	\item\lab{pkt:types-red}
		$\type(\sbf)=\{\red^k(\sigma)\mid\sigma\in\type(\sbf^l)\}$.
	\end{enumerate}
	Moreover, $\sbf$ is well-formed if and only if
	\begin{enumerate}[label=\new{(T\arabic*)}]
	\setcounter{enumi}{1}
	\item\lab{pkt:types-ass}
		$\type(\sbf^i)=\bigcup\{\pi_2(\ass^i(\sigma))\mid\sigma\in\type(\sbf^l)\}$ for every $i\in\prz{l+1}{k}$, 
	\item\lab{pkt:types-inj}
		$|\type(\sbf)|=|\type(\sbf^l)|$.
	\end{enumerate}
\end{prop}

\begin{proof}
	Induction on $k-l$.
	Suppose first that $k-l=0$.
	In this case, Item~\ref{pkt:types-red} holds because $\sbf=\sbf^l$ and because $\red^k(\sigma)=\sigma$ when $\sigma\in\calT^k$.
	Moreover $\sbf$ is well-formed by assumption, Item~\ref{pkt:types-ass} is true because $\prz{l+1}{k}=\emptyset$ is empty,
	and Item~\ref{pkt:types-inj} is true because $\sbf=\sbf^l$; thus we have the equivalence.

	Suppose now that $k-l\geq 1$, and denote $\tt=\sbf^{k-1}:\sbf^{k-2}:\dots:\sbf^l$.
	We then have $\sbf=\sbf^k:\tt$.
	We apply the induction assumption to $\tt=\sbf^{k-1}:\sbf^{k-2}:\dots:\sbf^l$.
	By Item~\ref{pkt:types-red} of the induction assumption, $\type(\tt)=\{\red^{k-1}(\sigma)\mid\sigma\in\type(\sbf^l)\}$,
	and by the definition of the type of $\sbf=\sbf^k:\tt$,
	\begin{align*}
		\type(\sbf)=\{\red^k(\sigma)\mid\sigma\in\type(\tt)\}=\{\red^k(\red^{k-1}(\sigma))\mid\sigma\in\type(\sbf^l)\}\,.
	\end{align*}
	Recalling that $\red^k(\red^{k-1}(\sigma))=\red^k(\sigma)$, we obtain Item~\ref{pkt:types-red}.

	Suppose that $\sbf$ is well-formed.
	Then, in particular, $\tt$ is well-formed, so Items~\ref{pkt:types-ass} and~\ref{pkt:types-inj} hold for the substack $\tt$ by the induction assumption.
	Item~\ref{pkt:types-ass} from the induction assumption is the same as our Item~\ref{pkt:types-ass} for $i\in\prz{l+1}{k-1}$.
	%We need to prove Item~\ref{pkt:types-ass} for $i=k$.
	By well-formedness of $\sbf=\sbf^k:\tt$,
	\begin{align*}
		\type(\sbf^k)=\bigcup\{\pi_2(\ass^k(\sigma))\mid\sigma\in\type(\tt)\}=\bigcup\{\pi_2(\ass^k(\red^{k-1}(\sigma)))\mid\sigma\in\type(\sbf^l)\}\,.
	\end{align*}
	Recalling that $\ass^k(\red^{k-1}(\sigma))=\ass^k(\sigma)$ we obtain Item~\ref{pkt:types-ass} for $i=k$.
	Moreover, Item~\ref{pkt:types-inj} of the induction assumption says that $|\type(\tt)|=|\type(\sbf^l)|$, and by well-formedness of $\sbf=\sbf^k:\tt$ we have that $|\type(\sbf)|=|\type(\tt)|$, 
	thus $|\type(\sbf)|=|\type(\sbf^l)|$ (Item~\ref{pkt:types-inj}).
	
	Conversely, suppose that Items~\ref{pkt:types-ass} and~\ref{pkt:types-inj} hold for $\sbf$.
	Because $\type(\sbf)$ is the image of $\type(\tt)$ under the function $\red^k$, and $\type(\tt)$ is the image of $\type(\sbf^l)$ under the function $\red^{k-1}$,
	we necessarily have $|\type(\sbf)|\leq|\type(\tt)|\leq|\type(\sbf^l)|$, and thus $|\type(\sbf)|=|\type(\sbf^l)|$ (Item~\ref{pkt:types-inj})
	implies that $|\type(\sbf)|=|\type(\tt)|=|\type(\sbf^l)|$; we thus have Item~\ref{pkt:types-inj} for $\tt$.
	Moreover, Item~\ref{pkt:types-ass} for $\tt$ is a direct consequence of this item for $\sbf$ (we only restrict the considered orders $i$ to $\prz{l+1}{k-1}$).
	Thus, by the induction assumption, $\tt$ is well-formed.
	Using Item~\ref{pkt:types-ass} for $i=k$, the equality $\ass^k(\sigma)=\ass^k(\red^{k-1}(\sigma))$, and Item~\ref{pkt:types-red} for $\tt$ we obtain that
	\begin{align*}
		\type(\sbf^k)&=\bigcup\{\pi_2(\ass^k(\sigma))\mid\sigma\in\type(\sbf^l)\}\\
			&=\bigcup\{\pi_2(\ass^k(\red^{k-1}(\sigma)))\mid\sigma\in\type(\sbf^l)\}
			=\bigcup\{\pi_2(\ass^k(\sigma))\mid\sigma\in\type(\tt)\}\,.
	\end{align*}
	Together with the equality $|\type(\sbf)|=|\type(\tt)|$ this implies that $\sbf=\sbf^k:\tt$ is well-formed.
\end{proof}

An annotated stack $\sbf$ is called \emph{singular} if $|\type(\sbf)|=1$.
When an annotated $n$-stack $\sbf^n$ is singular, we define $\conf(\sbf^n)$ to be the configuration $(q,\poslinv(\st(\sbf^n)))$, 
where $q$ is the state of the only run descriptor in $\type(\sbf^n)$.

As the type of a configuration $c$, denoted $\type_{\calA,\phi}(c)$, we take the union of $\type(\tops^0(\sbf^n))$ over all well-formed singular annotated $n$-stacks $\sbf^n$ such that $\conf(\sbf^n)=c$,
\begin{align*}
	\type_{\calA,\phi}(c)=\bigcup\{\type(\tops^0(\sbf^n))\mid\sbf^n\mbox{ well-formed, } \conf(\sbf^n)=c\}\,.
\end{align*}

We remark that in the union we could also allow well-formed annotated stacks $\sbf^n$ that are not necessarily singular, 
but are such that $\poslinv(\st(\sbf^n))$ is the stack of $c$ and the state of all run descriptors in $\type(\sbf^n)$ is the state of $c$.
%(it is more convenient not to include them in the definition, though).
On the other hand, in general there does not exist a single well-formed annotated stack $\sbf^n$ such that $\type(\tops^0(\sbf^n))=\type_{\calA,\phi}(c)$. % and $\st(\sbf^n)$ is the stack of $c$.
Namely, we can have a situation like in Example~\ref{ex:composer}, where we cannot assign both $\sigma_1$ and $\sigma_1'$ to the topmost $0$-stack,
as both of them result in the same run descriptor $\xi_1$ for the topmost $1$-stack.

Actually, we see a direct connection between the well-formedness property and composers.

\begin{prop}\lab{prop:composer}
	Let $0\leq l\leq k\leq n$, let $\sbf=\sbf^k:\sbf^{k-1}:\dots:\sbf^l$ be an annotated $k$-stack in which each $\sbf^i$ is well-formed, and let $\Psi^k\subseteq M\times\calT^k$.
	The following two conditions are equivalent:
	\begin{itemize}
	\item	there exists a composer $(\Phi^k,\Phi^{k-1},\dots,\Phi^l;\Psi^k;f)$ such that $\pi_2(\Phi^i)=\type(\sbf^i)$ for each $i\in\prz{l}{k}$, and
	\item	$\sbf$ is well-formed and $\pi_2(\Psi^k)=\type(\sbf)$.
	\end{itemize}
\end{prop}

\begin{proof}
	Suppose first that we have a composer $(\Phi^k,\Phi^{k-1},\dots,\Phi^l;\Psi^k;f)$ such that 
	\begin{align}
		\pi_2(\Phi^i)&=\type(\sbf^i)&&\mbox{for }i\in\prz{k}{l}\,.\label{eq:7.13.1}
	\end{align}
	Notice that
	\begin{align}\label{eq:7.13.2}
		\type(\sbf)=\{\red^k(\sigma)\mid\sigma\in\type(\sbf^l)\}=\{\red^k(\sigma)\mid\sigma\in\pi_2(\Phi^l)\}=\pi_2(\Psi^k)\,,
	\end{align}
	where the consecutive equalities follow from Item~\ref{pkt:types-red} of Proposition~\ref{prop:well-formed-multi}, from Equality~\eqref{eq:7.13.1}, 
	and from Condition~\ref{pkt:compo-red} of Definition~\ref{def:composer}, respectively.
	Moreover, for $i\in\prz{l+1}{k}$,
	\begin{align}\label{eq:7.13.3}
		\type(\sbf^i)=\pi_2(\Phi^i)=\bigcup\{\pi_2(\ass^i(\sigma))\mid\sigma\in\pi_2(\Phi^l)\}=\bigcup\{\pi_2(\ass^i(\sigma))\mid\sigma\in\type(\sbf^l)\}\,,
	\end{align}
	where the equalities are consequences of Equality~\eqref{eq:7.13.1}, of Condition~\ref{pkt:compo-ass} of Definition~\ref{def:composer}, and of Equality~\eqref{eq:7.13.1} again.
	%This gives Item~\ref{pkt:types-ass} of Proposition~\ref{prop:well-formed-multi}.
	Simultaneously,
	\begin{align}\label{eq:7.13.4}
		|\type(\sbf)|=|\pi_2(\Psi^k)|=|\pi_2(\Phi^l)|=|\type(\sbf^l)|\,,
	\end{align}
	where the consecutive equalities follow from Equality~\eqref{eq:7.13.2}, from Condition~\ref{pkt:compo-inj} of Definition~\ref{def:composer}, and from Equality~\eqref{eq:7.13.1}, respectively.
	Equalities~\eqref{eq:7.13.3} and~\eqref{eq:7.13.4} give Items~\ref{pkt:types-ass} and~\ref{pkt:types-inj} of Proposition~\ref{prop:well-formed-multi}, 
	which implies that $\sbf$ is well-formed; together with Equality~\eqref{eq:7.13.2} this gives the thesis.
	
	Conversely, suppose that $\sbf$ is well-formed and 
	\begin{align}\label{eq:7.13.5}
		\pi_2(\Psi^k)=\type(\sbf)\,.
	\end{align}
	We define
	\begin{align}
		\Phi^l&=\{(m,\sigma)\mid\sigma\in\type(\sbf^l)\land(m,\red^k(\sigma))\in\Psi^k\}\,,&&\mbox{and}\label{eq:7.13.6}\displaybreak[0]\\
		\Phi^i&=\bigcup\{m\circ\ass^i(\sigma)\mid(m,\sigma)\in\Phi^l\}\qquad\mbox{for }i\in\prz{l+1}{k}\,.\label{eq:7.13.7}
	\end{align}
	By Equality~\eqref{eq:7.13.5} and by Item~\ref{pkt:types-red} of Proposition~\ref{prop:well-formed-multi} we have that
	\begin{align}\label{eq:7.13.8}
		\pi_2(\Psi^k)=\type(\sbf)=\{\red^k(\sigma)\mid\sigma\in\type(\sbf^l)\}\,.
	\end{align}
	This means that for every $\sigma\in\type(\sbf^l)$ there is some $m$ such that $(m,\red^k(\sigma))\in\Psi^k$.
	In the light of Equality~\eqref{eq:7.13.6} this implies that 
	\begin{align}\label{eq:7.13.10}
		\pi_2(\Phi^l)=\type(\sbf^l)\,.
	\end{align}
	On the other hand, by Equality~\eqref{eq:7.13.7} and by Item~\ref{pkt:types-ass} of Proposition~\ref{prop:well-formed-multi}, for $i\in\prz{l+1}{k}$, we have that
	\begin{align}\label{eq:7.13.11}
		\pi_2(\Phi^i)=\bigcup\{\pi_2(\ass^i(\sigma))\mid\sigma\in\pi_2(\Phi^l)\}=\type(\sbf^i)\,.
	\end{align}
	Let us now check particular conditions of Definition~\ref{def:composer}.
	Condition~\ref{pkt:compo-ass} is immediate from Equality~\eqref{eq:7.13.7}.
	By Equality~\eqref{eq:7.13.6}, for every pair $(m,\sigma)\in\Phi^l$, the pair $(m,\red^k(\sigma))$ is in $\Psi^k$.
	Conversely, by Equality~\eqref{eq:7.13.8}, every pair in $\Psi^k$ is of the form $(m,\red^k(\sigma))$ with $\sigma\in\type(\sbf^l)$, hence $(m,\sigma)\in\Phi^l$ by Equality~\eqref{eq:7.13.6}.
	This implies Condition~\ref{pkt:compo-red}.
	Using consecutively Equality~\eqref{eq:7.13.5}, Item~\ref{pkt:types-inj} of Proposition~\ref{prop:well-formed-multi}, and Equality~\eqref{eq:7.13.10}, we obtain Condition~\ref{pkt:compo-inj}:
	\begin{align*}
		|\pi_2(\Psi^k)|=|\type(\sbf)|=|\type(\sbf^l)|=|\pi_2(\Phi^l)|\,.
	\end{align*}
	Condition~\ref{pkt:compo-flag} always holds for some $f\in\{\np,\pr\}$.
	Thus $(\Phi^k,\Phi^{k-1},\dots,\Phi^l;\Psi^k;f)$ is a composer, which together with Equalities~\eqref{eq:7.13.10} and~\eqref{eq:7.13.11} gives the thesis.	
\end{proof}

\begin{exa}\lab{exa:types}
	This is a continuation of Example~\ref{exa:der-tree}.
	With the derivation trees considered there,
	\begin{align*}
		\sbf_1&=[[(\gamma,\{E_4\}),(\gamma,\{E_3\})],[(\gamma,\emptyset),(\gamma,\{D_{1,\sharp,a}\})]]&&\mbox{and}\displaybreak[0]\\
		\sbf_2&=[[(\gamma,\{E_4\}),(\gamma,\{E_3\})],[(\gamma,\emptyset),(\gamma,\{D_{4,\sharp},D_{4,a}\}),(\gamma,\{D_{2,\pr,\np}\})]]
	\end{align*}
	are well-formed (singular) annotated $2$-stacks.
	On the other hand, the following annotated $2$-stacks are not well-formed:
	\begin{align*}
		\sbf_3&=[[(\gamma,\{E_4\}),(\gamma,\{E_3\})],[(\gamma,\{E_1\}),(\gamma,\{D_{1,\sharp,a}\})]]\,,\displaybreak[0]\\
		\sbf_4&=[[(\gamma,\{E_4\}),(\gamma,\{E_3\})],[(\gamma,\emptyset),(\gamma,\{D_{4,\sharp},D_{4,a},D_{4,b}\}),(\gamma,\{D_{2,\pr,\np}\})]]\,,&&\mbox{and}\displaybreak[0]\\
		\sbf_5&=[[(\gamma,\{E_4\}),(\gamma,\{E_3\})],[(\gamma,\emptyset),(\gamma,\{D_{4,\sharp}\}),(\gamma,\{D_{2,\pr,\np}\})]]\,.
	\end{align*}
	In $\sbf_3$, we provide a spare derivation tree $E_1$, not needed by the derivation trees assigned above.
	In $\sbf_4$, two derivation trees, $D_{4,a}$ and $D_{4,b}$, provide the same run descriptor.
	Finally, in $\sbf_5$, we are missing a derivation tree that would provide $\sigma_{4,\np}$.
	
\end{exa}

When $\widetilde\Psi$ is a subset of the type of a well-formed annotated $k$-stack $\sbf$, 
we can remove some of the annotations in $\sbf$ in order to obtain a well-formed annotated $k$-stack $\sbf{\restriction}_{\widetilde\Psi}$ whose type is $\widetilde\Psi$.
We do this by induction:
\begin{itemize}
\item	For $k=0$, we restrict the set of derivation trees in $\sbf$ to those trees whose run descriptor is in $\widetilde\Psi$.
\item	The type of the empty stack is empty, so we need not to restrict it in any way.
\item	For $\sbf=\sbf^k:\sbf^{k-1}$, we restrict $\sbf^{k-1}$ to the set $\widetilde\Phi$ containing those $\sigma\in\type(\sbf^{k-1})$
	for which $\red^k(\sigma)\in\widetilde\Psi$, and we restrict $\sbf^k$ to $\bigcup_{\sigma\in\widetilde\Phi}\pi_2(\ass^k(\sigma))$.
\end{itemize}

\subsubsection*{Plan for the Remaining Part of the Section}

We have already defined types of configurations, as needed for Theorem~\ref{thm:types}.
Type of a configuration $c$ is defined via existence of annotated stacks $\sbf$ such that $\conf(\sbf)=c$.
Theorem~\ref{thm:types} says that if two configurations $c,d$ have the same type, and we have a run starting in $c$, then a similar run starts $d$.
Roughly, the strategy of the proof is as follows.
First, basing on the run starting in $c$, we construct an annotated stack $\sbf$ with $\conf(\sbf)=c$, corresponding to this run.
More precisely, we do not process the whole run in this way, but rather its particular fragments.
Then, because the types of $c$ and $d$ equal, there exists an annotated stack $\tt$ with $\conf(\tt)=d$, and with $\rd(\tt)=\rd(\sbf)$.
Having $\tt$ we proceed in the opposite direction: basing on the annotated stack $\tt$ we construct a run starting in $d$, satisfying the thesis of the theorem.

Recall that an annotated stack is, roughly, a description of a run.
The run described by an annotated stack is called \emph{annotated run}, and is defined in Subsection~\ref{sec:annotated-run}.
In Lemma~\ref{lem:rd2run} (located in Subsection~\ref{sec:assumptions-used}) we prove that annotated runs have the expected form, 
that is, that using an assumption of a run descriptor corresponds to performing a return.
In Subsection~\ref{sec:completeness} we present the opposite direction: how to construct an annotated stack basing on a run.
The proof of Theorem~\ref{thm:types} is finalized in Subsection~\ref{sec:copy-upper}.

Simultaneously, we prepare ourselves for a proof of Theorem~\ref{thm:stypes}, which additionally talks about the number of $\sharp$ symbols read by a run.
We thus need to estimate the number of $\sharp$ symbols read by an annotated run.
To this end, in Subsection~\ref{sec:annotated-run} we define two numbers corresponding to an annotated stack $\sbf$, namely $\low(\sbf)$ and $\high(\sbf)$.
They provide a lower bound and, respectively, an upper bound for the number of $\sharp$ symbols read by an annotated run starting in $\sbf$, as showed in Lemma~\ref{lem:low-high-len}.
(We also define there a third number, $\len(\sbf)$, and we prove that it gives an upper bound for the length of an annotated run starting in $\sbf$.
Its role is auxiliary: it is only needed for showing that the constructed annotated run is finite.
Actually, this is needed already while proving Theorem~\ref{thm:types}.)
The essential property is that $\high(\sbf)$ can be bounded by a function of $\low(\sbf)$, as shown in Proposition~\ref{prop:common-bound}, contained in Subsection~\ref{sec:common-bound}.
Thus $\low(\sbf)$ itself estimates the number of $\sharp$ symbols read by an annotated run starting in $\sbf$.
With such a function $\low$ we can define sequence-equivalence, as needed in Theorem~\ref{thm:stypes}.
Namely, having a sequence of configurations (all of the same type), for every run descriptor $\sigma$ it is enough to know one think:
whether there is a sequence of annotated stacks corresponding to these configurations (and to the run descriptor $\sigma$), 
and such that the values of $\low$ for these annotated stacks are bounded.
If this is the case, using these annotated stacks we can reproduce runs that read a bounded number of $\sharp$ symbol.
If not, then a sequence of annotated stacks corresponding to the considered configurations also exists, but the values of $\low$ for these annotated stacks are unbounded, 
and in effect the reproduced runs read an unbounded number of $\sharp$ symbols.
A proof of Theorem~\ref{thm:stypes}, following these ideas, is given in Subsection~\ref{sec:seq-equiv}.

\subsection{Annotated Runs}\lab{sec:annotated-run}

In this subsection we describe how, having an annotated stack, we can reproduce a run.
This is formalized in the notion of annotated runs.
We also relate the number of $\sharp$ symbols read by a run with the number of productive run descriptors in annotations of the starting configuration.
Later, in Subsection~\ref{sec:completeness}, we do the converse: we show how to construct an annotated stack basing on a run.
Recall that when a well-formed annotated stack is singular, then its topmost $0$-stack is singular as well (cf.~Item~\ref{pkt:types-inj} of Proposition~\ref{prop:well-formed-multi}).

\begin{defi}\lab{def:successor}
	Let $\sbf=\sbf^n:\sbf^{n-1}:\dots:\sbf^0$ be a well-formed singular annotated $n$-stack, where $\sbf^0=(\gamma,\{D\})$.
	We define the \emph{successor} of $\sbf$.
	\begin{enumerate}
	\item\lab{pkt:succ-stop}
		If $D=(\dtempty\gamma,p)$, then $\sbf$ has no successor.
	\item	If $D=(\dtread p,D')$, then the successor is $\sbf^n:\sbf^{n-1}:\dots:\sbf^1:(\gamma,\{D'\})$.
	\item	If $D=(\dtpop\gamma,p,\tau^k)$, then the successor is $\sbf^n:\sbf^{n-1}:\dots:\sbf^k$.
	\item\lab{pkt:succ-push}
		Suppose that $D=(\dtpush\gamma,p,D',\mathfrak{D})$.
		Let $\alpha$, $k$, $\Psi^i$, $\Phi^i$ be as in Definition~\ref{def:types}(\ref{pkt:dt-push}).
		In this situation, the successor of $\sbf$ is
		\begin{align*}
			&\sbf^n:\sbf^{n-1}:\dots:\sbf^{k+1}:\tt^k:\sbf^{k-1}{\restriction}_{\pi_2(\Psi^{k-1})}:\sbf^{k-2}{\restriction}_{\pi_2(\Psi^{k-2})}:\dots:\sbf^1{\restriction}_{\pi_2(\Psi^1)}:(\alpha,\{D'\})\,,
		\end{align*}
		where $\tt^k=\sbf^k{\restriction}_{\pi_2(\Phi^k)}:\sbf^{k-1}{\restriction}_{\pi_2(\Phi^{k-1})}:\dots:\sbf^1{\restriction}_{\pi_2(\Phi^1)}:(\gamma,\mathfrak{D})$.
	\end{enumerate}
\end{defi}

We notice that in Case~(\ref{pkt:succ-push}) for $i\in\prz{1}{k}$ we have that $\pi_2(\Phi^i)\subseteq\type(\sbf^i)$,
and for $i\in\prz{1}{k-1}$ we have that $\pi_2(\Psi^i)\subseteq\type(\sbf^i)$, and thus the restrictions are legal.
Indeed, from Proposition~\ref{prop:well-formed-multi} applied to $\sbf$ we know that $\type(\sbf^i)=\pi_2(\ass^i(\rd(D)))$ for all $i\in\prz{1}{n}$;
moreover, by Definition~\ref{def:types}(\ref{pkt:dt-push}), $\ass^i(\sigma)=\Psi^i\cup\Phi^i$ for $i\in\prz{1}{k-1}$, and $\ass^k(\sigma)=\Phi^k$.

\begin{prop}%
	Let $\tt$ be 
	the successor of a well-formed singular annotated $n$-stack $\sbf$.
	Then $\tt$ is singular, well-formed, and $\conf(\tt)$ is a successor of $\conf(\sbf)$ (in the considered automaton).
\end{prop}

\begin{proof}
	Let $\sbf=\sbf^n:\sbf^{n-1}:\dots:\sbf^0$, and $\sbf^0=(\gamma,\{D\})$, and $\sigma=\rd(D)$.
	All $\sbf^i$ are well-formed, because $\sbf$ is well-formed.
	Moreover, by Proposition~\ref{prop:well-formed-multi}, $\type(\sbf^i)=\pi_2(\ass^i(\sigma))$ for all $i\in\prz{1}{n}$.
	
	We have several cases depending on the shape of $\sbf$.
	We cannot have $D=(\dtempty\gamma,p)$, as then $\sbf$ has no successor.
	
	Suppose that $D=(\dtread p,D')$.
	Recall from Definition~\ref{def:types} that $\pi_2(\ass^i(\rd(D'))=\pi_2(\ass^i(\sigma))$ for all $i\in\prz{1}{n}$.
	By Proposition~\ref{prop:well-formed-multi}, $\tt$ is singular (its $\type$ is a singleton $\{\red^n(\rd(D'))\}$) and well-formed.
	Moreover, $\delta(\gamma,p)=\read(\vec{q})$, where $\vec{q}(a)$ equals the state of $\rd(D')$ for some $a\in A$.
	The transition from $\conf(\sbf)$ reading this $a$ leads to $\conf(\tt)$.
	
	Next, suppose that $D=(\dtpop\gamma,p,\tau^k)$.
	By Definition~\ref{def:types} we have that $\pi_2(\ass^k(\sigma))=\{\tau^k\}$ (hence $\type(\sbf^k)=\{\tau^k\}$)
	and $\ass^i(\sigma)=\ass^i(\tau^k)$ (hence $\type(\sbf^i)=\pi_2(\ass^i(\tau^k))$) for $i\in\prz{k+1}{n}$.
	Proposition~\ref{prop:well-formed-multi} applied to $\tt=\sbf^n:\sbf^{n-1}:\dots:\sbf^k$ implies that it is singular and well-formed.
	Moreover, $\delta(\gamma,p)=(q,\pop^k)$, where $q$ is the state of $\tau^k$, so the transition from $\conf(\sbf)$ leads to $\conf(\tt)$.

	Finally, suppose that $D=(\dtpush\gamma,p,D',\mathfrak{D})$.
	Let $\alpha$, $k$, $\Psi^i$, $\Phi^i$, $f$ be as in Definition~\ref{def:types}(\ref{pkt:dt-push}), and $\tt^k$ as in Definition~\ref{def:successor}(\ref{pkt:succ-push}).
	Obviously, $\type(\sbf^i{\restriction}_{\pi_2(\Phi^i)})=\pi_2(\Phi^i)$, for all $i\in\prz{1}{k}$;
	moreover, $\type((\gamma,\mathfrak{D}))=\{\rd(E)\mid E\in\mathfrak{D}\}=\pi_2(\Phi^0)$, by Definition~\ref{def:types}(\ref{pkt:dt-push}).
	Furthermore, again by this definition, $(\Phi^k,\Phi^{k-1},\dots,\Phi^0;\Psi^k;f)$ is a composer.
	It follows from Proposition~\ref{prop:composer} that $\tt^k$ is well-formed, and $\type(\tt^k)=\pi_2(\Psi^k)$
	(recall that $\tt^k=\sbf^k{\restriction}_{\pi_2(\Phi^k)}:\sbf^{k-1}{\restriction}_{\pi_2(\Phi^{k-1})}:\dots:\sbf^1{\restriction}_{\pi_2(\Phi^1)}:(\gamma,\mathfrak{D})$).
	By Definition~\ref{def:types}(\ref{pkt:dt-push}), for all $i\in\prz{1}{n}$ we have that $\Psi^i=\ass^i(\rd(D'))$.
	Moreover,
	\begin{itemize}
	\item	for $i\in\prz{k+1}{n}$, we have $\ass^i(\sigma)=\Psi^i$ by Definition~\ref{def:types}(\ref{pkt:dt-push}), so
		$\type(\sbf^i)=\pi_2(\ass^i(\sigma))=\pi_2(\Psi^i)=\pi_2(\ass^i(\rd(D')))$;
	\item	$\type(\tt^k)=\pi_2(\Psi^k)=\pi_2(\ass^k(\rd(D'))$;
	\item	for $i\in\prz{1}{k-1}$, we have that $\type(\sbf^i{\restriction}_{\pi_2(\Psi^i)})=\pi_2(\Psi^i)=\pi_2(\ass^i(\rd(D'))$.
	\end{itemize}
	In effect, $\tt$, which is a composition of these stacks, and of $(\alpha,\{D'\})$ is singular and well-formed by Proposition~\ref{prop:well-formed-multi}.
	Additionally, $\delta(\gamma,p)=(q,\push^k_\alpha)$, where $q$ is the state of $\rd(D')$, so the transition from $\conf(\sbf)$ leads to $\conf(\tt)$.
\end{proof}

An \emph{annotated run} $\RR$ is a sequence $\sbf_0,\dots,\sbf_m$ of well-formed singular $n$-stacks in which $\sbf_i$ is the successor of $\sbf_{i-1}$ for each $i\in\prz{1}{m}$.
By replacing each $\sbf_i$ by $\conf(\sbf_i)$ we obtain a run denoted $\st(\RR)$.

Notice that an annotated stack $\sbf$ may have less successors than $\conf(\sbf)$.
Indeed, in the case of $D=(\dtempty\gamma,p)$ there are no successors of $\sbf$, but $\conf(\sbf)$ may have successors.
Similarly, in the case of $D=(\dtread p,D')$ there is exactly one successor of $\sbf$ (the state in $\rd(D')$ determines which letter should be read),
while in a run from $\conf(\sbf)$ we can read any letter.

\begin{exa}
	Recall the $2$-DPDA $\calA_1$ from Example~\ref{exa:der-tree},
	and the annotated stack $\sbf_1=[[(\gamma,\{E_4\}),(\gamma,\{E_3\})],[(\gamma,\emptyset),(\gamma,\{D_{1,\sharp,a}\})]]$ from Example~\ref{exa:types}.
	The successors of $\sbf_1$ are, consecutively,
	\begin{align*}
		&[[(\gamma,\{E_4\}),(\gamma,\{E_3\})],[(\gamma,\emptyset),(\gamma,\{D_{4,\sharp},D_{4,a}\}),(\gamma,\{D_{2,\pr,\np}\})]]\,,\displaybreak[0]\\
		&[[(\gamma,\{E_4\}),(\gamma,\{E_3\})],[(\gamma,\emptyset),(\gamma,\{D_{4,a}\}),(\gamma,\{D_{3,\np}\})],[(\gamma,\emptyset),(\gamma,\{D_{4,\sharp}\}),(\gamma,\{D_{3,\pr}\})]]\,,\displaybreak[0]\\
		&[[(\gamma,\{E_4\}),(\gamma,\{E_3\})],[(\gamma,\emptyset),(\gamma,\{D_{4,a}\}),(\gamma,\{D_{3,\np}\})],[(\gamma,\emptyset),(\gamma,\{D_{4,\sharp}\})]]\,,\displaybreak[0]\\
		&[[(\gamma,\{E_4\}),(\gamma,\{E_3\})],[(\gamma,\emptyset),(\gamma,\{D_{4,a}\}),(\gamma,\{D_{3,\np}\})],[(\gamma,\emptyset),(\gamma,\{D_7\})]]\,,\displaybreak[0]\\
		&[[(\gamma,\{E_4\}),(\gamma,\{E_3\})],[(\gamma,\emptyset),(\gamma,\{D_{4,a}\}),(\gamma,\{D_{3,\np}\})]]\,,\displaybreak[0]\\
		&[[(\gamma,\{E_4\}),(\gamma,\{E_3\})],[(\gamma,\emptyset),(\gamma,\{D_{4,a}\})]]\,,\displaybreak[0]\\
		&[[(\gamma,\{E_4\}),(\gamma,\{E_3\})],[(\gamma,\emptyset),(\gamma,\{D_5\})]]\,,\displaybreak[0]\\
		&[[(\gamma,\{E_4\}),(\gamma,\{E_3\})]]\,,\displaybreak[0]\\
		&[[(\gamma,\{E_4\})]]\,,\displaybreak[0]\\
		&[[(\gamma,\{E_7\})]]\,;
	\end{align*}
	the last of them has no more successors.
	In the transition between the first and the second line, $D_{2,\pr,\np}$ says that the new topmost $0$-stack should be annotated by $D_{3,\pr}$, 
	and that the previously topmost $0$-stack should be annotated by $D_{3,\np}$.
	Because $\ass^1(\rd(D_{3,\pr}))=\sigma_{4,\pr}^1$ and $\ass^1(\rd(D_{3,\np}))=\sigma_{4,\np}^1$, 
	we know that $D_{4,\sharp}$ should be taken to the topmost $1$-stack (we have that $\red^1(\rd(D_{4,\sharp}))=\sigma_{4,\pr}^1$),
	and that $D_{4,a}$ should be left in the second topmost $1$-stack (we have that $\red^1(\rd(D_{4,a}))=\sigma_{4,\np}^1$).

	We can see that not every run is of the form $\st(\RR)$ for some annotated run $\RR$.
	For example, this is the case for the run that starts in $(q_1,\poslinv([[\gamma,\gamma],[\gamma,\gamma]]))$ and reads $a$, then $b$, and then $\sharp$.
	Indeed, in order to obtain such a run as an annotated run, to the topmost $0$-stack we have to assign a hypothetical derivation tree $D_{1,a,b}$, saying that we should first read $a$ and then $b$, 
	but there is no such derivation tree (as already explained in Example~\ref{exa:der-tree}).
	Another run that is not of the form $\st(\RR)$ for any annotated run $\RR$ is the run that starts in $(q_2,\poslinv([[\gamma,\gamma],[\gamma,\gamma,\gamma]]))$ and reads $a$, then $b$, and then $\sharp$.
	This time the problem is that to the second topmost $0$-stack we cannot assign simultaneously $D_{4,a}$ and $D_{4,b}$, as they both have the same run descriptor.
	%(on the other hand, if this was allowed, it would not be specified in which order the derivation trees $D_{4,a}$ and $D_{4,b}$ should be used).
	More generally, the $\push$ in Case~(\ref{pkt:succ-push}) of Definition~\ref{def:successor} leaves the same annotations in the original substack as in the copied substack, up to a restriction, 
	which causes that fragments of the annotated run corresponding to these substacks have to be the same.
\end{exa}

A priori there might exist an infinite annotated run.
But, as we see below, this is impossible: 
always after some number of steps we reach an annotated stack with no successors (Case~(\ref{pkt:succ-stop}) of Definition~\ref{def:successor}).
Moreover, we show that the number of $\sharp$ symbols read by the run starting in an annotated stack $\sbf$ can be estimated by the number of productive run descriptors in the annotations of $\sbf$.
To this end, to each well-formed annotated stack $\sbf$ we assign three natural numbers: $\low(\sbf)$, $\high(\sbf)$, and $\len(\sbf)$.
The first two of them give a lower and an upper bound on the number of $\sharp$ symbols read by our run, and the last one gives an upper bound on the length of the run.

\begin{defi}\lab{def:exp}
	For positive integers $m_1,\dots,m_k$ we define $\pow(m_1,\dots,m_k)$ by induction on $k$:
	\begin{align*}
		&\pow()=1\,,&&\mbox{and}&&
		\pow(m_1,m_2,\dots,m_k)=(1+m_1)^{\pow(m_2,\dots,m_k)}-1\,.
	\end{align*}
\end{defi}

Notice that, in particular, $\pow(m_1)=m_1$ and $\pow(m_1,m_2)=(1+m_1)^{m_2}-1$.

\begin{defi}\lab{def:low-high}
	For a well-formed annotated $k$-stack $\sbf$ we define natural numbers $\low(\sbf)$, $\high(\sbf)$, and $\len(\sbf)$ by induction on the structure of $\sbf$.
	\begin{itemize}
	\item	If $\sbf=(\gamma,\mathfrak{D})$, we take
		\begin{align*}
			\low(\sbf)&=|\type(\sbf)\cap\calT_\pr|\,,\displaybreak[0]\\
			\high(\sbf)&=\prod_{D\in\mathfrak{D}\mid\rd(D)\in\calT_\pr}C_{\depth(D)}\,,&&\mbox{and}\displaybreak[0]\\
			\len(\sbf)&=\prod_{D\in\mathfrak{D}}C_{\depth(D)}\,,
		\end{align*}
		where $C_z$ is defined inductively:
		\begin{align*}
			C_0=2\,,&&\mbox{and}&& C_{z+1}=(2|\calT^0|)^n\cdot(C_z)^{|\calT^0|+1}\,.
		\end{align*}
	\item	We take $\low([\,])=0$ and $\high([\,])=\len([\,])=1$.
	\item	If $\sbf=\sbf^k:\sbf^{k-1}$, we take
		\begin{align*}
			\low(\sbf)&=\sum_{\sigma\in\type(\sbf^{k-1})}\big(\low(\sbf^k{\restriction}_{\pi_2(\ass^k(\sigma))})+\low(\sbf^{k-1}{\restriction}_{\{\sigma\}})\big)\,,\displaybreak[0]\\
			\high(\sbf)&=\prod_{\sigma\in\type(\sbf^{k-1})}\pow\big(\high(\sbf^k{\restriction}_{\pi_2(\ass^k(\sigma))}),\high(\sbf^{k-1}{\restriction}_{\{\sigma\}})\big)\,,&&\mbox{and}\displaybreak[0]\\
			\len(\sbf)&=\prod_{\sigma\in\type(\sbf^{k-1})}\pow\big(\len(\sbf^k{\restriction}_{\pi_2(\ass^k(\sigma))}),\len(\sbf^{k-1}{\restriction}_{\{\sigma\}})\big)\,.
		\end{align*}
	\end{itemize}
\end{defi}

The three numbers are interesting for us, because of the following lemma.
Recall that, for a run $R$, by $\sharp(R)$ we denote the number of $\sharp$ symbols read by $R$.

\begin{lem}\lab{lem:low-high-len}
	If $\RR$ is an annotated run,
	\begin{align*}
			&\low(\RR(0))\leq\sharp(\st(\RR))+\low(\RR(|\RR|))\,,\displaybreak[0]\\
			&\high(\RR(0))\geq\sharp(\st(\RR))+\high(\RR(|\RR|))\,,&&\mbox{and}\displaybreak[0]\\
			&\len(\RR(0))\geq|\RR|+\len(\RR(|\RR|))\,.
	\end{align*}
\end{lem}

This is one of key lemmas of this section.
We now give some examples and intuitions staying behind this lemma, and behind the definitions of $\low$, $\high$, and $\len$;
after that, we prove this lemma.

We see that the last inequality of Lemma~\ref{lem:low-high-len} bounds the length of an annotated run $\RR$ by $\len(\RR(0))$, 
that is, by a function of the annotated stack $\RR(0)$, in which the annotated run starts.
Similarly, the second inequality bounds the number of $\sharp$ symbols read by $\RR$, by another function of $\RR(0)$, namely by $\high(\RR(0))$.
The additional components added on the right of these inequalities only strengthen them.
Conversely, the role of the first inequality is to give a lower bound for the number of $\sharp$ symbols read by $\RR$.
If $\RR$ is maximal (i.e., cannot be prolonged), 
then the topmost $0$-stack of $\RR(\RR)$ is annotated by a derivation tree of the form $(\dtempty\gamma,q)$, and all other $0$-stacks are annotated by empty sets.
In effect $\low(\RR(|\RR|))=0$, and we simply obtain that $\low(\RR(0))\leq\sharp(\st(\RR))$.
But for an arbitrary run $\RR$, which may end prematurely, even before reading any $\sharp$ symbol, we have to add $\low(\RR(|\RR|))$ on the right side of the inequality.

Roughly speaking, $\low(\RR(0))$ counts the number of productive run descriptors in the annotations of $\RR(0)$.
The intuition is that every productive run descriptor is responsible for increasing the number of $\sharp$ symbols read, 
so the first inequality of Lemma~\ref{lem:low-high-len} should hold with such a definition of $\low$.

We see that the function $\high$ takes into account the same productive run descriptors as $\low$, 
but instead of sums we use products and the $\pow$ function.
Indeed, it is shown in Proposition~\ref{prop:pr-positive} that if all run descriptors in the annotations of a substack $\sbf^{k-1}$ are nonproductive, then $\high(\sbf^{k-1})=1$ (and $\low(\sbf^{k-1})=0$).
Suppose that $\sbf^{k-1}$ is singular, that is, its type is a singleton $\{\sigma\}$.
If $\sbf^k:\sbf^{k-1}$ is well-formed, we have $\type(\sbf^k)=\pi_2(\ass^k(\sigma))$.
Observe that $\pow(x,1)=(1+x)^1-1=x$ for every $x$, and thus 
\begin{align*}
	\high(\sbf^k:\sbf^{k-1})&=\pow(\high(\sbf^k{\restriction}_{\pi_2(\ass^k(\sigma))}),\high(\sbf^{k-1}{\restriction}_{\{\sigma\}}))\\
		&=\pow(\high(\sbf^k),\high(\sbf^{k-1}))
		=\pow(\high(\sbf^k),1)
		=\high(\sbf^k)\,.
\end{align*}
This is similar to the behavior of the $\low$ function, as in such a case we also have
\begin{align*}
	\low(\sbf^k:\sbf^{k-1})&=\low(\sbf^k{\restriction}_{\pi_2(\ass^k(\sigma))})+\low(\sbf^{k-1}{\restriction}_{\{\sigma\}})\\
		&=\low(\sbf^k)+\low(\sbf^{k-1})=\low(\sbf^k)+0=\low(\sbf^k)\,.
\end{align*}

Technical details of the definition of $\high$ were chosen so that it is possible to perform a proof, but two facts are important here.
First, nonproductive run descriptors cannot be responsible for increasing the number of $\sharp$ symbols read,
and thus (in order to obtain the second inequality of Lemma~\ref{lem:low-high-len}) 
it is enough to have a function $\high$ that takes into account only productive run descriptors (i.e., ignores nonproductive run descriptors).
Second, because $\high(\RR(0))$ and $\low(\RR(0))$ are counting the same productive run descriptors, only in a different way, the two values are related.
Namely, $\high(\RR(0))$ can be bounded by a function of $\low(\RR(0))$.
This essential property is shown in Proposition~\ref{prop:common-bound}, in the next subsection.

The function $\len$ is defined very similarly to $\high$, but it takes into account all run descriptors, not only productive ones.
Thus, roughly, it depends on the size of the annotated stack.

\begin{exa}
	We continue the previous examples, concerning the $2$-DPDA $\calA_1$ from Example~\ref{exa:der-tree}.
	Consider the annotated stack $\sbf_1=[[\tt_4,\tt_3],[\tt_2,\tt_1]]$, where 
	\begin{align*}
		&\tt_1=(\gamma,\{D_{1,\sharp,a}\})\,,&&\tt_3=(\gamma,\{E_3\})\,,\displaybreak[0]\\
		&\tt_2=(\gamma,\emptyset)\,,&&\tt_4=(\gamma,\{E_4\})\,.
	\end{align*}
	We have that
	\begin{align*}
		&\type(\tt_1)=\{\rd(D_{1,\sharp,a})\}=\{\sigma_{1,\pr}\}\,,&&\pi_2(\ass^1(\sigma_{1,\pr}))=\emptyset\,,\displaybreak[0]\\
		&\type(\tt_2)=\emptyset\,,\displaybreak[0]\\
		&\type(\tt_3)=\{\rd(E_3)\}=\{\tau_3\}\,,&&\pi_2(\ass^1(\tau_3))=\{\tau_4^1\}\,,\displaybreak[0]\\
		&\type(\tt_4)=\{\rd(E_4)\}=\{\tau_4\}\,.&&\pi_2(\ass^1(\tau_4))=\emptyset\,.
	\end{align*}
	In effect 
	\begin{align*}
		&\type([\tt_2,\tt_1])=\{\red^1(\sigma)\mid\sigma\in\type(\tt_1)\}=\{\red^1(\sigma_{1,\pr})\}\,,&&\pi_2(\ass^2(\red^1(\sigma_{1,\pr})))=\{(q_3,\pr)\}\,,\displaybreak[0]\\
		&\type([\tt_2])=\{\red^1(\sigma)\mid\sigma\in\type(\tt_2)\}=\emptyset\,,\displaybreak[0]\\
		&\type([\tt_4,\tt_3])=\{\red^1(\sigma)\mid\sigma\in\type(\tt_3)\}=\{\red^1(\tau_3)\}\,,&&\pi_2(\ass^2(\red^1(\tau_3)))=\emptyset\,,\displaybreak[0]\\
		&\type([\tt_4])=\{\red^1(\sigma)\mid\sigma\in\type(\tt_4)\}=\{\red^1(\tau_4)\}=\{\tau_4^1\}\,,\qquad\mbox{and finally}\hspace{-20em}\displaybreak[0]\\
		&\type([[\tt_4,\tt_3]])=\{\red^2(\sigma)\mid\sigma\in\type([\tt_4,\tt_3])\}=\{\red^2(\red^1(\tau_3))\}=\{(q_3,\pr)\}\,.\hspace{-20em}
	\end{align*}
	Recall that $\tau_3\in\calT_\np$ and $\sigma_{1,\pr},\tau_4\in\calT_\pr$.
	We can compute $\low(\sbf_1)$ as follows:
	\begin{align*}
		\low([\tt_4])&=\sum_{\sigma\in\type(\tt_4)}\big(\low([\,]{\restriction}_{\pi_2(\ass^1(\sigma))})+\low(\tt_4{\restriction}_{\{\sigma\}})\big)\\
			&=\low([\,]{\restriction}_\emptyset)+\low(\tt_4{\restriction}_{\{\tau_4\}})=\low([\,])+\low(\tt_4)=0+1=1\,,\displaybreak[0]\\
		\low([[\tt_4,\tt_3]])&=\sum_{\sigma\in\type([\tt_4,\tt_3])}\big(\low([\,]{\restriction}_{\pi_2(\ass^2(\sigma))})+\low([\tt_4,\tt_3]{\restriction}_{\{\sigma\}})\big)\\
			&=\low([\,]{\restriction}_\emptyset)+\low([\tt_4,\tt_3]{\restriction}_{\{\red^1(\tau_3)\}})=0+\low([\tt_4,\tt_3])\\
			&=\sum_{\sigma\in\type(\tt_3)}\big(\low([\tt_4]{\restriction}_{\pi_2(\ass^1(\sigma))})+\low(\tt_3{\restriction}_{\{\sigma\}})\big)\\
			&=\low([\tt_4]{\restriction}_{\{\tau_4^1\}})+\low(\tt_3{\restriction}_{\{\tau_3\}})=\low([\tt_4])+\low(\tt_3)=1+0=1\,,\displaybreak[0]\\
		\low([\tt_2])&=\sum_{\sigma\in\type(\tt_2)}\big(\low([\,]{\restriction}_{\pi_2(\ass^1(\sigma))})+\low(\tt_2{\restriction}_{\{\sigma\}})\big)=0\,,\displaybreak[0]\\
		\low([\tt_2,\tt_1])&=\sum_{\sigma\in\type(\tt_1)}\big(\low([\tt_2]{\restriction}_{\pi_2(\ass^1(\sigma))})+\low(\tt_1{\restriction}_{\{\sigma\}})\big)\\
			&=\low([\tt_2]{\restriction}_{\emptyset})+\low(\tt_1{\restriction}_{\{\sigma_{1,\pr}\}})=\low([\tt_2])+\low(\tt_1)=0+1=1\,,\displaybreak[0]\\
		\low(\sbf_1)&=\sum_{\sigma\in\type([\tt_2,\tt_1])}\big(\low([[\tt_4,\tt_3]]{\restriction}_{\pi_2(\ass^2(\sigma))})+\low([\tt_2,\tt_1]{\restriction}_{\{\sigma\}})\big)\\
			&=\low([[\tt_4,\tt_3]]{\restriction}_{\{(q_3,\pr)\}})+\low([\tt_2,\tt_1]{\restriction}_{\{\red^1(\sigma_{1,\pr})\}})\\
			&=\low([[\tt_4,\tt_3]])+\low([\tt_2,\tt_1])=1+1=2\,.
	\end{align*}
	We see that the restrictions of annotated stacks appearing in the above formulas do not modify these annotated stacks (this is the case because all annotations are either singletons or empty sets).

	Next, we compute $\high(\sbf_1)$.
	To this end, recall that $\depth(D_{1,\sharp,a})=2$ and $\depth(E_4)=1$.
	This time we give the formulas ignoring the restrictions, as again they do not change anything:
	\begin{align*}
		\high([\tt_4,\tt_3])&=\pow(\high([\tt_4]),\high(\tt_3))=\pow(\pow(\high([\,]),\high(\tt_4)),\high(\tt_3))\\
			&=\pow(\pow(1,C_{\depth(E_4)}),1)=\pow(\pow(1,C_1),1)=2^{C_1}-1\,,\displaybreak[0]\\
		\high([\tt_2,\tt_1])&=\pow(\high([\tt_2]),\high(\tt_1))=\pow(1,C_\depth(D_{1,\sharp,a}))=2^{C_2}-1\,,\displaybreak[0]\\
		\high(\sbf_1)&=\pow(\high([[\tt_4,\tt_3]]),\high([\tt_2,\tt_1]))\\
			&=\pow(\pow(\high([\,]),\high([\tt_4,\tt_3])),\high([\tt_2,\tt_1]))\\
			&=\pow(\pow(1,2^{C_1}-1),2^{C_2}-1)=2^{(2^{C_1}-1)(2^{C_2}-1)}-1\,.
	\end{align*}
	Notice that $\high(\tt_3)=1$, because $\type(\tt_3)$ contains only a nonproductive run descriptor.

	Similarly, we can compute $\len(\sbf_1)$, but this should also take into account $E_3$, whose depth is $0$:
	\begin{align*}
		\len([\tt_4,\tt_3])&=\pow(\len([\tt_4]),\len(\tt_3))=\pow(\pow(\len([\,]),\len(\tt_4)),\len(\tt_3))\\
			&=\pow(\pow(1,C_{\depth(E_4)}),C_{\depth(E_3)})=\pow(\pow(1,C_1),C_0)=2^{C_1\cdot C_0}-1\,,\displaybreak[0]\\
		\len([\tt_2,\tt_1])&=\pow(\len([\tt_2]),\len(\tt_1))=\pow(1,C_\depth(D_{1,\sharp,a}))=2^{C_2}-1\,,\displaybreak[0]\\
		\len(\sbf_1)&=\pow(\len([[\tt_4,\tt_3]]),\len([\tt_2,\tt_1]))\\
			&=\pow(\pow(\len([\,]),\len([\tt_4,\tt_3])),\len([\tt_2,\tt_1]))\\
			&=\pow(\pow(1,2^{C_1\cdot C_0}-1),2^{C_2}-1)=2^{(2^{C_1\cdot C_0}-1)(2^{C_2}-1)}-1\,.
	\end{align*}
\end{exa}

We have said that $\low$ counts the number of productive run descriptors in all annotations.
This is a good high-level intuition, but strictly speaking this is not true.
Indeed, say that we have two run descriptors $\sigma,\sigma'\in\type(\sbf^{k-1})$ such that $\ass^k(\sigma)=\ass^k(\sigma')$.
Then, in the formula for $\low(\sbf^k:\sbf^{k-1})$ we add $\low(\sbf^k{\restriction}_{\pi_2(\ass^k(\sigma))})$ twice (once for $\sigma$, and once for $\sigma'$).
This is illustrated by the next example.

\begin{exa}
	Consider the $2$-DPDA $\calA_2$ depicted below; its stack alphabet is $\{\gamma\}$, and input alphabet $\{\sharp\}$.
	\begin{center}
		\import{pics/}{example-7-24.pdf_tex_ok}
	\end{center}
	This time we consider the trivial monoid $M=\{1\}$.
	Denote
	\begin{align*}
		\sigma_i^1=(q_i,\{(1,(q_6,\pr))\},\pr)&&\mbox{for }i\in\{3,4,7\}\,.
	\end{align*}
	We are interested in derivation trees
	\begin{align*}
		D_1&=(\dtpush\gamma,q_1,(\dtpop\gamma,q_2,\sigma_3^1),\{(\dtpop\gamma,q_6,\sigma_7^1)\})\,,\displaybreak[0]\\
		D_i&=(\dtpop\gamma,q_i,\sigma_4^1)&&\mbox{for $i\in\{3,7\}$, and}\displaybreak[0]\\
		D_4&=(\dtread q_4,(\dtpop\gamma,q_5,(q_6,\pr)))\,.
	\end{align*}
	The annotated $1$-stack $\sbf=[(\gamma,\{D_4\}),(\gamma,\{D_3,D_7\}),(\gamma,\{D_1\})]$ is well-formed.
	Notice that the run descriptor of $D_4$ is productive, while run descriptors of $D_1$, $D_3$, and $D_7$ are nonproductive.
	We can see, though, that $D_1$ uses both $D_3$ and $D_7$ as assumptions, and both $D_3$ and $D_7$ use $D_4$ as an assumption.
	In effect $\low(\sbf)$ counts the nonproductive run descriptor $\rd(D_4)$ twice:
	\begin{align*}
		\low(\sbf)&=\low([(\gamma,\{D_4\}),(\gamma,\{D_3,D_7\})])+\low((\gamma,\{D_1\}))\displaybreak[0]\\
			&=\low([(\gamma,\{D_4\})])+\low((\gamma,\{D_3\}))+\low([(\gamma,\{D_4\})])+\low((\gamma,\{D_7\}))+0\\
			&=1+0+1+0+0=2\,.
	\end{align*}
\end{exa}

\begin{rem}
	Consider the function $\beth_n(k)$ defined by
	\begin{align*}
		\beth_0(k)=k&&\mbox{and}&&\beth_{n+1}(k)=2^{\beth_n(k)}\,.
	\end{align*}
	One can construct an $n$-DPDA $\calA$ that recognizes the language $\{\sharp^ka\sharp^{\beth_{n-1}(k)}\mid k\in\Nat\}$
	(see Blumensath~\cite[Example 9]{blumensath-pumping} for a very similar construction).
	After reading a prefix $\sharp^ka$, the number of $0$-stacks in the $n$-stack $s$ of $\calA$ is linear in $k$.
	It is possible to annotate $s$, resulting in an annotated stack $\sbf$, such that the maximal annotated run starting from $\sbf$ reads $\beth_{n-1}(k)$ $\sharp$ symbols.
	It follows that the $\high$ function (and thus $\len$ as well) has to be at least $(n-1)$-fold exponential in the size of an annotated $n$-stack.
	According to our definition, $\high$ and $\len$ are (in the worst case) $(n+1)$-fold exponential, which is slightly larger than necessary.
	We believe that it is possible to save these two exponentiations, at the cost of complicating proofs.
\end{rem}

We now prove Lemma~\ref{lem:low-high-len}, which fills the rest of this subsection.
We start by proving some (in)equalities regarding the $\pow$ function.

\begin{prop}
	The following is true for all positive integers:
	\begin{align}
		&\pow(a_1,\dots,a_k,\pow(b_1,\dots,b_l))=\pow(a_1,\dots,a_k,b_1,\dots,b_l)\,,\label{eq:wl1}\displaybreak[0]\\
		&\pow(a_1,\dots,a_k,\pow(c_0,c_1,\dots,c_l),b_1,\dots,b_l)\leq\nonumber\\
			&\hspace{5cm}\leq \pow(a_1,\dots,a_k,c_0,b_1c_1,\dots,b_lc_l)\,,\label{eq:wl2}\displaybreak[0]\\
		&\pow(a_1,\dots,a_{i-1},a_i^x,a_{i+1},\dots,a_{k-1},a_k)\leq \pow(a_1,\dots,a_{k-1},xa_k)&&\mbox{for }i<k\,,\label{eq:wl3}\displaybreak[0]\\
		&\pow(a_1,\dots,a_{k-1},a_k)+1\leq \pow(a_1,\dots,a_{k-1},a_k+1)\,,\label{eq:wl4}\displaybreak[0]\\
		&\pow(a_1,\dots,a_k)\cdot\pow(b_1,\dots,b_k)\leq \pow(a_1b_1,\dots,a_kb_k)\,.\label{eq:wl5}
	\end{align}
\end{prop}

\proof
	Equality~\eqref{eq:wl1} can be shown by induction on $k$.
	For $k=0$ we have that
	\begin{align*}
		\pow(\pow(b_1,\dots,b_l))&=(1+\pow(b_1,\dots,b_l))^{\pow()}-1\\
			&=(1+\pow(b_1,\dots,b_l))^1-1=\pow(b_1,\dots,b_l)\,,
	\end{align*}
	and for $k>0$ we directly use the induction assumption:
	\begin{align*}
		\pow(a_1,\dots,a_k,\pow(b_1,\dots,b_l))&=(1+a_1)^{\pow(a_2,\dots,a_k,\pow(b_1,\dots,b_l))}-1\\
			&=(1+a_1)^{\pow(a_2,\dots,a_k,b_1,\dots,b_l)}-1\\
			&=\pow(a_1,\dots,a_k,b_1,\dots,b_l)\,.
	\end{align*}
	
	For Inequality~\eqref{eq:wl2} suppose first that $k=0$.
	By Inequality~\eqref{eq:wl5}, which we prove below, it follows that
	\begin{align*}
		\pow(\pow(c_0,c_1,\dots,c_l),b_1,\dots,b_l)&=(1+\pow(c_0,c_1,\dots,c_l))^{\pow(b_1,\dots,b_l)}-1\\
			&=(1+(1+c_0)^{\pow(c_1,\dots,c_l)}-1)^{\pow(b_1,\dots,b_l)}-1\\
			&=(1+c_0)^{\pow(b_1,\dots,b_l)\cdot\pow(c_1,\dots,c_l)}-1\\
			&\leq(1+c_0)^{\pow(b_1c_1,\dots,b_lc_l)}-1
			=\pow(c_0,b_1c_1,\dots,b_lc_l)\,.
	\end{align*}
	It is easy to see that $\pow$ is monotone, thus the general form of Inequality~\eqref{eq:wl2} follows from the above special form thanks to Equality~\eqref{eq:wl1}:
	\begin{align*}
		&\pow(a_1,\dots,a_k,\pow(c_0,c_1,\dots,c_l),b_1,\dots,b_l)\\
			&\hspace{4em}=\pow(a_1,\dots,a_k,\pow(\pow(c_0,c_1,\dots,c_l),b_1,\dots,b_l))\\
			&\hspace{4em}\leq\pow(a_1,\dots,a_k,\pow(c_0,b_1c_1,\dots,b_lc_l))=\pow(a_1,\dots,a_k,c_0,b_1c_1,\dots,b_lc_l)\,.
	\end{align*}
	
	Heading toward proving Inequality~\eqref{eq:wl3}, we first show that
	\begin{align}
		x\cdot\pow(a_{i+1},\dots,a_k)\leq\pow(a_{i+1},\dots,a_{k-1},xa_k)\,,\label{eq:aux3}
	\end{align}
	where $i<k$, and the numbers $x,a_{i+1},\dots,a_k$ are positive integers.
	This is shown by induction on $k-i$.
	When $k-i=1$, we simply have that
	\begin{align*}
		x\cdot\pow(a_k)=x\cdot((1+a_k)^1-1)=(1+xa_k)^1-1=\pow(xa_k)\,.
	\end{align*}
	Suppose that $k-i>1$.
	Notice that $xb\leq b^x$ for all $x\in\Nat$ and $b\geq 2$.
	Thus,
	\begin{align*}
		x\cdot\pow(a_{i+1},\dots,a_k)&=x\cdot((1+a_{i+1})^{\pow(a_{i+2},\dots,a_k)}-1)\\
			&\leq x\cdot(1+a_{i+1})^{\pow(a_{i+2},\dots,a_k)}-1
			\leq (1+a_{i+1})^{x\cdot \pow(a_{i+2},\dots,a_k)}-1\\
			&\leq (1+a_{i+1})^{\pow(a_{i+2},\dots,a_{k-1},xa_k)}-1
			=\pow(a_{i+1},\dots,a_{k-1},xa_k)\,.
	\end{align*}
	Above, the first inequality holds because $x\geq 1$; 
	the second inequality follows from the inequality $xb\leq b^x$, where we notice that $(1+a_{i+1})^{\pow(a_{i+2},\dots,a_k)}\geq 2$ (because $a_{i+1}\geq 1$ and $\pow(\dots)\geq 1$);
	the third inequality is an application of the induction assumption.
	From Inequality~\eqref{eq:aux3} it follows that
	\begin{align*}
		\pow(a_i^x,a_{i+1},\dots,a_k)&=(1+a_i^x)^{\pow(a_{i+1},\dots,a_k)}-1
			\leq(1+a_i)^{x\cdot\pow(a_{i+1},\dots,a_k)}-1\\
			&\leq(1+a_i)^{\pow(a_{i+1},\dots,a_{k-1},xa_k)}-1=\pow(a_{i+1},\dots,a_{k-1},xa_k)\,.
	\end{align*}
	This gives the thesis, by Equality~\eqref{eq:wl1} and by monotonicity of $\pow$:
	\begin{align*}
		&\pow(a_1,\dots,a_{i-1},a_i^x,a_{i+1},\dots,a_{k-1},a_k)=\pow(a_1,\dots,a_{i-1},\pow(a_i^x,a_{i+1},\dots,a_{k-1},a_k))\\
		&\hspace{5em}\leq\pow(a_1,\dots,a_{i-1},\pow(a_i,\dots,a_{k-1},xa_k))
		=\pow(a_1,\dots,a_{k-1},xa_k)\,.
	\end{align*}
	
	Inequality~\eqref{eq:wl4} is shown by induction on $k$.
	For $k=1$ we simply have
	\begin{align*}
		\pow(a_1)+1=(1+a_1)^1-1+1=(1+a_1+1)^1-1=\pow(a_1+1)\,.
	\end{align*}
	For $k>1$ we use the induction assumption as follows (the first inequality below holds because $a_1\geq 1$):
	\begin{align*}
		\pow(a_1,\dots,a_{k-1},a_k)+1&=(1+a_1)^{\pow(a_2,\dots,a_k)}-1+1
		\leq(1+a_1)^{\pow(a_2,\dots,a_k)+1}-1\\
		&\leq(1+a_1)^{\pow(a_2,\dots,a_{k-1},a_k+1)}-1
		=\pow(a_1,\dots,a_{k-1},a_k+1)\,.
	\end{align*}
	
	Inequality~\eqref{eq:wl5} is also shown by induction on $k$.
	For $k=0$ the thesis is trivial:
	\begin{align*}
		\pow()\cdot\pow()=1\cdot 1=1=\pow()\,.
	\end{align*}
	Suppose now that $k\geq 1$, and denote $x=\pow(a_2,\dots,a_k)$ and $y=\pow(b_2,\dots,b_k)$.
	We claim that
	\begin{align}
		((1+a_1)^x-1)((1+b_1)^y-1)\leq(1+a_1b_1)^{xy}-1\,.\label{eq:aux5}
	\end{align}
	Let us prove this inequality.
	By symmetry, we can assume that $x\geq y$.
	We have three cases.
	If $x=y=1$, Inequality~\eqref{eq:aux5} simply says that
	\begin{align*}
		((1+a_1)^1-1)((1+b_1)^1-1)=a_1b_1\leq(1+a_1b_1)^{1\cdot 1}-1\,.
	\end{align*}
	Next, suppose that $x\geq 2$ and $y=1$.
	We see that
	\begin{align*}
		0&\leq(b_1-1)^2\,,\displaybreak[0]\\
		0&\leq b_1^2-2b_1+1\,,\displaybreak[0]\\
		4b_1&\leq b_1^2+2b_1+1\,,\displaybreak[0]\\
		4b_1&\leq (b_1+1)^2\,,\displaybreak[0]\\
		b_1&\leq \left(\frac{b_1+1}{2}\right)^2\,.
	\end{align*}
	Because $x\geq 2$ and $b_1\geq 1$, it follows that
	\begin{align}
		b_1\leq \left(\frac{b_1+1}{2}\right)^x\,.\label{eq:aux5a}
	\end{align}
	Next, observe that
	\begin{align}
		0&\leq(a_1-1)(b_1-1)\,,\displaybreak[0]\nonumber\\
		0&\leq a_1b_1-a_1-b_1+1\,,\displaybreak[0]\nonumber\\
		a_1b_1+a_1+b_1+1&\leq 2a_1b_1+2\,,\displaybreak[0]\nonumber\\
		(1+a_1)(b_1+1)&\leq 2+2a_1b_1\,,\displaybreak[0]\nonumber\\
		(1+a_1)\left(\frac{b_1+1}{2}\right)&\leq 1+a_1b_1\,.\label{eq:aux5b}
	\end{align}
	Using Inequalities~\eqref{eq:aux5a} and~\eqref{eq:aux5b} we obtain Inequality~\eqref{eq:aux5}:
	\begin{align*}
		((1+a_1)^x-1)((1+b_1)^1-1)&=(1+a_1)^x\cdot b_1-b_1
			\leq (1+a_1)^x\cdot b_1-1\\
			&\leq(1+a_1)^x\left(\frac{b_1+1}{2}\right)^x-1\leq(1+a_1b_1)^{x\cdot 1}-1\,.
	\end{align*}
	The remaining case is when $x\geq y\geq 2$.
	In this case we have that 
	\begin{align*}
		((1+a_1)^x-1)((1+b_1)^y-1)&\leq(1+a_1)^x(1+b_1)^y-1
			\leq(1+a_1)^x(1+b_1)^x-1\\
			&\leq(1+a_1b_1)^x(1+a_1b_1)^x-1
			=(1+a_1b_1)^{x\cdot 2}-1\\
			&\leq(1+a_1b_1)^{x\cdot y}-1\,.
	\end{align*}
	Thus, we have shown Inequality~\eqref{eq:aux5} in all cases.
	Using this inequality and the induction assumption, we can conclude that
	\begin{align*}
		\pow(a_1,\dots,a_k)\cdot\pow(b_1,\dots,b_k)&=((1+a_1)^x-1)((1+b_1)^y-1)
			\leq(1+a_1b_1)^{x\cdot y}-1\\
			&\leq(1+a_1b_1)^{\pow(a_2b_2,\dots,a_kb_k)}-1
			=\pow(a_1b_1,\dots,a_kb_k)\,.\tag*{\qed}
	\end{align*}

Heading toward the proof of Lemma~\ref{lem:low-high-len}, we now observe some auxiliary properties.

\begin{prop}\lab{prop:pr-positive}
	Let $\sbf$ be a well-formed annotated stack.
	If $\type(\sbf)\subseteq\calT_\np$ then $\low(\sbf)=0$ and $\high(\sbf)=1$; otherwise $\low(\sbf)\geq 1$ and $\high(\sbf)\geq 2$.
\end{prop}

\proof
	By induction on the structure of $\sbf$. In the base cases of a $0$-stack and of an empty $k$-stack, the thesis follows directly from Definition~\ref{def:low-high}.
	In the induction step denote $\sbf=\sbf^k:\sbf^{k-1}$.
	Recall that $\type(\sbf)=\{\red^k(\sigma)\mid\sigma\in\type(\sbf^{k-1})\}$ (by the definition of types), so
	\begin{align*}
		\type(\sbf)\subseteq\calT_\np\Leftrightarrow \forall_{\sigma\in\type(\sbf^{k-1})}\red^k(\sigma)\in\calT_\np\,.
	\end{align*}
	Moreover,
	\begin{align*}
		\red^k(\sigma)\in\calT_\np\Leftrightarrow(\sigma\in\calT_\np\land\pi_2(\ass^k(\sigma))\subseteq\calT_\np)
	\end{align*}
	for $\sigma\in\calT^{k-1}$ (by the definition of $\red^k$).
	It follows that
	\begin{align*}
		\type(\sbf)\subseteq\calT_\np\Leftrightarrow \forall_{\sigma\in\type(\sbf^{k-1})}(\sigma\in\calT_\np\land\pi_2(\ass^k(\sigma))\subseteq\calT_\np)\,.
	\end{align*}
	
	If $\type(\sbf)\subseteq\calT_\np$ then, by the induction assumption, $\low(\sbf^k{\restriction}_{\pi_2(\ass^k(\sigma))})=\low(\sbf^{k-1}{\restriction}_{\{\sigma\}})=0$ 
	and $\high(\sbf^k{\restriction}_{\pi_2(\ass^k(\sigma))})=\high(\sbf^{k-1}{\restriction}_{\{\sigma\}})=1$ for all $\sigma\in\type(\sbf^{k-1})$;
	in effect
	\begin{align*}
		\low(\sbf)&=\sum_{\sigma\in\type(\sbf^{k-1})}\big(\low(\sbf^k{\restriction}_{\pi_2(\ass^k(\sigma))})+\low(\sbf^{k-1}{\restriction}_{\{\sigma\}})\big)
			=\sum_{\sigma\in\type(\sbf^{k-1})}(0+0)=0\,,&\mbox{and}\displaybreak[0]\\
		\high(\sbf)&=\prod_{\sigma\in\type(\sbf^{k-1})}\pow\big(\high(\sbf^k{\restriction}_{\pi_2(\ass^k(\sigma))}),\high(\sbf^{k-1}{\restriction}_{\{\sigma\}})\big)\,.\\
			&=\prod_{\sigma\in\type(\sbf^{k-1})}\pow(1,1)
			=\prod_{\sigma\in\type(\sbf^{k-1})}1=1\,.
	\end{align*}

	Conversely, if $\type(\sbf)\not\subseteq\calT_\np$ then, by the induction assumption, at least one among $\low(\sbf^k{\restriction}_{\pi_2(\ass^k(\sigma))})$ 
	and $\low(\sbf^{k-1}{\restriction}_{\{\sigma\}})$ for $\sigma\in\type(\sbf^{k-1})$ is positive (and all other are nonnegative);
	in effect $\low(\sbf)$, being their sum, is positive.
	Similarly, at least one among $\high(\sbf^k{\restriction}_{\pi_2(\ass^k(\sigma))})$ and $\high(\sbf^{k-1}{\restriction}_{\{\sigma\}})$ is greater than $1$ (and all other are positive);
	in effect some $\pow\big(\high(\sbf^k{\restriction}_{\pi_2(\ass^k(\sigma))}),\high(\sbf^{k-1}{\restriction}_{\{\sigma\}})\big)$ is greater than $1$
	(notice that $\pow(2,1)=(1+2)^1-1=2$, and $\pow(1,2)=(1+1)^2-1=3$, and that $\pow$ is monotone), 
	and thus their product $\high(\sbf)$ is greater than $1$.
\qed

\begin{prop}\label{prop:expand-to-singular}
	For every well-formed annotated stack $\sbf$,
	\begin{gather*}
		\low(\sbf)=\sum_{\sigma\in\type(\sbf)}\low(\sbf{\restriction}_{\{\sigma\}})\,,\qquad
		\high(\sbf)=\prod_{\sigma\in\type(\sbf)}\high(\sbf{\restriction}_{\{\sigma\}})\,,\qquad\mbox{and}\displaybreak[0]\\
		\len(\sbf)=\prod_{\sigma\in\type(\sbf)}\len(\sbf{\restriction}_{\{\sigma\}})\,.
	\end{gather*}
\end{prop}

\proof
	We analyze Definition~\ref{def:low-high}.
	Suppose first that $\sbf=(\gamma,\mathfrak{D})$ (i.e., that $\sbf$ is of order $0$).
	Because $\type(\sbf{\restriction}_{\{\sigma\}})=\{\sigma\}$, 
	\begin{align*}
		\low(\sbf)&=|\type(\sbf)\cap\calT_\pr|
			=\sum_{\sigma\in\type(\sbf)}|\{\sigma\}\cap\calT_\pr|\\
			&=\sum_{\sigma\in\type(\sbf)}|\type(\sbf{\restriction}_{\{\sigma\}})\cap\calT_\pr|
			=\sum_{\sigma\in\type(\sbf)}\low(\sbf{\restriction}_{\{\sigma\}})\,.
	\end{align*}
	Recall that $\type(\sbf)=\{\rd(D)\mid D\in\mathfrak{D}\}$, and that $\sbf{\restriction}_{\{\sigma\}}=(\gamma,\{D\in\mathfrak{D}\mid\rd(D)=\sigma\})$; thus,
	\begin{align*}
		\high(\sbf)&=\prod_{D\in\mathfrak{D}\mid\rd(D)\in\calT_\pr}C_{\depth(D)}\\
			&=\prod_{\sigma\in\type(\sbf)}\left(\prod_{D\in\mathfrak{D}\mid\rd(D)=\sigma\in\calT_\pr}C_{\depth(D)}\right)
			=\prod_{\sigma\in\type(\sbf)}\high(\sbf{\restriction}_{\{\sigma\}})\,,\qquad\mbox{and}\displaybreak[0]\\
		\len(\sbf)&=\prod_{D\in\mathfrak{D}}C_{\depth(D)}
			=\prod_{\sigma\in\type(\sbf)}\left(\prod_{D\in\mathfrak{D}\mid\rd(D)=\sigma}C_{\depth(D)}\right)
			=\prod_{\sigma\in\type(\sbf)}\len(\sbf{\restriction}_{\{\sigma\}})\,.
	\end{align*}

	If $\sbf=[\,]$, then $\type(\sbf)=\emptyset$, and thus
	\begin{gather*}
		\low(\sbf)=0=\sum_{\sigma\in\emptyset}\low(\sbf{\restriction}_{\{\sigma\}})\,,\qquad
		\high(\sbf)=1=\prod_{\sigma\in\emptyset}\high(\sbf{\restriction}_{\{\sigma\}})\,,\qquad\mbox{and}\displaybreak[0]\\
		\len(\sbf)=1=\prod_{\sigma\in\emptyset}\len(\sbf{\restriction}_{\{\sigma\}})\,.
	\end{gather*}
	
	Finally, suppose that $\sbf=\sbf^k:\sbf^{k-1}$.
	Recall that $\type(\sbf)=\{\red^k(\tau)\mid\tau\in\type(\sbf^{k-1})\}$.
	Moreover, by well-formedness of $\sbf$, for every $\sigma\in\type(\sbf)$ there is exactly one $\tau\in\type(\sbf^{k-1})$ such that $\red^k(\tau)=\sigma$;
	denote it $\tau_\sigma$.
	By the definition of a restriction, we have that $\sbf{\restriction}_{\{\sigma\}}=\sbf^k{\restriction}_{\pi_2(\ass^k(\tau_\sigma))}:\sbf^{k-1}{\restriction}_{\{\tau_\sigma\}}$.
	Recalling that $\type(\sbf{\restriction}_{\{\sigma\}})=\{\sigma\}$, we obtain
	\begin{align*}
		\low(\sbf)&=\sum_{\tau\in\type(\sbf^{k-1})}\big(\low(\sbf^k{\restriction}_{\pi_2(\ass^k(\tau))})+\low(\sbf^{k-1}{\restriction}_{\{\tau\}})\big)\\
			&=\sum_{\sigma\in\type(\sbf)}\big(\low(\sbf^k{\restriction}_{\pi_2(\ass^k(\tau_\sigma))})+\low(\sbf^{k-1}{\restriction}_{\{\tau_\sigma\}})\big)
			=\sum_{\sigma\in\type(\sbf)}\low(\sbf{\restriction}_{\{\sigma\}})\,,\displaybreak[0]\\
		\high(\sbf)&=\prod_{\tau\in\type(\sbf^{k-1})}\pow\big(\high(\sbf^k{\restriction}_{\pi_2(\ass^k(\tau))}),\high(\sbf^{k-1}{\restriction}_{\{\tau\}})\big)\\
			&=\prod_{\sigma\in\type(\sbf)}\pow\big(\high(\sbf^k{\restriction}_{\pi_2(\ass^k(\tau_\sigma))}),\high(\sbf^{k-1}{\restriction}_{\{\tau_\sigma\}})\big)
			=\prod_{\sigma\in\type(\sbf)}\high(\sbf{\restriction}_{\{\sigma\}})\,,&&\mbox{and}\displaybreak[0]\\
		\len(\sbf)&=\prod_{\tau\in\type(\sbf^{k-1})}\pow\big(\len(\sbf^k{\restriction}_{\pi_2(\ass^k(\tau))}),\len(\sbf^{k-1}{\restriction}_{\{\tau\}})\big)\\
			&=\prod_{\sigma\in\type(\sbf)}\pow\big(\len(\sbf^k{\restriction}_{\pi_2(\ass^k(\tau_\sigma))}),\len(\sbf^{k-1}{\restriction}_{\{\tau_\sigma\}})\big)
			=\prod_{\sigma\in\type(\sbf)}\len(\sbf{\restriction}_{\{\sigma\}})\,.
		\tag*{\qed}
	\end{align*}

\begin{prop}\lab{prop:low-high-singular}
	Let $0\leq l\leq k\leq n$, and let $\sbf=\sbf^k:\sbf^{k-1}:\dots:\sbf^l$ be a well-formed annotated $k$-stack that is singular.
	In this situation
	\begin{align*}
		\low(\sbf)&=\sum_{i=l}^k\low(\sbf^i)\,,\displaybreak[0]\\
		\high(\sbf)&=\pow(\high(\sbf^k),\high(\sbf^{k-1}),\dots,\high(\sbf^l))\,,\qquad\mbox{and}\displaybreak[0]\\
		\len(\sbf)&=\pow(\len(\sbf^k),\len(\sbf^{k-1}),\dots,\len(\sbf^l))\,.
	\end{align*}
\end{prop}

\begin{proof}
	Induction on $k-l$.
	For $k-l=0$, we simply have $\sbf=\sbf^k$; both sides of each equality are the same (recall that $\pow(x)=x$ for every $x$).
	
	Suppose that $k-l\geq 1$, and denote $\tt=\sbf^{k-1}:\sbf^{k-2}:\dots:\sbf^l$; we have that $\sbf=\sbf^k:\tt$.
	Because $\sbf$ is well-formed and $\type(\sbf)$ is a singleton, also $\type(\tt)$ is a singleton $\{\sigma\}$, where $\type(\sbf^k)=\pi_2(\ass^k(\sigma))$.
	In effect, restricting $\tt$ to $\{\sigma\}$ or $\sbf^k$ to $\pi_2(\ass^k(\sigma))$ does not change the annotated stacks, so, by definition,
	\begin{align*}
		\low(\sbf)&=\sum_{\sigma\in\type(\tt)}\big(\low(\sbf^k{\restriction}_{\pi_2(\ass^k(\sigma))})+\low(\tt{\restriction}_{\{\sigma\}})\big)
			=\low(\sbf^k)+\low(\tt)\,.
	\end{align*}
	Similarly,
	\begin{align*}
		\high(\sbf)=\pow(\high(\sbf^k),\high(\tt))&&\mbox{and}&&\len(\sbf)=\pow(\len(\sbf^k),\len(\tt))\,.
	\end{align*}
	From the induction assumption we know that
	\begin{align*}
		\low(\tt)&=\sum_{i=l}^{k-1}\low(\sbf^i)\,,\\
		\high(\tt)&=\pow(\high(\sbf^{k-1}),\high(\sbf^{k-2}),\dots,\high(\sbf^l))\,,\qquad\mbox{and}\\
		\len(\tt)&=\pow(\len(\sbf^{k-1}),\len(\sbf^{k-2}),\dots,\len(\sbf^l))\,.
	\end{align*}
	By substituting this to equalities for $\low(\sbf)$, $\high(\sbf)$, and $\len(\sbf)$ we obtain the thesis,
	where in the case of $\high$ and $\len$ we additionally use Equality~\eqref{eq:wl1}.
\end{proof}

\begin{prop}\lab{prop:restriction}
	Let $0\leq l\leq k\leq n$, let $\sbf=\sbf^k:\sbf^{k-1}:\dots:\sbf^l$ be a well-formed annotated $k$-stack, and let $\sigma\in\type(\sbf^l)$.
	In this situation
	\begin{align*}
		\sbf{\restriction}_{\{\red^k(\sigma)\}}=\sbf^k{\restriction}_{\pi_2(\ass^k(\sigma))}:\sbf^{k-1}{\restriction}_{\pi_2(\ass^{k-1}(\sigma))}:\dots:
			\sbf^{l+1}{\restriction}_{\pi_2(\ass^{l+1}(\sigma))}:\sbf^l{\restriction}_{\{\sigma\}}\,.
	\end{align*}
\end{prop}

\begin{proof}
	Proposition~\ref{prop:well-formed-multi} used for $\sbf$ implies that 
	$\type(\sbf)=\{\red^k(\sigma)\mid\sigma\in\type(\sbf^l)\}$ and $|\type(\sbf)|=|\type(\sbf^l)|$;
	the latter means that $\red^k(\sigma)=\red^k(\sigma')$ implies $\sigma=\sigma'$ for $\sigma,\sigma'\in\type(\sbf^l)$.

	Observe that $\tops^l(\sbf{\restriction}_{\{\red^k(\sigma)\}})$ equals $\sbf^l$ restricted to a subset of $\type(\sbf^l)$.
	Proposition~\ref{prop:well-formed-multi} used for $\sbf{\restriction}_{\{\red^k(\sigma)\}}$ implies that this subset is a singleton $\{\sigma'\}$,
	and that $\red^k(\sigma')=\red^k(\sigma)$.
	This implies that $\sigma'=\sigma$, by the previous paragraph.
	Using also Item~\ref{pkt:types-ass} of Proposition~\ref{prop:well-formed-multi}, we see that 
	\begin{align*}
		\sbf{\restriction}_{\{\red^k(\sigma)\}}=\sbf^k{\restriction}_{\pi_2(\ass^k(\sigma))}:\sbf^{k-1}{\restriction}_{\pi_2(\ass^{k-1}(\sigma))}:\dots:
			\sbf^{l+1}{\restriction}_{\pi_2(\ass^{l+1}(\sigma))}:\sbf^l{\restriction}_{\{\sigma\}}\,,
	\end{align*}
	as required.
\end{proof}

Next, we observe how the functions $\low$, $\high$, and $\len$ interplay with composing annotated stacks.

\begin{lem}\lab{lem:low-high-composer}
	Let $0\leq l\leq k\leq n$, let $(\Phi^k,\Phi^{k-1},\dots,\Phi^l;\Psi^k;f)$ be a composer, 
	and let $\sbf=\sbf^k:\sbf^{k-1}:\dots:\sbf^l$ be a well-formed annotated $k$-stack such that $\type(\sbf^i)=\pi_2(\Phi^i)$ for each $i\in\prz{l}{k}$.
	In this situation
	\begin{align}
		&\sum_{i=l}^k\low(\sbf^i)\leq \low(\sbf)\,,\label{eq:lhc1}\\
		&\sum_{i=l}^k\low(\sbf^i)<\low(\sbf)&&\mbox{if }f=\pr\,,\label{eq:lhc2}\\
		&\pow\big(\high(\sbf^k),\high(\sbf^{k-1}),\dots,\high(\sbf^{l+1}),\left|\calT^0\right|^n\cdot\high(\sbf^l)\big)\geq \high(\sbf)\,,\label{eq:lhc3}\\
		&\pow\big(\high(\sbf^k),\high(\sbf^{k-1}),\dots,\high(\sbf^{l+1}),\high(\sbf^l)\big)\geq \high(\sbf)&&\mbox{if }f=\np\,,\label{eq:lhc4}\\
		&\pow\big(\len(\sbf^k),\len(\sbf^{k-1}),\dots,\len(\sbf^{l+1}),\left|\calT^0\right|^n\cdot\len(\sbf^l)\big)\geq \len(\sbf)\,.\label{eq:lhc5}
	\end{align}
\end{lem}

\proof

	Because $\type(\sbf^l)=\pi_2(\Phi^l)$, Proposition~\ref{prop:well-formed-multi} used for $\sbf$ implies that 
	$\type(\sbf)=\{\red^k(\sigma)\mid\sigma\in\pi_2(\Phi^l)\}$ and $|\type(\sbf)|=|\pi_2(\Phi^l)|$,
	which means that the mapping defined by $\sigma\mapsto\red^k(\sigma)$ is a bijection between $\pi_2(\Phi^l)$ and $\type(\sbf)$.
	Moreover, for $\sigma\in\pi_2(\Phi^l)$, 
	\begin{align*}
		\sbf{\restriction}_{\{\red^k(\sigma)\}}=\sbf^k{\restriction}_{\pi_2(\ass^k(\sigma))}:\sbf^{k-1}{\restriction}_{\pi_2(\ass^{k-1}(\sigma))}:\dots:
			\sbf^{l+1}{\restriction}_{\pi_2(\ass^{l+1}(\sigma))}:\sbf^l{\restriction}_{\{\sigma\}}\,,
	\end{align*}
	by Proposition~\ref{prop:restriction}.
	
	For $i\in\prz{l+1}{k}$ and $\sigma\in\pi_2(\Phi^l)$ denote $H^i_\sigma=\high(\sbf^i{\restriction}_{\pi_2(\ass^i(\sigma))})$.
	Using Proposition~\ref{prop:expand-to-singular}, the above property, and Proposition~\ref{prop:low-high-singular}%
	, we obtain that
	\begin{align}
		\low(\sbf)&=\sum_{\tau\in\type(\sbf)}\low(\sbf{\restriction}_{\{\tau\}})
			=\sum_{\sigma\in\pi_2(\Phi^l)}\low(\sbf{\restriction}_{\{\red^l(\sigma)\}})\nonumber\\
			&=\sum_{\sigma\in\pi_2(\Phi^l)}\Big(\low(\sbf^l{\restriction}_{\{\sigma\}}))+\sum_{i=l+1}^k\low(\sbf^i{\restriction}_{\pi_2(\ass^i(\sigma))})\Big)\nonumber\\
			&=\low(\sbf^l)+\sum_{i=l+1}^k\sum_{\sigma\in\pi_2(\Phi^l)}\low(\sbf^i{\restriction}_{\pi_2(\ass^i(\sigma))})\,,
				\qquad\mbox{and}\label{eq:low-composer-1}\displaybreak[0]\\
		\high(\sbf)&=\prod_{\tau\in\type(\sbf)}\high(\sbf{\restriction}_{\{\tau\}})
			=\prod_{\sigma\in\pi_2(\Phi^l)}\high(\sbf{\restriction}_{\{\red^k(\sigma)\}})\nonumber\\
			&=\prod_{\sigma\in\pi_2(\Phi^l)}\pow\big(H^k_{\sigma},H^{k-1}_{\sigma},\dots,H^{l+1}_{\sigma},\high(\sbf^l{\restriction}_{\{\sigma\}})\big)\,.\label{eq:high-composer-1}
	\end{align}

	For each $i\in\prz{l+1}{k}$ it holds that $\type(\sbf^i)=\pi_2(\Phi^i)=\bigcup\{\pi_2(\ass^i(\sigma))\mid\sigma\in\pi_2(\Phi^l)\}$
	(by the definition of a composer, Condition~\ref{pkt:compo-ass}), so, by Proposition~\ref{prop:expand-to-singular},
	\begin{align}
		\low(\sbf^i)=\sum_{\tau\in\type(\sbf^i)}\low(\sbf^i{\restriction}_{\{\tau\}})
			\leq\sum_{\sigma\in\pi_2(\Phi^l)}\sum_{\tau\in\pi_2(\ass^i(\sigma))}\low(\sbf^i{\restriction}_{\{\tau\}})
			=\sum_{\sigma\in\pi_2(\Phi^l)}\low(\sbf^i{\restriction}_{\pi_2(\ass^i(\sigma))})\,.\label{eq:low-composer-2}
	\end{align}
	Altogether, Equality~\eqref{eq:low-composer-1} and Inequality~\eqref{eq:low-composer-2} used for all $i\in\prz{l+1}{k}$ imply Inequality~\eqref{eq:lhc1}.
	
	For Inequality~\eqref{eq:lhc2}, recall from the definition of a composer (Condition~\ref{pkt:compo-flag}) that if $f=\pr$, then for some $i\in\prz{l+1}{k}$,
	some $\tau\in\calT_\pr$ appears in $\pi_2(\ass^i(\sigma))$ simultaneously for two different $\sigma\in\pi_2(\Phi^l)$.
	By Proposition~\ref{prop:pr-positive}, $\low(\sbf^i{\restriction}_{\{\tau\}})\geq 1$ for this $\tau$ (because it is productive).
	Thus, some positive component $\low(\sbf^i{\restriction}_{\{\tau\}})$
	appears in two sums $\sum_{\tau\in\pi_2(\ass^i(\sigma))}\low(\sbf^i{\restriction}_{\{\tau\}})$ in Inequality~\eqref{eq:low-composer-2} used for this $i$, 
	so this inequality (and, in effect, Inequality~\eqref{eq:lhc1}) becomes strict.

	Using Equality~\eqref{eq:high-composer-1}, Inequality~\eqref{eq:wl5}, and Proposition~\ref{prop:expand-to-singular} we obtain
	\begin{align}
		\high(\sbf)&=\prod_{\sigma\in\pi_2(\Phi^l)}\pow\big(H^k_\sigma,H^{k-1}_\sigma,\dots,H^{l+1}_\sigma,\high(\sbf^l{\restriction}_{\{\sigma\}})\big)\nonumber\\
		&\leq\pow\Big(\prod_{\sigma\in\pi_2(\Phi^l)}H^k_\sigma,\prod_{\sigma\in\pi_2(\Phi^l)}H^{k-1}_\sigma,\dots,\prod_{\sigma\in\pi_2(\Phi^l)}H^{l+1}_\sigma,
			\prod_{\sigma\in\pi_2(\Phi^l)}\high(\sbf^l{\restriction}_{\{\sigma\}})\Big)\nonumber\\
		&=\pow\Big(\prod_{\sigma\in\pi_2(\Phi^l)}H^k_\sigma,\prod_{\sigma\in\pi_2(\Phi^l)}H^{k-1}_\sigma,\dots,\prod_{\sigma\in\pi_2(\Phi^l)}H^{l+1}_\sigma,\high(\sbf^l)\Big)\,.\label{eq:uvuv1}
	\end{align}
	Observe that by restricting an annotated stack we can only decrease the value of $\high$.
	Thus, for each $i\in\prz{l+1}{k}$,
	\begin{align*}
		\prod_{\sigma\in\pi_2(\Phi^l)}H^i_\sigma&=\prod_{\sigma\in\pi_2(\Phi^l)}\high(\sbf^i{\restriction}_{\pi_2(\ass^i(\sigma))})\\
			&\leq\prod_{\sigma\in\pi_2(\Phi^l)}\high(\sbf^i)
			=(\high(\sbf^i))^{|\pi_2(\Phi^l)|}\leq(\high(\sbf^i))^{|\calT^0|}\,.
	\end{align*}
	The last inequality is true because $|\pi_2(\Phi^l)|\leq|\calT^{l}|\leq|\calT^0|$.
	Using Inequality~\eqref{eq:wl3} we move the $|\calT^0|$ exponents (there is at most $n$ of them) into the last argument of $\pow$
	and we obtain Inequality~\eqref{eq:lhc3}:
	\begin{align*}
		\high(\sbf)&\leq\pow\big((\high(\sbf^k))^{|\calT^0|},(\high(\sbf^{k-1}))^{|\calT^0|},\dots,(\high(\sbf^{l+1}))^{|\calT^0|},\high(\sbf^l)\big)\leq\\
		&\leq\pow\big(\high(\sbf^k),\high(\sbf^{k-1}),\dots,\high(\sbf^{l+1}),|\calT^0|^n\cdot\high(\sbf^l)\big)\,.
	\end{align*}
	
	Now suppose that $f=\np$. 
	It implies, for each $i\in\prz{l+1}{k}$, that each $\tau\in\pi_2(\Phi^i)\cap\calT_\pr$ belongs to the set $\pi_2(\ass^i(\sigma))$ only for one $\sigma\in\pi_2(\Phi^l)$,
	so all the common factors are equal to $1$ (cf.~Proposition~\ref{prop:pr-positive}): 
	\begin{align*}
		\prod_{\sigma\in\pi_2(\Phi^l)}H^i_\sigma=\prod_{\sigma\in\pi_2(\Phi^l)}\prod_{\tau\in\pi_2(\ass^i(\sigma))}\high(\sbf^i{\restriction}_{\{\tau\}})=
		\prod_{\tau\in\pi_2(\Phi^i)}\high(\sbf^i{\restriction}_{\{\tau\}})=\high(\sbf^i)\,.
	\end{align*}
	By substituting this to Inequality~\eqref{eq:uvuv1} we obtain Inequality~\eqref{eq:lhc4}.
	
	Inequality~\eqref{eq:lhc5} is obtained in the same way as Inequality~\eqref{eq:lhc3}, because the definitions of $\len$ and $\high$ differ only in the base case.
\qed

We are now ready to prove Lemma~\ref{lem:low-high-len}.

\proof[Proof of Lemma~\ref{lem:low-high-len}]
	It is enough to prove the lemma for annotated runs of length $1$.
	Then the result for longer runs follow by an immediate induction.
	Thus, assume that $|\RR|=1$, and denote $\RR(0)=\sbf^n:\sbf^{n-1}:\dots:\sbf^0$ with $\sbf^0=(\gamma,\{D\})$.
	Recall that our goal is to prove the following inequalities:
	\begin{align*}
		&\low(\RR(0))\leq\sharp(\st(\RR))+\low(\RR(1))\,,\displaybreak[0]\\
		&\high(\RR(0))\geq\sharp(\st(\RR))+\high(\RR(1))\,,&&\mbox{and}\displaybreak[0]\\
		&\len(\RR(0))\geq 1+\len(\RR(1))\,.
	\end{align*}

	We have four cases, depending on the shape of $D$.
	
	\paragraph*{Case 1}
	It is impossible that $D$ is of the form $(\dtempty\gamma,p)$, since then $\RR(0)$ would not have a successor.

	\paragraph*{Case 2}
	Suppose that $D=(\dtread p,D')$.
	Then $\RR(1)=\sbf^n:\sbf^{n-1}:\dots:\sbf^1:(\gamma,\{D'\})$.
	Because both $\RR(0)$ and $\RR(1)$ are singular, by Proposition~\ref{prop:low-high-singular} we have that
	\begin{align*}
		\low(\RR(0))=\low(\sbf^0)+\sum_{i=1}^n\low(\sbf^i)&&\mbox{and}&&\low(\RR(1))=\low((\gamma,\{D'\}))+\sum_{i=1}^n\low(\sbf^i)\,.
%		\high(\sbf)&=\pow(\high(\sbf^k),\high(\sbf^{k-1}),\dots,\high(\sbf^l))\,,\qquad\mbox{and}\\
%		\len(\sbf)&=\pow(\len(\sbf^k),\len(\sbf^{k-1}),\dots,\len(\sbf^l))\,.
	\end{align*}
	Thus, the required inequality about $\low$ can be restated as
	\begin{align*}
		\low(\sbf^0)\leq\sharp(\st(\RR))+\low((\gamma,\{D'\}))\,.
	\end{align*}
	It holds when $\rd(D)\in\calT_\np$ (as then $\low(\sbf^0)=0$).
	If $\rd(D)\in\calT_\pr$, then $\low(\sbf^0)=1$, and either $\rd(D')\in\calT_\pr$ or the letter read by $\st(\RR)$ is $\sharp$,
	so the right side is positive.
	
	Again by Proposition~\ref{prop:low-high-singular} we have that
	\begin{align*}
		\high(\RR(0))&=\pow(\high(\sbf^n),\high(\sbf^{n-1}),\dots,\high(\sbf^1),\high(\sbf^0))\,,\qquad\mbox{and}\displaybreak[0]\\
		\high(\RR(1))&=\pow(\high(\sbf^n),\high(\sbf^{n-1}),\dots,\high(\sbf^1),\high((\gamma,\{D'\})))\,.
	\end{align*}%
	If $\rd(D)\in\calT_\np$, then $\sharp(\st(\RR))=0$ and $\rd(D')\in\calT_\np$.
	In this case $\high(\sbf^0)=\high((\gamma,\{D'\}))=1$, so the two sides of the inequality are equal:
	\begin{align*}
		\high(\RR(0))=\sharp(\st(\RR))+\high(\RR(1))\,.
	\end{align*}
	If $\rd(D)\in\calT_\pr$, then using Inequality~\eqref{eq:wl4} we obtain
	\begin{align*}
		\high(\RR(0))&=\pow(\high(\sbf^n),\high(\sbf^{n-1}),\dots,\high(\sbf^1),C_{\depth(D)})\\
			&\geq\pow(\high(\sbf^n),\high(\sbf^{n-1}),\dots,\high(\sbf^1),C_{\depth(D')}+1)\\
			&\geq\pow(\high(\sbf^n),\high(\sbf^{n-1}),\dots,\high(\sbf^1),C_{\depth(D')})+1\\
			&\geq\sharp(\st(\RR))+\high(\RR(1))\,.
	\end{align*}
	In the same way we obtain the required inequality for $\len$.
	
	\paragraph*{Case 3}
	Suppose that $D=(\dtpop\gamma,p,\tau^k)$.
	Then $\RR(1)=\sbf^n:\sbf^{n-1}:\dots:\sbf^k$.
	The operation between $\conf(\RR(0))$ and $\conf(\RR(1))$ is $\pop^k$, so $\sharp(\st(\RR))=0$.
	For $i\in\prz{1}{k-1}$, by Definition~\ref{def:types}(\ref{pkt:dt-pop}) we have that $\ass^i(\rd(D))=\emptyset$, 
	and by well-formedness of $\RR(0)$ we have that $\type(\sbf^i)=\pi_2(\ass^i(\rd(D)))$ (cf.~Proposition~\ref{prop:well-formed-multi});
	in effect $\type(\sbf^i)=\emptyset$, hence $\low(\sbf^i)=0$ and $\high(\sbf^i)=1$.
	By Definition~\ref{def:types}(\ref{pkt:dt-pop}), $\rd(R)\in\calT_\np$, so $\low(\sbf^0)=0$ and $\high(\sbf^0)=1$; moreover $\len(\sbf^0)=C_0=2$.
	Because $\RR(0)$ and $\RR(1)$ are singular, by Proposition~\ref{prop:low-high-singular} we can write
	\begin{align*}
		\low(\RR(0))&=\sum_{i=0}^n\low(\sbf^i)=\sum_{i=k}^n\low(\sbf^i)=\sharp(\st(\RR))+\low(\RR(1))\,,\displaybreak[0]\\
		\high(\RR(0))&=\pow\big(\high(\sbf^n),\high(\sbf^{n-1}),\dots,\high(\sbf^{k}),1,\dots,1\big)\\
			&=\pow\big(\high(\sbf^n),\high(\sbf^{n-1}),\dots,\high(\sbf^{k})\big)=\sharp(\st(\RR))+\high(\RR(1))\,,\qquad\mbox{and}\displaybreak[0]\\
		\len(\RR(0))&=\pow\big(\len(\sbf^n),\len(\sbf^{n-1}),\dots,\len(\sbf^{k}),1,\dots,1,2\big)\\
			&\geq\pow\big(\len(\sbf^n),\len(\sbf^{n-1}),\dots,\len(\sbf^{k}),1,\dots,1,1\big)+1\\
			&=\pow\big(\len(\sbf^n),\len(\sbf^{n-1}),\dots,\len(\sbf^{k}))+1=1+\len(\RR(1))\,,
	\end{align*}
	as required.

	\paragraph*{Case 4}
	Suppose that $D=(\dtpush\gamma,p,D',\mathfrak{D})$.
	Let $\alpha$ and $k$ be such that $\delta(\gamma,p)$ performs $\push^k_\alpha$;
	this transition does not read anything, so $\sharp(\st(\RR))=0$.
	By Definition~\ref{def:types}(\ref{pkt:dt-push}) we have a composer $(\Phi^k,\Phi^{k-1},\dots,\Phi^0;\Psi^k;f)$ such that $\pi_2(\Phi^0)=\{\rd(E)\mid E\in\mathfrak{D}\}$.
	Denote also $\Psi^i=\ass^i(\rd(D'))$ for $i\in\prz{1}{k-1}$.
	By Definition~\ref{def:successor}(\ref{pkt:succ-push}),
	\begin{align*}
		&\RR(1)=\sbf^n:\sbf^{n-1}:\dots:\sbf^{k+1}:\tt^k:\sbf^{k-1}{\restriction}_{\pi_2(\Psi^{k-1})}:\sbf^{k-2}{\restriction}_{\pi_2(\Psi^{k-2})}:\dots:\sbf^1{\restriction}_{\pi_2(\Psi^1)}:(\alpha,\{D'\})\,,
	\end{align*}
	where $\tt^k=\sbf^k{\restriction}_{\pi_2(\Phi^k)}:\sbf^{k-1}{\restriction}_{\pi_2(\Phi^{k-1})}:\dots:\sbf^1{\restriction}_{\pi_2(\Phi^1)}:(\gamma,\mathfrak{D})$.

	From Lemma~\ref{lem:low-high-composer} we obtain the following inequality, which is strict if $f=\pr$:
	\begin{align}
		&\sum_{i=1}^k\low(\sbf^i{\restriction}_{\pi_2(\Phi^i)})+\low((\gamma,\mathfrak{D}))\leq\low(\tt^k)\,.\label{eq:push-uv1}
	\end{align}
	Because $\RR(0)$ is well-formed, $\type(\sbf^i)=\pi_2(\ass^i(\rd(D)))$ for all $i\in\prz{1}{n}$ (cf.~Proposition~\ref{prop:well-formed-multi}).
	By Definition~\ref{def:types}(\ref{pkt:dt-push}), $\ass^k(\sigma)=\Phi^k$, and $\ass^i(\sigma)=\Psi^i\cup\Phi^i$ for $i\in\prz{1}{k-1}$.
	In effect, $\type(\sbf^k)=\pi_2(\Phi^k)$ and $\type(\sbf^i)=\pi_2(\Psi^i)\cup\pi_2(\Phi^i)$ for $i\in\prz{1}{k-1}$, so
	\begin{align}
		\low(\sbf^k)&=\low(\sbf^k{\restriction}_{\pi_2(\Phi^k)})\,,&&\mbox{and}\displaybreak[0]\label{eq:push-low-1}\\
		\low(\sbf^i)&=\sum_{\sigma\in\type(\sbf^i)}\low(\sbf^i{\restriction}_{\{\sigma\}})\nonumber\\
			&\leq\sum_{\sigma\in\pi_2(\Psi^i)}\low(\sbf^i{\restriction}_{\{\sigma\}})+\sum_{\sigma\in\pi_2(\Phi^i)}\low(\sbf^i{\restriction}_{\{\sigma\}})\nonumber\\
			&=\low(\sbf^i{\restriction}_{\pi_2(\Psi^i)})+\low(\sbf^i{\restriction}_{\pi_2(\Phi^i)})&&\mbox{for }i\in\prz{1}{k-1}\,,\label{eq:push-low-2}
	\end{align}
	where the equalities in the second formula are by Proposition~\ref{prop:expand-to-singular}.
	Moreover, if $\pi_2(\Psi^i)\cap\pi_2(\Phi^i)\not\subseteq\calT_\np$ for some $i\in\prz{1}{k-1}$, then the appropriate inequality is strict
	(since the component corresponding to $\tau\in\pi_2(\Psi^i)\cap\pi_2(\Phi^i)\cap\calT_\pr$,
	which is positive by Proposition~\ref{prop:pr-positive}, appears in both sums on the right side, 
	and only once on the left side).
	Because $\RR(0)$ and $\RR(1)$ are singular, 
	\begin{align*}
		\low(\RR(0))&=\sum_{i=0}^n\low(\sbf^i)\,,&\mbox{and}\\
		\low(\RR(1))&=\sum_{i=k+1}^n\low(\sbf^i)+\low(\tt^k)+\sum_{i=1}^{k-1}\low(\sbf^i{\restriction}_{\pi_2(\Psi^i)})+\low((\alpha,\{D'\}))
	\end{align*}
	by Proposition~\ref{prop:low-high-singular}.
	We apply (In)equalities~\eqref{eq:push-low-1} and~\eqref{eq:push-low-2} to the formula for $\low(\RR(0))$; 
	next, we substitute Inequality~\eqref{eq:push-uv1}; we obtain
	\begin{align*}
		\low(\RR(0))&\leq\sum_{i=k+1}^n\low(\sbf^i)+\sum_{i=1}^{k-1}\low(\sbf^i{\restriction}_{\pi_2(\Psi^i)})+\sum_{i=1}^k\low(\sbf^i{\restriction}_{\pi_2(\Phi^i)})+\low(\sbf^0)\leq\\
		&\leq\sum_{i=k+1}^n\low(\sbf^i)+\sum_{i=1}^{k-1}\low(\sbf^i{\restriction}_{\pi_2(\Psi^i)})+\low(\tt^k)-\low((\gamma,\mathfrak{D}))+\low(\sbf^0)=\\
		&=\low(\RR(1))-\low((\alpha,\{D'\}))-\low((\gamma,\mathfrak{D}))+\low(\sbf^0)\leq
			\low(\RR(1))+\low(\sbf^0)\,.
	\end{align*}
	If $\{\rd(D')\}\cup\pi_2(\Phi^0)\not\subseteq\calT_\np$, the last inequality is strict, as we have removed negative components.
	Because $\low(\sbf^0)\leq 1$, if some of the above inequalities was strict, we can remove $\low(\sbf^0)$, and we obtain $\low(\RR(0))\leq\low(\RR(1))$, as required.
	On the other hand, if none of these inequalities was strict, then
	$\pi_2(\Psi^i)\cap\pi_2(\Phi^i)\subseteq\calT_\np$ for each $i\in\prz{1}{k-1}$, and $f=\np$, and $\{\rd(D')\}\cup\pi_2(\Phi^0)\subseteq\calT_\np$;
	from Definition~\ref{def:types}(\ref{pkt:dt-push}) it follows that in this case $\rd(D)\in\calT_\np$, so $\low(\sbf^0)=0$, and we obtain the required inequality as well.

	Next, we prove the inequality for $\high$.
	Denote 
	\begin{align*}
		a_i&=\high(\sbf^i)&&\mbox{for }i\in\prz{k+1}{n}\,,\\
		a_i&=\high(\sbf^i{\restriction}_{\pi_2(\Psi^i)})&&\mbox{for }i\in\prz{1}{k-1}\,,\\
		b_i&=\high(\sbf^i{\restriction}_{\pi_2(\Phi^i)})&&\mbox{for }i\in\prz{1}{k}\,.
	\end{align*}
	Suppose first that $\rd(D)\in\calT_\np$.
	Then, by Definition~\ref{def:types}(\ref{pkt:dt-push}), $\{\rd(D')\}\cup\pi_2(\Phi^0)\subseteq\calT_\np$, and $f=\np$,
	and $\pi_2(\Psi^i)\cap\pi_2(\Phi^i)\subseteq\calT_\np$ for each $i\in\prz{1}{k-1}$. In consequence, $\high(\sbf^0)=\high((\gamma,\mathfrak{D}))=\high((\alpha,\{D'\}))=1$.
	Recall that $\type(\sbf^i)=\pi_2(\Psi^i)\cup\pi_2(\Phi^i)$ for $i\in\prz{1}{k-1}$.
	Thus, by Proposition~\ref{prop:expand-to-singular},
	\begin{align}
		\high(\sbf^i)&=\prod_{\sigma\in\type(\sbf^i)}\high(\sbf^i{\restriction}_{\{\sigma\}})
			=\prod_{\sigma\in\pi_2(\Psi^i)}\high(\sbf^i{\restriction}_{\{\sigma\}})\cdot\prod_{\sigma\in\pi_2(\Phi^i)}\high(\sbf^i{\restriction}_{\{\sigma\}})\hspace{-10em}\nonumber\\
			&=a_i\cdot b_i&&\mbox{for }i\in\prz{1}{k-1}\,;\label{eq:push-eq1}
	\end{align}
	the second equality above holds because for $\sigma\in\pi_2(\Psi^i)\cap\pi_2(\Phi^i)\subseteq\calT_\np$ we have $\high(\sbf^i{\restriction}_{\{\sigma\}})=1$ 
	(cf.~Proposition~\ref{prop:pr-positive}).
	Because $f=\np$, from Lemma~\ref{lem:low-high-composer} we know that
	\begin{align*}
		\pow(b_k,b_{k-1},\dots,b_1,1)=\pow\big(b_k,b_{k-1},\dots,b_1,\high((\gamma,\mathfrak{D}))\big)\geq\high(\tt^k)\,.
	\end{align*}
	Using Proposition~\ref{prop:low-high-singular} and Equalities~\eqref{eq:push-eq1}, then Inequality~\eqref{eq:wl2}, 
	then the above inequality, and then again Proposition~\ref{prop:low-high-singular}, we obtain
	\begin{align*}
		\high(\RR(0))&=\pow(a_n,a_{n-1},\dots,a_{k+1},b_k,a_{k-1}b_{k-1},a_{k-2}b_{k-2},\dots,a_1b_1,1)\\
		&\geq \pow(a_n,a_{n-1},\dots,a_{k+1},\pow(b_k,b_{k-1},\dots,b_1,1),a_{k-1},a_{k-2}\dots,a_1,1)\\
		&\geq\pow(a_n,a_{n-1},\dots,a_{k+1},\high(\tt^k),a_{k-1},a_{k-2},\dots,a_1,1)=\high(\RR(1))\,.
	\end{align*}

	Next, suppose that $\rd(D)\in\calT_\pr$.
	Then Lemma~\ref{lem:low-high-composer} gives us the inequality
	\begin{align}
		\pow\big(b_k,b_{k-1},\dots,b_1,|\calT^0|^n\cdot\high((\gamma,\mathfrak{D}))\big)\geq\high(\tt^k)\,.\label{eq:push-uv3}
	\end{align}
	By definition it holds
	\begin{align*}
		\high(\sbf^0)&=C_{\depth(D)}=(2|\calT^0|)^n\cdot(C_{\depth(D)-1})^{|\calT^0|+1}\\
			&\geq 2^{k-1}\cdot|\calT^0|^n\cdot C_{\depth(D')}\cdot\prod_{E\in\mathfrak{D}}C_{\depth(E)}\\
			&\geq 2^{k-1}\cdot |\calT^0|^n\cdot\high((\alpha,\{D'\}))\cdot\high((\gamma,\mathfrak{D})).
	\end{align*}
	Using Inequality~\eqref{eq:wl3} we replace $2^{k-1}$ in the last argument of $\pow$ by $2$ in the $k-1$ previous arguments;
	then we observe that $(\high(\sbf^i))^2\geq a_ib_i$ for each $i\in\prz{1}{k-1}$; 
	then we use Inequality~\eqref{eq:wl2}, and finally Inequality~\eqref{eq:push-uv3}:
	\begin{align*}
		\high(\RR(0))
		&=\pow\big(\high(\sbf^n),\high(\sbf^{n-1}),\dots,\high(\sbf^1),\high(\sbf^0)\big)\\
		&\geq\pow\big(\high(\sbf^n),\high(\sbf^{n-1}),\dots,\high(\sbf^1),\\
			&\hspace{3.5cm}2^{k-1}\cdot |\calT^0|^n\cdot\high((\alpha,\{D'\}))\cdot\high((\gamma,\mathfrak{D}))\big)\\
		&\geq \pow\big(\high(\sbf^n),\high(\sbf^{n-1}),\dots,\high(\sbf^k),(\high(\sbf^{k-1}))^2,(\high(\sbf^{k-2}))^2,\dots,\\
			&\hspace{3.5cm}(\high(\sbf^1))^2,|\calT^0|^n\cdot\high((\alpha,\{D'\}))\cdot\high((\gamma,\mathfrak{D}))\big)\\
		&\geq \pow\big(a_n,a_{n-1},\dots,a_{k+1},b_k,a_{k-1}b_{k-1},a_{k-2}b_{k-2},\dots,a_1b_1,\\
			&\hspace{3.5cm}|\calT^0|^n\cdot\high((\alpha,\{D'\}))\cdot\high((\gamma,\mathfrak{D}))\big)\\
		&\geq \pow\big(a_n,a_{n-1},\dots,a_{k+1},
			\pow\!\left(b_k,b_{k-1},\dots,b_1,|\calT^0|^n\cdot\high((\gamma,\mathfrak{D}))\right),\\
			&\hspace{3.5cm}a_{k-1},a_{k-2},\dots,a_1,\high((\alpha,\{D'\}))\big)\\
		&\geq\pow\big(a_n,a_{n-1},\dots,a_{k+1},\high(\tt^k),a_{k-1},a_{k-2},\dots,a_1,\high((\alpha,\{D'\}))\big)\\
		&=\high(\RR(1))\,.
	\end{align*}
	
	The inequality for $\len$ can be proved in a very similar way as that for $\high$ in the case $\rd(D)\in\calT_\pr$.
\qed

\subsection{Relating Upper and Lower Bounds}\lab{sec:common-bound}

\makebox[0cm]{\ }%    <--- przekolorowanie nagłówka
As already mentioned, it is meaningful to consider the functions $\low$ and $\high$ because they are closely related: one is bounded if the other is bounded.
This is shown in the following proposition.

\begin{prop}\lab{prop:common-bound}
	There exists a function $H\colon\Nat\to\Nat$ such that for each configuration $c$ and each run descriptor $\sigma\in\type_{\calA,\phi}(c)$ there exists
	a well-formed annotated $n$-stack $\sbf$ for which $\type(\tops^0(\sbf))=\{\sigma\}$, and $\conf(\sbf)=c$, and $\high(\sbf)\leq H(\low(\sbf))$.
\end{prop}

\proof
	Let $d$ be a number such that for each derivation tree there exists a derivation tree with the same conclusion and depth at most $d$;
	such a number exists, because there are only finitely many possible conclusions.
	For each $k\geq 0$ we define a function $N_k\colon\Nat\to\Nat$, and we take $H=N_n$.
	The definition is inductive: $N_k(0)=1$, and, for $L>0$,
	\begin{align*}
		&N_0(L)=(C_d)^{|\calT^0|}\,,\\
		&N_k(L)=\big(\pow(N_k(L-1),N_{k-1}(L))\big)^{|\calT^{k-1}|}&&\mbox{for }k>0\,,
	\end{align*}
	where $C_d$ is the constant from Definition~\ref{def:low-high}.
	
	By definition of a type, for each configuration $c$ and each run descriptor $\sigma\in\type_{\calA,\phi}(c)$ there exists a well-formed annotated $n$-stack $\sbf$ such that
	$\type(\tops^0(\sbf))=\{\sigma\}$ and $\conf(\sbf)=c$.
	We can assume without loss of generality that all derivation trees in $\sbf$ have depth at most $d$: 
	we can safely replace each tree by another (smaller) tree having the same conclusion.
	Thus, it is enough to prove that for each well-formed annotated $k$-stack $\sbf$, in which all derivation trees have depth at most $d$, 
	it holds $\high(\sbf)\leq N_k(\low(\sbf))$.
	
	Denote $L=\low(\sbf)$.
	If $L=0$ then $\high(\sbf)=1=N_k(L)$, thanks to Proposition~\ref{prop:pr-positive}.
	Suppose that $L>0$.
	In this case we prove the thesis by induction on the structure of $\sbf$.
	For a stack $\sbf=(\gamma,\mathfrak{D})$ of order $0$ it holds
	\begin{align*}
		\high(\sbf)\leq\prod_{D\in\mathfrak{D}}C_{\depth(D)}\leq(C_d)^{|\calT^0|}=N_0(L)\,.
	\end{align*}

	Next, consider a stack $\sbf=\sbf^k:\sbf^{k-1}$.
	Recall that $\low(\sbf)$ equals the sum of $\low$ for $\sbf^k{\restriction}_{\pi_2(\ass^k(\sigma))}$ and $\sbf^{k-1}{\restriction}_{\{\sigma\}}$
	over all $\sigma\in\type(\sbf^{k-1})$.
	We have two cases.
	One possibility is that $\low(\sbf^k{\restriction}_{\pi_2(\ass^k(\sigma))})=L$ for some $\sigma\in\type(\sbf^{k-1})$.
	Then $\low$ for all other considered stacks is $0$, so their $\high$ is $1$.
	Using the induction assumption we obtain
	\begin{align*}
		\high(\sbf)\leq\pow(N_k(L),1)\cdot\prod_{\tau\in\type(\sbf^{k-1})\setminus\{\sigma\}}\pow(1,1)=N_k(L)\,.
	\end{align*}
	The opposite situation is that $\low(\sbf^k{\restriction}_{\pi_2(\ass^k(\sigma))})\leq L-1$ for each $\sigma\in\type(\sbf^{k-1})$.
	Observing that $N_k$ is monotone, by the induction assumption 
	$\high(\sbf^k{\restriction}_{\pi_2(\ass^k(\sigma))})\leq N_k(L-1)$ and $\high(\sbf^{k-1}{\restriction}_{\{\sigma\}})\leq N_{k-1}(L)$ for each $\sigma\in\type(\sbf^{k-1})$,
	so we obtain
	\begin{align*}
		\high(\sbf)\leq\prod_{\sigma\in\type(\sbf^{k-1})}\pow(N_k(L-1),N_{k-1}(L))\leq N_k(L)\,.\tag*{\qed}
	\end{align*}

\subsection{Assumptions Are Used in Returns}\lab{sec:assumptions-used}

Our next goal is to formally prove that whenever an assumption of a run descriptor is used in an annotated run, then we have a return.

\begin{lem}\lab{lem:rd2run}
	Let $\sbf=\sbf^n:\sbf^{n-1}:\dots:\sbf^0$ be a well-formed singular annotated $n$-stack, where $\type(\sbf^0)=\{\sigma\}$.
	If $(m,\xi)\in\ass^r(\sigma)$, then there exists an annotated run $\RR$ starting in $\sbf$ such that $\st(\RR)$ is an $r$-return, $\phi(\st(\RR))=m$,
	and $\tops^r(\RR(|\RR|))=\sbf^r{\restriction}_{\{\xi\}}$.
\end{lem}

\proof
	We use induction on $\len(\sbf)$.
	Thanks to Lemma~\ref{lem:low-high-len} we can always use the induction assumption for the successor of $\sbf$.
	We have several cases depending on the shape of the derivation tree $D$ in $\sbf^0$ (that is, on the first operation in an annotated run starting in $\sbf$).
	We use the characterization of returns from Proposition~\ref{prop:return}.
	
	\paragraph*{Case 1}
	If $D=(\dtempty\gamma,p)$ then $\ass^r(\sigma)=\emptyset$, so the assumptions cannot hold.
	
	\paragraph*{Case 2}
	Suppose that $D=(\dtread p,D')$.
	Then the successor $\tt$ of $\sbf$ differs from $\sbf$ only in the topmost $0$-stack;
	the new topmost $0$-stack has type $\{\tau\}$ such that $\ass^r(\sigma)=\phi(a)\circ\ass^r(\tau)$,
	where $a$ is the letter read by the step between $\conf(\sbf)$ and $\conf(\tt)$.
	Because $(m,\xi)\in\ass^r(\sigma)$, there exists $m'$ such that $m=\phi(a)\cdot m'$ and $(m',\xi)\in\ass^r(\tau)$.
	By the induction assumption for $\tt$,
	there exists an annotated run $\SS$ starting in $\tt$ such that $\st(\SS)$ is an $r$-return, $\phi(\st(\SS))=m'$, 
	and $\tops^r(\SS(|\SS|))=\sbf^r{\restriction}_{\{\xi\}}$.
	Together with the step between $\sbf$ and $\tt$, it gives us an annotated run as required.
	
	\paragraph*{Case 3}
	Suppose that $D=(\dtpop\gamma,p,\tau)$, where $\tau\in\calT^k$.
	The successor of $\sbf$ is $\tt=\sbf^n:\sbf^{n-1}:\dots:\sbf^k$.
	Recall that $\ass^i(\sigma)=\emptyset$ for $i<k$, so $r\geq k$.
	Moreover, $\ass^k(\sigma)=\{(\mathbf{1}_M,\tau)\}$, and $\type(\sbf^k)=\pi_2(\ass^k(\sigma))=\{\tau\}$ by well-formedness of $\sbf$ (cf.~Proposition~\ref{prop:well-formed-multi}).
	If $r=k$, then $(m,\xi)=(\mathbf{1}_M,\tau)$.
	In this case the annotated run of length $1$ satisfies the thesis.
	Otherwise $r>k$, and $(m,\xi)\in\ass^r(\sigma)=\ass^r(\tau)$.
	Then as well $(m,\xi)\in\ass^r(\tau')$, where $\tau'$ is the run descriptor in the type of $\tops^0(\sbf^k)$ (since $\tau=\red^k(\tau')$).
	The induction assumption for $\tt$ gives us an annotated run $\SS$ starting in $\tt$ such that $\st(\SS)$ is an $r$-return, $\phi(\st(\SS))=m$, 
	and $\tops^r(\SS(|\SS|))=\sbf^r{\restriction}_{\{\xi\}}$.
	Together with the step between $\sbf$ and $\tt$, it gives us an annotated run as required.
	
	\paragraph*{Case 4}
	Suppose that $D=(\dtpush\gamma,p,D',\mathfrak{D})$.
	Let $\alpha$, $k$, $\Psi^i$, $\Phi^i$ be as in Definition~\ref{def:types}(\ref{pkt:dt-push}).
	The successor of $\sbf$ is
	\begin{align*}
		&\tt=\sbf^n:\sbf^{n-1}:\dots:\sbf^{k+1}:\tt^k:\sbf^{k-1}{\restriction}_{\pi_2(\Psi^{k-1})}:\sbf^{k-2}{\restriction}_{\pi_2(\Psi^{k-2})}:\dots:\sbf^1{\restriction}_{\pi_2(\Psi^1)}:(\alpha,\{D'\})\,,
	\end{align*}
	where $\tt^k=\sbf^k{\restriction}_{\pi_2(\Phi^k)}:\sbf^{k-1}{\restriction}_{\pi_2(\Phi^{k-1})}:\dots:\sbf^1{\restriction}_{\pi_2(\Phi^1)}:(\gamma,\mathfrak{D})$.
	Recall from our previous proofs that $\type(\sbf^i)=\pi_2(\Psi^i)$ for $i\in\prz{k+1}{n}$, and $\type(\tt^k)=\pi_2(\Psi^k)$.
	If $(m,\xi)\in\Psi^r$ and $r\neq k$, then the induction assumption for $\tt$ gives us an annotated run $\SS$ starting in $\tt$ such that $\st(\SS)$ is an $r$-return, $\phi(\st(\SS))=m$, 
	and $\tops^r(\SS(|\SS|))=\sbf^r{\restriction}_{\{\xi\}}$; together with the step between $\sbf$ and $\tt$, it gives us an annotated run as required.
	Otherwise $(m,\xi)\in\Phi^r$ and $r\leq k$.
	Recall that we have a composer $(\Phi^k,\Phi^{k-1},\dots,\Phi^0;\Psi^k;f)$.
	The definition of a composer (Condition~\ref{pkt:compo-ass}) gives us some $(m_1,\tau)\in\Phi^0$ and $m_2\in M$ such that $m=m_1\cdot m_2$ and $(m_2,\xi)\in\ass^r(\tau)$.
	We see that $(m_1,\red^k(\tau))\in\Psi^k$ (by Condition~\ref{pkt:compo-red} of the definition).
	The induction assumption for $\tt$ and $(m_1,\red^k(\tau))\in\Psi^k$ gives us an annotated run $\SS$ starting in $\tt$ such that $\st(\SS)$ is a $k$-return, $\phi(\st(\SS))=m_1$, 
	and $\tops^k(\SS(|\SS|))=\tt^k{\restriction}_{\{\red^k(\tau)\}}$.
	By Proposition~\ref{prop:restriction},
	\begin{align*}
		\tt^k{\restriction}_{\{\red^k(\tau)\}}=\sbf^k{\restriction}_{\pi_2(\ass^k(\tau))}:\sbf^{k-1}{\restriction}_{\pi_2(\ass^{k-1}(\tau))}:\dots:\sbf^1{\restriction}_{\pi_2(\ass^1(\tau))}:(\gamma,\mathfrak{D}){\restriction}_{\{\tau\}}\,.
	\end{align*}
	Recalling that $r\leq k$, 
	the induction assumption for $\SS(|\SS|)$ and $(m_2,\xi)\in\ass^r(\tau)$ gives us an annotated run $\TT$ starting in $\SS(|\SS|)$ such that $\st(\TT)$ is an $r$-return, $\phi(\st(\TT))=m_2$,
	and $\tops^r(\TT(|\TT|))=\sbf^r{\restriction}_{\{\xi\}}$.
	The step between $\sbf$ and $\tt$ composed with $\SS$ and then with $\TT$ gives us an annotated run as required.
\qed

\subsection{Completeness of Types}\lab{sec:completeness}

In the previous subsection we have proved soundness of the type system, which means that if a run descriptor is contained in the type of a configuration
then a corresponding run exists from this configuration.
As usual, we need the opposite direction (completeness) as well;
that is, having a run from a configuration, we want to imply that the corresponding run descriptor is in the type of this configuration.
This is shown in Lemma~\ref{lem:run2rd}, which is a converse of Lemma~\ref{lem:rd2run}.
While reversing Lemma~\ref{lem:rd2run} we have to remember that not every run can be extended to an annotated run, 
so in Lemma~\ref{lem:run2rd} we need to use runs, not annotated runs.

\begin{lem}\lab{lem:run2rd}
	Let $R$ be an $r$-return, and let $\xi\in\type_{\calA,\phi}(R(|R|))$.
	Then there exists a run descriptor $\sigma\in\type_{\calA,\phi}(R(0))$ such that $(\phi(R),\red^r(\xi))\in\ass^r(\sigma)$.
\end{lem}

Before proving Lemma~\ref{lem:run2rd}, we first state four auxiliary lemmas.
These lemmas are used not only in the proof of Lemma~\ref{lem:run2rd}, but also in the next subsection.
Actually, every of these lemmas has a part denoted by $(\star)$; these parts are needed only in the next subsection.

\begin{lem}\lab{lem:run2rd-read}
	Let $R$ be a run of length $1$ whose transition is $\read$, and let $\tau\in\type_{\calA,\phi}(R(1))$.
	Then there exists $\sigma\in\type_{\calA,\phi}(R(0))$ such that $\ass^i(\sigma)=\phi(R)\circ\ass^i(\tau)$ for each $i\in\prz{1}{n}$.
	Moreover, there exists a well-formed singular annotated $0$-stack $\vv^0$ such that $\type(\vv^0)=\{\sigma\}$, and the following is satisfied.
	\begin{itemize}[label={$(\star)$}]
	\item	Let $\sbf'$ be a well-formed annotated $n$-stack with $\tops^0(\sbf')=\vv^0$.
		Then there exists an annotated run $\SS$ of length $1$ such that $\SS(0)=\sbf'$, the transition of $\st(\SS)$ is $\read$, 
		it reads the same letter as the transition of $R$, and $\type(\tops^0(\SS(1)))=\{\tau\}$.
	\end{itemize}
\end{lem}

\proof
	By definition of $\type_{\calA,\phi}$, 
	there exists a well-formed annotated stack $\sbf=\sbf^n:\sbf^{n-1}:\dots:\sbf^1:(\gamma,\{D'\})$ such that $\rd(D')=\tau$ and $\conf(\sbf)=R(1)$.
	Well-formedness of $\sbf$ implies that $\type(\sbf^i)=\pi_2(\ass^i(\tau))$ for each $i\in\prz{1}{n}$ (cf.~Proposition~\ref{prop:well-formed-multi}).
	When $p$ is the state of $R(0)$, Definition~\ref{def:types}(\ref{pkt:dt-read}) implies that $D=(\dtread p,D')$ is a derivation tree with conclusion $\gamma\vdash\sigma$, where
	$\sigma=(p,\Phi^n,\Phi^{n-1},\dots,\Phi^1,f)$ and $\Phi^i=\phi(R)\circ\ass^i(\tau)$ for each $i\in\prz{1}{n}$.
	Because $\pi_2(\ass^i(\sigma))=\pi_2(\ass^i(\tau))$, 
	the annotated stack $\sbf^n:\sbf^{n-1}:\dots:\sbf^1:(\gamma,\{D\})$ is well-formed (again, cf.~Proposition~\ref{prop:well-formed-multi}),
	so $\sigma\in\type_{\calA,\phi}(R(0))$.
	
	In order to prove Property~$(\star)$, as $\vv^0$ we take $(\gamma,\{D\})$.
	Clearly $\type(\vv^0)=\{\rd(D)\}=\{\sigma\}$.
	Consider now any well-formed annotated $n$-stack $\sbf'$ with $\tops^0(\sbf')=\vv^0$.
	Let $\SS$ be the annotated run from $\sbf'$ to its successor.
	Because the topmost $0$-stack of $\sbf'$ is annotated by $D=(\dtread p,D')$, the successor of $\sbf'$ indeed exists,
	and the transition of $\st(\SS)$ is $\read$.
	The state of $R(1)$ and the state of $\conf(\SS(1))$ are the same (namely $\pi_1(\rd(D'))$), so $\st(\SS)$ reads the same letter as $R$.
	Moreover, $\tops^0(\SS(1))=(\gamma,\{D'\})$, so $\type(\tops^0(\SS(1)))=\{\rd(D')\}=\{\tau\}$, as required.
\qed

\begin{lem}\lab{lem:run2rd-pop}
	Let $R$ be a run of length $1$ performing $\pop^k$, and let $\tau\in\type_{\calA,\phi}(R(1))$.
	Then there exists $\sigma\in\type_{\calA,\phi}(R(0))$ such that $\ass^i(\sigma)=\ass^i(\tau)$ for each $i\in\prz{k+1}{n}$, 
	and $\ass^k(\sigma)=\{(\mathbf{1}_M,\red^k(\tau))\}$.
	Moreover, there exists a well-formed singular annotated $k$-stack $\vv^k$ such that $\type(\tops^0(\vv^k))=\{\sigma\}$, and the following is satisfied.
	\begin{itemize}[label={$(\star)$}]
	\item	Let $\sbf'$ be a well-formed annotated $n$-stack with $\tops^k(\sbf')=\vv^k$.
		Then there exists an annotated run $\SS$ of length $1$ such that $\SS(0)=\sbf'$, and $\st(\SS)$ performs $\pop^k$, and $\type(\tops^0(\SS(1)))=\{\tau\}$.
	\end{itemize}
\end{lem}

\proof
	Denote $R(0)=(p,s^n:s^{n-1}:\dots:s^0)$; then $\pi_2(R(1))=s^n:s^{n-1}:\dots:s^k$.
	By definition of $\type_{\calA,\phi}$, there exists a well-formed annotated stack $\sbf=\sbf^n:\sbf^{n-1}:\dots:\sbf^k$ 
	such that $\type(\tops^0(\sbf^k))=\{\tau\}$, 
	and $\st(\sbf^i)=\posl(s^i)$ for each $i\in\prz{k}{n}$.
	Then, by Proposition~\ref{prop:well-formed-multi}, $\type(\sbf^k)=\{\red^k(\tau)\}$.
	Well-formedness of $\sbf$ implies that $\type(\sbf^i)=\pi_2(\ass^i(\tau))$ for each $i\in\prz{k+1}{n}$ (cf.~Proposition~\ref{prop:well-formed-multi}).
	For $i\in\prz{1}{k-1}$ let $\sbf^i$ be the well-formed annotated $i$-stack such that $\type(\sbf^i)=\emptyset$ and $\st(\sbf^i)=\posl(s^i)$ (we annotate $s^i$ by empty sets).
	Finally, by Definition~\ref{def:types}(\ref{pkt:dt-pop}), $D=(\dtpop s^0,p,\red^k(\tau))$ is a derivation tree with conclusion $\gamma\vdash\sigma$, 
	where $\gamma=\posl(s^0)$ and
	\begin{align*}
		\sigma=(p,\ass^n(\tau),\ass^{n-1}(\tau),\dots,\ass^{k+1}(\tau),\{(\mathbf{1}_M,\red^k(\tau)\},\emptyset,\dots,\emptyset,\np)\,,
	\end{align*}
	so $\sbf^0=(\gamma,\{D\})$ has type $\{\sigma\}$.
	Using Proposition~\ref{prop:well-formed-multi} we observe that $\sbf^n:\sbf^{n-1}:\dots:\sbf^0$ is well-formed, so $\sigma\in\type_{\calA,\phi}(R(0))$.
	
	In order to prove Property~$(\star)$, as $\vv^k$ we take $\sbf^k:\sbf^{k-1}:\dots:\sbf^0$.
	Clearly $\type(\tops^0(\vv^k))=\type(\sbf^0)=\{\sigma\}$.
	Consider now any well-formed annotated $n$-stack $\sbf'$ with $\tops^k(\sbf')=\vv^k$.
	Let $\SS$ be the annotated run from $\sbf'$ to its successor.
	Because the topmost $0$-stack of $\sbf'$ is annotated by $D$, the successor of $\sbf'$ indeed exists, and $\st(\SS)$ performs $\pop^k$.
	Moreover, $\tops^k(\SS(1))=\sbf^k$, so $\type(\tops^0(\SS(1)))=\type(\tops^0(\sbf^k))=\{\tau\}$.
\qed

\begin{lem}\lab{lem:run2rd-push}
	Let $R$ be a run of length $1$ performing $\push^k_\alpha$, and let $\tau\in\type_{\calA,\phi}(R(1))$.
	Then there exists $\sigma\in\type_{\calA,\phi}(R(0))$ such that $\ass^i(\tau)\subseteq\ass^i(\sigma)$ for each $i\in\prz{1}{n}\setminus\{k\}$.
	Moreover, there exists a well-formed singular annotated $0$-stack $\vv^0$ such that $\type(\vv^0)=\{\sigma\}$, and the following is satisfied.
	\begin{itemize}[label={$(\star)$}]
	\item	Let $\sbf'$ be a well-formed annotated $n$-stack with $\tops^0(\sbf')=\vv^0$.
		Then there exists an annotated run $\SS$ of length $1$ such that $\SS(0)=\sbf'$, and $\st(\SS)$ performs $\push^k_\alpha$, and $\type(\tops^0(\SS(1)))=\{\tau\}$.
	\end{itemize}
\end{lem}

\proof
	Before starting the actual proof, we observe that for each pair of well-formed annotated stacks $\sbf$, $\tt$ such that $\st(\sbf)=\st(\tt)$
	we can construct a well-formed annotated stack $\sbf\oplus\tt$ whose type is $\type(\sbf)\cup\type(\tt)$ and such that $\st(\sbf\oplus\tt)=\st(\sbf)$.
	We construct $\sbf\oplus\tt$ by induction on the structure of $\sbf$.
	Denote $\widetilde\Psi=\type(\tt)\setminus\type(\sbf)$.
	If $\sbf$ is of order $0$, then we take $\sbf\oplus\tt=(\gamma,\mathfrak{D}\cup\mathfrak{D}')$, where $\sbf=(\gamma,\mathfrak{D})$ and $\tt{\restriction}_{\widetilde\Psi}=(\gamma,\mathfrak{D}')$.
	If $\sbf=\tt=[\,]$, then $\sbf\oplus\tt=[\,]$ is fine.
	If $\sbf=\sbf^j:\sbf^{j-1}$ and $\tt{\restriction}_{\widetilde\Psi}=\tt^j:\tt^{j-1}$, then as $\sbf\oplus\tt$ we take $(\sbf^j\oplus\tt^j):(\sbf^{j-1}\oplus\tt^{j-1})$;
	observe that it is well-formed, because the types of $\sbf$ and $\tt{\restriction}_{\widetilde\Psi}$ are disjoint.

	Denote $R(0)=(p,s^n:s^{n-1}:\dots:s^1:(\gamma,x))$; then $\pi_2(R(1))$ equals
	\begin{align*}
		s^n:s^{n-1}:\dots:s^{k+1}:(s^k:s^{k-1}:\dots:s^1:(\gamma,x)):\posp(s^{k-1}:s^{k-2}:\dots:s^1:(\alpha,x))\,.
	\end{align*}
	By definition of $\type_{\calA,\phi}$, there exists a well-formed annotated stack $\sbf=\sbf^n:\sbf^{n-1}:\dots:\sbf^1:(\alpha,\{D'\})$ 
	in which $\sbf^k=\tt^k:\tt^{k-1}:\dots:\tt^1:(\gamma,\mathfrak{D})$,
	such that $\rd(D')=\tau$, and $\st(\sbf^i)=\posl(s^i)$ for each $i\in\prz{1}{n}\setminus\{k\}$, and $\st(\tt^i)=\posl(s^i)$ for each $i\in\prz{1}{k}$.
	Denote $\Psi^i=\ass^i(\tau)$ for each $i\in\prz{1}{n}$.
	Well-formedness of $\sbf$ implies that $\type(\sbf^i)=\pi_2(\Psi^i)$ for each $i\in\prz{1}{n}$ (cf.~Proposition~\ref{prop:well-formed-multi}),
	and, thanks to Proposition~\ref{prop:composer}, there exists a composer $(\Phi^k,\Phi^{k-1},\dots,\Phi^0;\Psi^k;f)$ such that $\type(\tt^i)=\pi_2(\Phi^i)$ for each $i\in\prz{1}{k}$
	and $\{\rd(E)\mid E\in\mathfrak{D}\}=\pi_2(\Phi^0)$.
	By Definition~\ref{def:types}(\ref{pkt:dt-push}), $D=(\dtpush\gamma,p,D',\mathfrak{D})$ is a derivation tree with conclusion $\gamma\vdash\sigma$, where
	\begin{align*}
		\sigma=(p,\Psi^n,\Psi^{n-1},\dots,\Psi^{k+1},\Phi^k,\Psi^{k-1}\cup\Phi^{k-1},\Psi^{k-2}\cup\Phi^{k-2},\dots,\Psi^1\cup\Phi^1,g)\,.
	\end{align*}
	Using Proposition~\ref{prop:well-formed-multi} we observe that the annotated stack
	\begin{align*}
		\sbf^n:\sbf^{n-1}:\dots:\sbf^{k+1}:\tt^k:(\sbf^{k-1}\oplus\tt^{k-1}):(\sbf^{k-2}\oplus\tt^{k-2}):\dots:(\sbf^1\oplus\tt^1):(\gamma,\{D\})
	\end{align*}
	is well-formed, so $\sigma\in\type_{\calA,\phi}(R(0))$.
	
	In order to prove Property~$(\star)$, as $\vv^0$ we take $(\gamma,\{D\})$.
	Clearly $\type(\vv^0)=\{\rd(D)\}=\{\sigma\}$.
	Consider now any well-formed annotated $n$-stack $\sbf'$ with $\tops^0(\sbf')=\vv^0$.
	Let $\SS$ be the annotated run from $\sbf'$ to its successor.
	Because the topmost $0$-stack of $\sbf'$ is annotated by $D$, the successor of $\sbf'$ indeed exists, and $\st(\SS)$ performs $\push^k_\alpha$.
	Moreover, $\tops^0(\SS(1))=(\alpha,\{D'\})$, so $\type(\tops^0(\SS(1)))=\{\rd(D')\}=\{\tau\}$.
\qed

\begin{lem}\lab{lem:run2rd-push-ret}
	Let $R$ be a run in which $\subrun{R}{0}{1}$ performs $\push^k_\alpha$ and $\subrun{R}{1}{|R|}$ is a $k$-return, let $\tau\in\type_{\calA,\phi}(R(|R|))$,
	and let $\rho\in\type_{\calA,\phi}(R(1))$ be such that $(\phi(R),\red^k(\tau))\in\ass^k(\rho)$.
	Then there exists $\sigma\in\type_{\calA,\phi}(R(0))$ such that $\phi(R)\circ\ass^i(\tau)\subseteq\ass^i(\sigma)$ for each $i\in\prz{1}{k}$.
	Moreover, there exists a well-formed singular annotated $0$-stack $\vv^0$ such that $\type(\vv^0)=\{\sigma\}$, and the following is satisfied.
	\begin{itemize}[label={$(\star)$}]
	\item	Let $\sbf'$ be a well-formed annotated $n$-stack with $\tops^0(\sbf')=\vv^0$.
		Then there exists an annotated run $\SS$ of length $1$ such that $\SS(0)=\sbf'$, and $\st(\SS)$ performs $\push^k_\alpha$, and $\type(\tops^0(\SS(1)))=\{\rho\}$,
		and $\tau\in\type(\tops^0(\pop^k(\SS(1))))$.
	\end{itemize}
\end{lem}

\proof
	Denote $R(0)=(p,s^n:s^{n-1}:\dots:s^1:(\gamma,x))$; then $\pi_2(R(1))$ is as in the previous lemma,
	and $\tops^k(R(0))\cong\tops^k(R(|R|))$ (due to Proposition~\ref{prop:push-return-jest-fajny}).
	The definition of $\type_{\calA,\phi}$ gives us a well-formed annotated $n$-stack $\uu$ such that $\type(\tops^0(\uu))=\{\tau\}$ 
	and $\conf(\uu)=R(|R|)$,
	and a well-formed annotated stack $\sbf=\sbf^n:\sbf^{n-1}:\dots:\sbf^1:(\alpha,\{D'\})$
	such that $\rd(D')=\rho$ and $\conf(\sbf)=R(1)$.
	Denote $\uu^k=\tops^k(\uu)$.
	Then $\type(\tops^0(\uu^k))=\{\tau\}$, and $\st(\uu^k)=\posl(s^k:s^{k-1}:\dots:s^1:(\gamma,x))$, and $\st(\sbf^i)=\posl(s^i)$ for each $i\in\prz{1}{n}\setminus\{k\}$,
	and $\st(\sbf^k)=\posl(s^k:s^{k-1}:\dots:s^1:(\gamma,x))$.
	Denote $\Psi^i=\ass^i(\rho)$ for each $i\in\prz{1}{n}$.
	Well-formedness of $\sbf$ implies that $\type(\sbf^i)=\pi_2(\Psi^i)$ for each $i\in\prz{1}{n}$ (cf.~Proposition~\ref{prop:well-formed-multi}).
	By Proposition~\ref{prop:well-formed-multi}, $\type(\uu^k)=\{\red^k(\tau)\}$.
	Thanks to the assumption $\red^k(\tau)\in\pi_2(\ass^k(\rho))$ we have that $\type(\uu^k)\subseteq\pi_2(\Psi^k)$. 
	In effect, the annotated stack $\uu^k\oplus\sbf^k$ has type $\pi_2(\Psi^k)$, 
	equal to the type of $\sbf^k$, but additionally $\tau\in\type(\tops^0(\uu^k\oplus\sbf^k))$ 
	(recalling the construction from the previous proof, 
	we see that to $\uu^k\oplus\sbf^k$ we take all annotations from $\uu^k$ and some annotations from $\sbf^k$).
	Denote $\uu^k\oplus\sbf^k=\tt^k:\tt^{k-1}:\dots:\tt^1:(\gamma,\mathfrak{D})$.
	By Proposition~\ref{prop:composer} we have a composer $(\Phi^k,\Phi^{k-1},\dots,\Phi^0;\Psi^k;f)$ such that $\type(\tt^i)=\pi_2(\Phi^i)$ for each $i\in\prz{1}{k}$
	and $\{\rd(E)\mid E\in\mathfrak{D}\}=\pi_2(\Phi^0)$.
	Because $\tau\in\pi_2(\Phi^0)$ and $(\phi(R),\red^k(\tau))\in\Psi^k$, it holds $(\phi(R),\tau)\in\Phi^0$ 
	(thanks to Conditions~\ref{pkt:compo-red} and~\ref{pkt:compo-inj} of the definition of a composer), 
	which implies $\phi(R)\circ\ass^i(\tau)\subseteq\Phi^i$ for each $i\in\prz{1}{k}$ (thanks to Condition~\ref{pkt:compo-ass} of the definition).
	By Definition~\ref{def:types}(\ref{pkt:dt-push}), $D=(\dtpush\gamma,p,D',\mathfrak{D})$ is a derivation tree with conclusion $\gamma\vdash\sigma$, where
	\begin{align*}
		\sigma=(p,\Psi^n,\Psi^{n-1},\dots,\Psi^{k+1},\Phi^k,\Psi^{k-1}\cup\Phi^{k-1},\Psi^{k-2}\cup\Phi^{k-2},\dots,\Psi^1\cup\Phi^1,g)\,.
	\end{align*}
	We observe that the annotated stack
	\begin{align*}
		\sbf^n:\sbf^{n-1}:\dots:\sbf^{k+1}:\tt^k:(\sbf^{k-1}\oplus\tt^{k-1}):(\sbf^{k-2}\oplus\tt^{k-2}):\dots:(\sbf^1\oplus\tt^1):(\gamma,\{D\})
	\end{align*}
	is well-formed (by Proposition~\ref{prop:well-formed-multi}), so $\sigma\in\type_{\calA,\phi}(R(0))$.

	In order to prove Property~$(\star)$, as $\vv^0$ we take $(\gamma,\{D\})$.
	Clearly $\type(\vv^0)=\{\rd(D)\}=\{\sigma\}$.
	Consider now any well-formed annotated $n$-stack $\sbf'$ with $\tops^0(\sbf')=\vv^0$.
	Let $\SS$ be the annotated run from $\sbf'$ to its successor.
	Because the topmost $0$-stack of $\sbf'$ is annotated by $D$, the successor of $\sbf'$ indeed exists, and $\st(\SS)$ performs $\push^k_\alpha$.
	Moreover, $\tops^0(\SS(1))=(\alpha,\{D'\})$, so $\type(\tops^0(\SS(1)))=\{\rd(D')\}=\{\rho\}$.
	On the other hand, $\tops^0(\pop^k(\SS(1)))$ is $(\gamma,\mathfrak{D})$, 
	so its type is $\{\rd(E)\mid E\in\mathfrak{D}\}=\pi_2(\Phi^0)$, and we know that $\tau\in\pi_2(\Phi^0)$.
\qed

\proof[Proof of Lemma~\ref{lem:run2rd}]
	Recall that we are given an $r$-return $R$, and a run descriptor $\xi\in\type_{\calA,\phi}(R(|R|))$,
	and we have to show existence of a run descriptor $\sigma\in\type_{\calA,\phi}(R(0))$ such that $(\phi(R),\red^r(\xi))\in\ass^r(\sigma)$.
	We use induction on the length of the $r$-return $R$.
	Proposition~\ref{prop:return} gives us possible forms of $R$; we analyze these cases.

	Suppose first that $|R|=1$ and the only transition of $R$ performs $\pop^r$.
	We take $\sigma$ from Lemma~\ref{lem:run2rd-pop}, where we take $\xi$ as $\tau$ and $r$ as $k$.
	By assumption $\phi(R)=\mathbf{1}_M$, so $(\phi(R),\red^r(\xi))\in\ass^r(\sigma)$.

	Next, suppose that $\subrun{R}{1}{|R|}$ is an $r$-return, and the first transition of $R$ is $\read$, or performs $\pop^k$ for $k<r$, or $\push^k_\alpha$ for $k\neq r$.
	The induction assumption for $\subrun{R}{1}{|R|}$ gives us a run descriptor $\tau\in\type_{\calA,\phi}(R(1))$ such that $(\phi(\subrun{R}{1}{|R|}),\red^r(\xi))\in\ass^r(\tau)$,
	and Lemma~\ref{lem:run2rd-read}, or \ref{lem:run2rd-pop}, or \ref{lem:run2rd-push}, respectively, 
	used for $\subrun{R}{0}{1}$ gives us a run descriptor $\sigma\in\type_{\calA,\phi}(R(0))$ such that $\phi(\subrun{R}{0}{1})\circ\ass^r(\tau)\subseteq\ass^r(\sigma)$ 
	(where $\phi(\subrun{R}{0}{1})$ may be nontrivial only when the transition is $\read$).

	Finally, suppose that the first transition of $R$ performs $\push^k_\alpha$ for $k\geq r$ and $\subrun{R}{1}{|R|}=S\circ T$ for some $k$-return $S$ and $r$-return $T$.
	The induction assumption for $T$ gives us a run descriptor $\tau\in\type_{\calA,\phi}(T(0))$ such that $(\phi(T),\red^r(\xi))\in\ass^r(\tau)$, and
	the induction assumption for $S$ gives us a run descriptor $\rho\in\type_{\calA,\phi}(R(1))$ such that $(\phi(S),\red^k(\tau))\in\ass^k(\rho)$.
	Using Lemma~\ref{lem:run2rd-push-ret} for $\subrun{R}{0}{1}\circ S$ we obtain a run descriptor $\sigma\in\type_{\calA,\phi}(R(0))$ such that $\phi(S)\circ\ass^r(\tau)\subseteq\ass^r(\sigma)$ 
	(recalling that $r\leq k$), so $(\phi(R),\red^r(\xi))=(\phi(S)\circ\phi(T),\red^r(\xi))\in\phi(S)\circ\ass^r(\tau)\subseteq\ass^r(\sigma)$.
\qed

\subsection{Reproducing Upper Runs}\lab{sec:copy-upper}

Till now we were using types to describe returns from a configuration, but thanks to the decomposition given by Proposition~\ref{prop:upper} we can also describe $r$-upper runs.
This is stated in the following lemma.

\begin{lem}\lab{lem:rownowazne}
	Let $R$ be an $r$-upper run (where $r\in\prz{0}{n}$), and let $\tau\in\type_{\calA,\phi}(R(|R|))$.
	Then there exists a run descriptor $\sigma\in\type_{\calA,\phi}(R(0))$
	and a monotone function $f_R\colon\mathbb{N}\to\mathbb{N}$ such that the following is satisfied.
	\begin{itemize}[label={$(\star)$}]
	\item	Let $\sbf$ be a well-formed annotated $n$-stack such that $\type(\tops^0(\sbf))=\{\sigma\}$ 
		and $\tops^r(\conf(\sbf))\allowbreak\cong\tops^r(R(0))$.
		Then there exists a well-formed annotated $n$-stack $\tt$ such that $\type(\tops^0(\tt))\allowbreak=\{\tau\}$,
		and there exists a run $S$ from $\conf(\sbf)$ to $\conf(\tt)$ that is $(r,\phi)$-parallel to $R$, and such that
		\begin{align*}
			\low(\sbf)\leq f_R(\sharp(S)+\low(\tt))\,,&&\mbox{and}&& f_R(\high(\sbf))\geq\sharp(S)+\high(\tt)\,.
		\end{align*}
	\end{itemize}
\end{lem}

The idea staying behind a proof of this lemma is that the run $R$ can be split into parts of two kinds.
First, we have parts for which the topmost $r$-stack of $R(0)$ is responsible.
Since in $\conf(\sbf)$ the topmost $r$-stack is the same, we can execute them also from $\conf(\sbf)$.
Second, we have parts not controlled by the topmost $r$-stack of $R(0)$, but (according to Proposition~\ref{prop:upper}) these are returns.
Analogous returns can be executed from $\conf(\sbf)$, because of the run descriptor $\sigma$ in $\type(\tops^0(\sbf))$.

Two aspects of the statement of the lemma can be understood basing on the above idea.
First, the correction function $f_R$ is really needed (i.e., the lemma would be false if the identity function was always taken as $f_R$).
Second, there does not need to exist an annotated run from $\sbf$ to $\tt$; we only prove existence of a (non-annotated) run from $\conf(\sbf)$ to $\conf(\tt)$.
The justification of both these phenomena is the same:
while creating the run from $\conf(\sbf)$, we completely ignore the annotations contained in the topmost $r$-stack of $\sbf$
(while an annotated run from $\sbf$ necessarily follows them);
instead, as long as the topmost $r$-stack of $\st(\sbf)$ controls the run, we copy steps of the run $R$ 
(only after leaving the topmost $r$-stack, we start using the annotations from $\sbf$).
In a sense, $f_R$ describes how much can be lost while ignoring annotations in the topmost $r$-stack of $\sbf$ 
(recall that, while ignoring annotation, this $r$-stack is the same as in $R(0)$, thus fixed; only the annotations are not fixed).

Before proving Lemma~\ref{lem:rownowazne} we show how Theorem~\ref{thm:types} follows from it.
For this purpose the inequalities regarding $\low$ and $\high$ are redundant; they are used later to prove Theorem~\ref{thm:stypes}.

\proof[Proof of Theorem~\ref{thm:types}]
	Recall that we are given a $k$-upper run $R$, and a configuration $c$ having the same $(\calA,\phi)$-type and the same positionless topmost $k$-stack as $R(0)$.
	Consider the run descriptor $\tau=(\pi_1(R(|R|)),\emptyset,\dots,\emptyset,\np)$ and observe that $\tau\in\type_{\calA,\phi}(R(|R|))$ 
	(we annotate the topmost $0$-stack of $R(|R|)$ by the derivation tree from Definition~\ref{def:types}(\ref{pkt:dt-empty})).
	Applying Lemma~\ref{lem:rownowazne} we obtain a run descriptor $\sigma\in\type_{\calA,\phi}(R(0))=\type_{\calA,\phi}(c)$.
	Then we take any well-formed annotated $n$-stack $\sbf$ such that $\type(\tops^0(\sbf))=\{\sigma\}$ and $\conf(\sbf)=c$,
	existing by the definition of $\type_{\calA,\phi}$.
	Since $\tops^k(\conf(\sbf))\cong\tops^k(R(0))$,
	from Property~$(\star)$ of Lemma~\ref{lem:rownowazne} we obtain a run $S$ that starts in $c$ and is $(k,\phi)$-parallel to $R$, as required.
\qed

Below, we give an auxiliary lemma, showing how to construct a function $f_R$ needed for Lemma~\ref{lem:rownowazne}.

\begin{lem}\lab{lem:swap-annotations}
	Let $k\in\prz{0}{n}$, and let $\vv^k$ be a well-formed singular annotated $k$-stack.
	Then there exists a monotone function $f_{\vv^k}\colon\Nat\to\Nat$ such that for all well-formed singular annotated $n$-stacks $\sbf, \sbf'$ with
	$\sbf=\sbf^n:\sbf^{n-1}:\dots:\sbf^{k+1}:\sbf^k$ and $\sbf'=\sbf^n:\sbf^{n-1}:\dots:\sbf^{k+1}:\vv^k$ and $\st(\sbf^k)=\st(\vv^k)$ 
	it holds that $\low(\sbf)\leq f_{\vv^k}(\low(\sbf'))$ and $f_{\vv^k}(\high(\sbf))\geq\high(\sbf')$.
\end{lem}

\proof
	We first define the function $f_{\vv^k}$, and then we show that it satisfies the thesis.
	Suppose that $k$ and $\vv^k$ are fixed.
	We construct $f_{\vv^k}(N)$ by induction on $N$.
	Consider some $N\in\Nat$.
	First, we ensure that $f_{\vv^k}(N)\geq N+\low(\sbf^k)$ for all annotated $k$-stacks $\sbf^k$ such that $\st(\sbf^k)=\st(\vv^k)$.
	This is possible, because $\low(\sbf^k)$ equals the number of productive run descriptors altogether in the types of all $0$-stacks in $\sbf^k$.
	So, although there are infinitely many annotated $k$-stacks $\sbf^k$ such that $\st(\sbf^k)=\st(\vv^k)$, 
	the value of $\low(\sbf^k)$ is bounded by the number of $0$-stacks in $\st(\sbf^k)$ times $|\calT^0|$.
	Next, we ensure that $f_{\vv^k}(N)\geq\pow(a_n,a_{n-1},\dots,a_{k+1},\high(\vv^k))$ for all tuples $(a_n,a_{n-1},\dots,a_{k+1},a_k)$ of positive integers 
	such that $\pow(a_n,a_{n-1},\dots,a_{k+1},a_k)=N$.
	Notice that there are only finitely many such tuples (in particular, none of $a_i$ may be greater than $N$).
	Finally, we ensure that $f_{\vv^k}(N)\geq f_{\vv^k}(N-1)$ (unless $N=0$), in order to ensure monotonicity of $f_{\vv^k}$
	(since we are defining $f_{\vv^k}$ by induction, $f_{\vv^k}(N-1)$ is already defined).
	
	Consider now well-formed singular annotated $n$-stacks $\sbf, \sbf'$ such that $\sbf=\sbf^n:\sbf^{n-1}:\dots:\sbf^{k+1}:\sbf^k$,
	and $\sbf'=\sbf^n:\sbf^{n-1}:\dots:\sbf^{k+1}:\vv^k$, and $\st(\sbf^k)=\st(\vv^k)$.
	Using Proposition~\ref{prop:low-high-singular} and properties of $f_{\vv^k}$ ensured in its definition, we obtain the required inequalities:
	\begin{align*}
		\low(\sbf)&=\low(\sbf^k)+\sum_{i=k+1}^n\low(\sbf^i)\leq f_{\vv^k}\Big(\sum_{i=k+1}^n\low(\sbf^i)\Big)\\
			&\leq f_{\vv^k}\Big(\low(\vv^k)+\sum_{i=k+1}^n\low(\sbf^i)\Big)=f_{\vv^k}(\low(\sbf'))\,,\\
		f_{\vv^k}(\high(\sbf))&=f_{\vv^k}\big(\pow\big(\high(\sbf^n),\high(\sbf^{n-1}),\dots,\high(\sbf^{k+1}),\high(\sbf^k)\big)\big)\\
			&\geq\pow\big(\high(\sbf^n),\high(\sbf^{n-1}),\dots,\high(\sbf^{k+1}),\high(\vv^k)\big)=\high(\sbf')\,.\tag*{\qed}
	\end{align*}

\proof[Proof of Lemma~\ref{lem:rownowazne}]
%	Let $R(0)=(p,s^n:s^{n-1}:\dots:s^0)$ and $\tau=(q,\Psi^n,\Psi^{n-1},\dots,\Psi^1,f)$.
	The proof is by induction on the length of the $r$-upper run $R$.
	Proposition~\ref{prop:upper} gives us possible forms of $R$; we analyze these cases.

	If $R$ has length $0$, then we can take $\tau$ as $\sigma$ and the identity function as $f_R$;
	given $\sbf$ we take it as $\tt$, and as $S$ we take the run of length $0$ from $\conf(\sbf)$.
	
	Suppose that $R$ has length $1$, and its transition either is $\read$ or performs $\push^k_\alpha$.
	Then we construct $\sigma$ and $\vv^0$ out of $R$ and $\tau$ as in Lemma~\ref{lem:run2rd-read} or Lemma~\ref{lem:run2rd-push}, respectively.
	Recall that $\sigma\in\type_{\calA,\phi}(R(0))$, as needed.
	As $f_R$ we take the function $f_{\vv^0}$ constructed in Lemma~\ref{lem:swap-annotations} for the annotated $0$-stack $\vv^0$.
	Next, we are given a well-formed annotated $n$-stack $\sbf$ such that $\type(\tops^0(\sbf))=\{\sigma\}$ 
	and $\tops^r(\conf(\sbf))\cong\tops^r(R(0))$.
	Let $\sbf'$ be the annotated $n$-stack obtained from $\sbf$ by replacing its topmost $0$-stack with $\vv^0$.
	Because $\type(\vv^0)=\{\sigma\}=\type(\tops^0(\sbf))$, we have that $\sbf'$ is also well-formed.
	As $\tt$ we take the successor of $\sbf'$, and as $S$ the one-step run from $\conf(\sbf')$ (i.e., from $\conf(\sbf)$) to $\conf(\tt)$.
	By Property~$(\star)$ of Lemma~\ref{lem:run2rd-read} or Lemma~\ref{lem:run2rd-push}, respectively,
	the successor of $\sbf'$ indeed exists, and  $\type(\tops^0(\tt))=\{\tau\}$;
	moreover, the run $S$ performs the same transition as $R$, and in the case of $\read$ it reads the same letter, so $S$ is $(r,\phi)$-parallel to $R$.
	Recalling that $f_R$ satisfies the thesis of Lemma~\ref{lem:swap-annotations}, and using Lemma~\ref{lem:low-high-len},
	we obtain the required inequalities:
	\begin{align*}
		&\low(\sbf)\leq f_R(\low(\sbf'))\leq f_R(\sharp(S)+\low(\tt))\,,&&\mbox{and}\\
		&f_R(\high(\sbf))\geq\high(\sbf')\geq\sharp(S)+\high(\tt)\,.
	\end{align*}

	Next, suppose that $R$ has length $1$ and performs $\pop^k$ for $k\leq r$.
	We construct $\sigma$ and $\vv^k$ out of $R$ and $\tau$ as in Lemma~\ref{lem:run2rd-pop}.
	Recall that $\sigma\in\type_{\calA,\phi}(R(0))$, as needed.
	As $f_R$ we take the function $f_{\vv^k}$ constructed in Lemma~\ref{lem:swap-annotations} for the annotated $k$-stack $\vv^k$.
	Then, we are given an annotated stack $\sbf=\sbf^n:\sbf^{n-1}:\dots:\sbf^{k+1}:\sbf^k$ such that $\type(\tops^0(\sbf))=\{\sigma\}$ and $\tops^r(\conf(\sbf))\cong\tops^r(R(0))$.
	Consider $\sbf'=\sbf^n:\sbf^{n-1}:\dots:\sbf^{k+1}:\vv^k$.
	Because $\type(\tops^0(\vv^k))=\{\sigma\}=\type(\tops^0(\sbf))$, 
	by Proposition~\ref{prop:well-formed-multi} we have that $\type(\vv^k)=\{\red^k(\sigma)\}=\type(\tops^k(\sbf))$;
	in effect $\sbf'$ is also well-formed.
	As $\tt$ we take the successor of $\sbf'$, and as $S$ the one-step run from $\conf(\sbf')$ to $\conf(\tt)$.
	By Property~$(\star)$ of Lemma~\ref{lem:run2rd-pop}, the successor of $\sbf$ indeed exists, and $\type(\tops^0(\tt))=\{\tau\}$,
	and the transition of $S$ performs $\pop^k$.
	Clearly $S$ is $(r,\phi)$-parallel to $R$.
	The required inequalities are obtained in the same way as in the previous case, due to Lemmas~\ref{lem:swap-annotations} and~\ref{lem:low-high-len}.
	
	Next, suppose that $\subrun{R}{0}{1}$ performs $\push^k_\alpha$ and $\subrun{R}{1}{|R|}$ is a $k$-return, where $k\geq r+1$.
	This case is similar, but slightly more complicated.
	First, using Lemma~\ref{lem:run2rd}, we construct a run descriptor $\rho\in\type_{\calA,\phi}(R(1))$ such that $(\phi(R),\red^k(\tau))\in\ass^k(\rho)$.
	Then, we construct $\sigma$ and $\vv^0$ out of $R$, $\tau$, and $\rho$ as in Lemma~\ref{lem:run2rd-push-ret}.
	Recall that $\sigma\in\type_{\calA,\phi}(R(0))$, as needed.
	As $f_R$ we take the function $f_{\vv^0}$ constructed in Lemma~\ref{lem:swap-annotations} for the annotated $0$-stack $\vv^0$.
	When we are given $\sbf$, we proceed as follows.
	First, as $\sbf'$ we take the annotated $n$-stack obtained from $\sbf$ by replacing its topmost $0$-stack with $\vv^0$.
	As in the previous cases, $\sbf'$ is well-formed.
	Let $\SS$ be the one-step annotated run from $\sbf'$, and let $\SS(1)=\uu^n:\uu^{n-1}:\dots:\uu^0$.
	By Property~$(\star)$ of Lemma~\ref{lem:run2rd-push-ret}, $\SS$ indeed exists, 
	and $\st(\SS)$ performs $\push^k_\alpha$, and $\type(\uu^0)=\{\rho\}$, and $\tau\in\type(\tops^0(\uu^k))$.
	Because $(\phi(R),\red^k(\tau))\in\ass^k(\rho)$, 
	Lemma~\ref{lem:rd2run} gives us an annotated run $\TT$ starting in $\SS(1)$ such that $\st(\TT)$ is a $k$-return, 
	$\phi(\st(\TT))=\phi(R)$, and $\tops^k(\TT(|\TT|))=\uu^k{\restriction}_{\red^k(\tau)}$.
	As $S$ we take $\st(\SS\circ\TT)$, and as $\tt$ we take $\TT(|\TT|)$.
	Proposition~\ref{prop:restriction} implies that $\{\tau\}=\type(\tops^0(\uu^k{\restriction}_{\red^k(\tau)}))=\type(\tops^0(\tt))$.
	By Proposition~\ref{prop:push-return-jest-fajny}
	we obtain that $\tops^k(R(0))\cong\tops^k(R(|R|))$
	and $\tops^k(S(0))\cong\tops^k(S(|S|))$.
	Since $k\geq r+1$, $\tops^r(R(|R|))\cong\tops^r(S(|S|))$ as well.
	Together with $\phi(S)=\phi(\subrun{S}{1}{|S|})=\phi(\subrun{R}{1}{|R|})=\phi(R)$ this means that $R$ and $S$ are $(r,\phi)$-parallel, 
	because by definition no suffix of a $k$-return can be $(k-1)$-upper ($r$-upper).
	The inequalities are obtained as in the previous cases.

	Finally, suppose that $R$ is a composition of shorter $k$-upper runs $R_1$ and $R_2$.
	The induction assumption used for $R_2$ and for $\tau$ gives us a run descriptor $\rho\in\type_{\calA,\phi}(R_2(0))$ and a function $f_2$.
	Then, the induction assumption used for $R_1$ and for $\rho$ gives us a run descriptor $\sigma\in\type_{\calA,\phi}(R(0))$ and a function $f_1$.
	As $f_R$ we take a monotone function such that for each pair $a$, $b$ of natural numbers it holds $f_R(a)\geq f_1(a)+f_2(f_1(a))$ and $f_R(a+b)\geq f_1(a+f_2(b))$.
	
	Then, we are given a well-formed annotated $n$-stack $\sbf$ such that $\type(\tops^0(\sbf))=\{\sigma\}$ 
	and $\tops^r(\conf(\sbf))\cong\tops^r(R(0))$.
	From the induction assumption for $R_1$ we obtain a well-formed annotated $n$-stack $\uu$ such that $\type(\tops^0(\uu))=\{\rho\}$, 
	and a run $S_1$ from $\conf(\sbf)$ to $\conf(\uu)$ being $(r,\phi)$-parallel to $R_1$.
	Then, from the induction assumption for $R_2$ we obtain a well-formed annotated $n$-stack $\tt$ such that $\type(\tops^0(\tt))=\{\tau\}$,
	and a run $S_2$ from $\conf(\uu)$ to $\conf(\tt)$ being $(r,\phi)$-parallel to $R_2$.
	As $S$ we take the composition of $S_1$ and $S_2$; it is $(r,\phi)$-parallel to $R$.
	Using the inequalities from the induction assumption we obtain
	\begin{align*}
		\low(\sbf)&\leq f_1(\sharp(S_1)+\low(\uu))\leq f_1(\sharp(S_1)+f_2(\sharp(S_2)+\low(\tt)))\\
			&\leq f_R(\sharp(S_1)+\sharp(S_2)+\low(\tt))=f_R(\sharp(S)+\low(\tt))
		\,\\
		f_R(\high(\sbf))&\geq f_1(\high(\sbf))+f_2(f_1(\high(\sbf))
			\geq\sharp(S_1)+\high(\uu)+f_2(\sharp(S_1)+\high(\uu))\\
			&\geq\sharp(S_1)+f_2(\high(\uu))
			\geq\sharp(S_1)+\sharp(S_2)+\high(\tt)=\sharp(S)+\high(\tt)\,.\tag*{\qed}
	\end{align*}

\subsection{Sequence-Equivalence}\lab{sec:seq-equiv}

In the final part of this section we define sequence-equivalence, and we prove Theorem~\ref{thm:stypes}.

\begin{defi}
	Let $(c_i)_{i=1}^\infty$ be a sequence of configurations.
	We define $\stype\!\left((c_i)_{i=1}^\infty\right)\subseteq\calT^0$ to be the set of such $\sigma\in\calT^0$ that
	there exists a sequence of well-formed annotated $n$-stacks $(\sbf_i)_{i=1}^\infty$ for which $\type(\tops^0(\sbf_i))=\{\sigma\}$ and $\conf(\sbf_i)=c_i$ for each $i$, 
	and the sequence $(\high(\sbf_i))_{i=1}^\infty$ is bounded
	(notice that we require the same type $\{\sigma\}$ for all $i$).
	We say that two sequences of configurations, $(c_i)_{i=1}^\infty$ and $(d_i)_{i=1}^\infty$,
	are \emph{$(\calA,\varphi)$-sequence-equivalent} when it holds that
	$\stype\!\left((c_i)_{i=1}^\infty\right)=\stype\!\left((d_i)_{i=1}^\infty\right)$.
\end{defi}

\proof[Proof of Theorem~\ref{thm:stypes}]
	Recall that we are given a run $R\circ R'$ in which $R$ is $k$-upper and $R'$ is an $n$-return;
	we are also given two infinite sequences of configurations $c_1,c_2,\dots$ and $d_1,d_2,\dots$ that are $(\calA,\phi)$-sequence-equivalent, and 
	in which all configurations have the same $(\calA,\phi)$-type and the same positionless topmost $k$-stack as $R(0)$.
	Our goal is to construct, for each $i$, runs $S_i\circ S_i'$ from $c_i$, and $T_i\circ T_i'$ from $d_i$ in which $S_i$ and $T_i$ are $(k,\phi)$-parallel to $R$,
	and $S_i'$ and $T_i'$ are $n$-returns such that $\phi(S_i')=\phi(T_i')=\phi(R')$, and such that
	the sequences $\sharp(S_1\circ S_1'),\sharp(S_2\circ S_2'),\dots$ and $\sharp(T_1\circ T_1'),\sharp(T_2\circ T_2'),\dots$ are either both bounded or both unbounded.
	Let $\xi=(\pi_1(R'(|R'|)),\emptyset,\dots,\emptyset,\np)$. We see that $\xi\in\type_{\calA,\phi}(R'(|R'|))$,
	because we can annotate the topmost $0$-stack $(\gamma,x)$ by $\{(\dtempty\gamma,\pi_1(R'(|R'|)))\}$ and all other $0$-stacks by $\emptyset$.
	Lemma~\ref{lem:run2rd} applied to $R'$ and $\xi$ implies that $\type_{\calA,\phi}(R'(0))$ contains a run descriptor $\tau$ such that $(\phi(R'),\red^n(\xi))\in\ass^n(\tau)$.
	Then, Lemma~\ref{lem:rownowazne} applied to $R$ and $\tau$ gives us a run descriptor $\sigma\in\type_{\calA,\phi}(R(0))$ and a function $f_R$.
	We have two cases.
	
	\paragraph*{Case 1}
	Suppose first that $\sigma\in\stype\!\left((c_i)_{i=1}^\infty\right)$ (hence also $\sigma\in\stype\!\left((d_i)_{i=1}^\infty\right)$).
	Then we have a sequence of annotated $n$-stacks $(\sbf_i)_{i=1}^\infty$ such that $\type(\tops^0(\sbf_i))=\{\sigma\}$ and $\conf(\sbf_i)=c_i$ for each $i$, 
	and the sequence $(\high(\sbf_i))_{i=1}^\infty$ is bounded.
	Recall that the topmost $k$-stacks of $c_i$ and of $R(0)$ are positionless-equal, for each $i$.
	We use Property $(\star)$ of Lemma~\ref{lem:rownowazne} for the annotated stack $\sbf_i$.
	We obtain a well-formed annotated $n$-stack $\tt_i$ such that $\type(\tops^0(\tt_i))=\{\tau\}$, and a run $S_i$ from $c_i$ to $\conf(\tt_i)$ being $(k,\phi)$-parallel to $R$ and such that 
	$f_R(\high(\sbf_i))\geq\sharp(S_i)+\high(\tt_i)$.
	Next, for each $i$ we apply Lemma~\ref{lem:rd2run} for $\tt_i$ and for the pair $(\phi(R'),\red^n(\xi))$.
	We obtain an annotated run $\SS'_i$ starting in $\tt_i$ such that $\st(\SS'_i)$ is an $n$-return, $\phi(\st(\SS'_i))=\phi(R')$, and $\type(\SS'_i(|\SS'_i|))=\{\red^n(\xi)\}$.
	Let $S'_i=\st(\SS'_i)$, and $\uu_i=\SS'_i(|\SS'_i|)$.
	Thanks to Lemma~\ref{lem:low-high-len}, $\high(\tt_i)\geq\sharp(S'_i)+\high(\uu_i)$.
	Because $(\high(\sbf_i))_{i=1}^\infty$ is bounded, we see that the sequence $(\sharp(S_i\circ S'_i))_{i=1}^\infty$ is bounded as well.
	
	We perform the same construction for $(d_i)_{i=1}^\infty$, obtaining runs $T_i\circ T'_i$ from $d_i$, such that $(\sharp(T_i\circ T'_i))_{i=1}^\infty$ is bounded.

	\paragraph*{Case 2}
	This is the opposite case: we suppose that $\sigma\not\in\stype\!\left((c_i)_{i=1}^\infty\right)$.
	Recall that $\sigma\in\type_{\calA,\phi}(R(0))=\type_{\calA,\phi}(c_i)$ for each $i$.
	Using Proposition~\ref{prop:common-bound} we construct a well-formed annotated $n$-stack $\sbf_i$ such that $\type(\tops^0(\sbf_i))=\{\sigma\}$, and $\conf(\sbf_i)=c_i$, 
	and $\high(\sbf_i)\leq H(\low(\sbf_i))$ for a function $H$ not depending on $i$.
	Our assumption ensures that $(\high(\sbf_i))_{i=1}^\infty$ is unbounded, so $(\low(\sbf_i))_{i=1}^\infty$ is unbounded as well.
	We construct the runs exactly in the same way as in Case 1, but this time we concentrate on the opposite inequalities.
	For each $i$ it holds that $\low(\sbf_i)\leq f_R(\sharp(S_i)+\low(\tt_i))\leq f_R(\sharp(S_i\circ S_i')+\low(\uu_i))$.
	Additionally $\low(\uu_i)=0$, because $\type(\uu_i)=\type(\SS_i'(|\SS_i'|))=\{\red^n(\xi)\}\subseteq\calT_\np$ (cf.~Proposition~\ref{prop:pr-positive}).
	It follows that $(\sharp(S_i\circ S_i'))_{i=1}^\infty$ is unbounded, and similarly $(\sharp(T_i\circ T_i'))_{i=1}^\infty$.
\qed

\section{Milestone Configurations}\lab{sec:milestone}

In this section we define so-called milestone configurations and we show their basic properties.
The intuitions are as follows.
Consider a long run reading only stars.
Looking globally, the stack grows (or remains unchanged).
Locally, however, some parts of the stack might be constructed, and a few steps later removed.
In order to handle this behavior, we concentrate on those configurations of the run in which the stack is minimal (in appropriate sense) and will not be destroyed later;
they are called milestone configurations.

The idea of considering milestone configurations comes from Kartzow \cite{kartzow-tree-automatic}, but our definition is slightly different
(namely, their definition is relative to a run, which can be arbitrary, while our definition is absolute: we always consider the run reading only stars).
For this section we fix an $n$-DPDA $\calA$ with stack alphabet $\Gamma$ and with input alphabet $A$ containing a distinguished symbol denoted $\star$ (star).

\begin{defi}\lab{def:milestone}
	We say that a configuration $c$ is a \emph{milestone} (or a milestone configuration) if there exists an infinite run $R$ from $c$ reading only stars, 
	and an infinite set $I$ of indices such that $0\in I$, and $\subrun{R}{i}{j}\in \up^0$ for all $i,j\in I$, $i\leq j$.
\end{defi}

\begin{exa}
	Consider a DPDA of order 3.
	Suppose that there is a run that begins in a stack $\poslinv([[[a,a]]])$, and performs forever the following sequence of operations, in a loop:
	\begin{align*}
		\push^2_a\,,\ \push^3_a\,,\ \pop^1\,,\ \push^3_a\,,\ \pop^2\,,\ \push^3_a\,.
	\end{align*}
	Then the positionless topmost $2$-stack is, alternately, 
	$[[a,a]]$, or $[[a,a],[a,a]]$, or $[[a,a],[a]]$.
	This run does not read any symbols, so it is a degenerate case of an infinite run that reads only stars.
	Configurations with positionless topmost $2$-stack $[[a,a]]$ are milestones (and no other configurations in this run).
	To obtain a less degenerate case, we may consider a loop of transitions as above, but containing additionally a $\read$ transition;
	when a star is read, the loop continues (we do not care what happens when any other symbol is read).
	Then again configurations having $[[a,a]]$ as the topmost $2$-stack are milestones.
\end{exa}

If $c$ is a milestone, $R$ the (unique) infinite run from $c$ reading only stars, and $I$ a set like in the definition of a milestone, 
then for each $i\in I$ the configuration $R(i)$ is a milestone as well.
The following lemma shows that in fact the set $I$ can contain all indices $i$ for which $R(i)$ is a milestone.

\begin{lem}\lab{lem:milestones-sequence}
	Let $R$ be a run between two milestone configurations. If $R$ reads only stars, then it is $0$-upper.
\end{lem}

\proof 
	We prove by induction on $n-k$, where $k\in\prz{0}{n}$, that each run $R$ as in the lemma is $k$-upper.
	Trivially each run is $n$-upper.
	Now suppose that the thesis holds for some $k>0$,  take a run $R$ between two milestone configurations, and suppose that it reads only stars.
	Let $S$ be the infinite run that starts in $R(0)$ and reads only stars (since $R(0)$ is a milestone, the run is really infinite); $R$ is its prefix.
	Notice that we can find a milestone $S(i)$ such that $i\geq|R|$ and $\subrun{S}{0}{i}$ is $(k-1)$-upper (it can be even $0$-upper):
	it is enough to take any $i\geq|R|$ from the infinite set $I$ from Definition~\ref{def:milestone}.
	From the induction assumption we know that $\subrun{S}{|R|}{i}$ is $k$-upper.
	We conclude that $\subrun{S}{0}{|R|}=R$ is $(k-1)$-upper using Proposition~\ref{prop:pre-jest-fajne} for the run $\subrun{S}{0}{|R|}\circ\subrun{S}{|R|}{i}$.
\qed

Another important property is that in a very long run reading only stars we can find a milestone configuration.
What ``very long'' means of course depends on the size of the configuration where the run starts.

\begin{lem}\lab{lem:schowane-malo-zmieniane}
	Let $l\in\prz{1}{n}$.
	There exists a function $\beta$, assigning a natural number to every positionless $l$-stack, having the following property.
	Let $R$ be a run that reads only stars, let $s^l_{|R|}$ be an $l$-stack of $R(|R|)$, 
	and let $s^l_{i}=\hist(\subrun{R}{i}{|R|},s^l_{|R|})$ for all $i\in\prz{0}{|R|}$.
	If there exist at least $\beta(\posl(s^l_{0}))$ indices $i$ such that 
	$s^l_i=\topp^l(R(i))$,
	then for some index $i$ the configuration $R(i)$ is a milestone and $s^l_i=\topp^l(R(i))$.
\end{lem}

\begin{cor}\lab{cor:inf-milestone}
	If $R$ is an infinite run reading only stars,
	then for infinitely many indices $i$ the configuration $R(i)$ is a milestone.
\end{cor}

\proof
	To obtain a first milestone configuration, 
	it is enough to use Lemma~\ref{lem:schowane-malo-zmieniane} for $l:=n$ and for the prefix of $R$ of length $\beta(\posl(\pi_2(R(0))))$.
	We repeat this procedure for the remaining suffix of $R$.\footnote{
		There exists a direct proof of this corollary, not presented here, which is much easier than the proof of Lemma~\ref{lem:schowane-malo-zmieniane}.}
\qed

In order to get some intuitions on Lemma~\ref{lem:schowane-malo-zmieniane}, let us first see why it works for $l=n$.
In this case $s^l_0$ is just the whole $n$-stack of $R(0)$.
Moreover, the assumption that there exist at least $\beta(\posl(s^l_0))$ indices $i$ such that $s^l_i=\topp^l(R(i))$
simply expresses that $|R|+1\geq\beta(\posl(s^l_0))$.
Thus, the lemma says that if we have a long enough run that starts in a configuration with stack $s^l_0$ and reads only stars, then the run reaches a milestone configuration.
This, in turn, means that we cannot decrease the stack $s^l_0$ forever.
Indeed, recall the intuition that a milestone configuration is a minimal configuration, that is, such that the run reading only stars never visits a ``smaller'' configuration.
It is just enough to consider the infinite run reading only stars, and take the ,,smallest'' configuration visited by this run; this should be a milestone configuration.

When $l<n$, the lemma concentrates on the history of a single $l$-stack $s^l_{|R|}$
(another point of view is that it concentrates on the future of a single $l$-stack $s^l_0$).
We look at the fragments of $R$ where this $l$-stack is the topmost $l$-stack; 
the length of these fragments is required to be at least $\beta(\posl(s^l_0))$ in total.
The lemma says that regardless of what happens in other fragments of $R$, 
controlled by other parts of the $n$-stack of $R(0)$ (outside of $s^l_0$),
the stack $s^l_0$ can itself ensure that a milestone configuration is reached.

In the remaining part of the section we prove Lemma~\ref{lem:schowane-malo-zmieniane}.
Our proof strategy is as follows.
The indices $i$ for which $s^l_i$ is the topmost $l$-stack give us a decomposition of an infix of $R$ into many $l$-upper runs.
As a first step, consecutively for $k=l-1,l-2,\dots,0$ we construct a decomposition of an infix of $R$ into many $k$-upper runs.
Then, among the borders of the constructed $0$-upper runs we find two configurations having the same type.
Using Theorem~\ref{thm:types} we can replicate the $0$-upper run between them into arbitrarily many consecutive $0$-upper runs,
proving that these two configurations are milestones.

The division of an infix of $R$ into $k$-upper runs is described using $k$-advancing sets, defined as follows.
Assuming that $R$ is fixed, a set $I^k\subseteq\prz{0}{|R|}$ is called \emph{$k$-advancing} if
\begin{align*}
	\emptyset\neq I^k=\{i\in\prz{\min I^k}{\max I^k}\mid\subrun{R}{i}{\max I^k}\in\up^k\}\,.
\end{align*}
Notice that when $\min I^k\leq i\leq j\in I^k$, then $i$ belongs to $I^k$ if and only if $\subrun{R}{i}{j}$ is $k$-upper.
In other words, a $k$-advancing set not only gives us a decomposition into $k$-upper runs, but also these $k$-upper runs cannot be further subdivided into shorter $k$-upper runs.
The following auxiliary lemma describes our induction step.

\begin{lem}\lab{lem:schowane-malo-zmieniane-pom}
	Let $k\in\prz{1}{n}$, and $N\in\Nat$.
	There exists a function $f^k_N\colon\Nat\to\Nat$, having the following properties.
	Let $R$ be a run that reads only stars, and let $I^k$ be a $k$-advancing set.
	If $|I^k|\geq f^k_N(|\tops^k(R(\min I^k))|)$,
	then there exists a $(k-1)$-advancing subset $I^{k-1}\subseteq I^k$ of size at least $N$, such that 
	$\posl(\tops^{k-1}(R(\min I^{k-1})))$ is one of the $(k-1)$-stacks in $\posl(\tops^k(R(\min I^k)))$.
\end{lem}

\proof
	We prove the lemma by induction on $N$.
	For $N=1$ we can take $f^k_1(r):=1$, and then $I^{k-1}:=\{\min I^k\}$.
	Let now $N\geq 2$.
	We take
	\begin{align*}
		f^k_N(r):=1+\sum_{m=1}^rf^k_{N-1}(m+1)\,.
	\end{align*}
	Fix some $R$ and $I^k$ satisfying the assumptions.
	Let $a:=\min I^k$ and $r:=|\tops^k(R(\min I^k))|$.
	For each $j\in I^k$ denote 
	\begin{align*}
		r_j:=|\tops^k(R(j))|&&\mbox{and}&&m_j:=\min\{r_i\colon i\in I^k\land i\leq j\}\,.
	\end{align*}
	Notice that $1\leq m_j\leq r$ (because $r_a=r$) and that $m_j\geq m_{j'}$ for $j\leq j'$.
	From the formula for $f^k_N(r)$ we see that for some $m$ we have at least $f^k_{N-1}(m+1)+1$ indices $j\in I^k$ such that $m_j=m$,
	by the pigeonhole principle (if for every $m$ there were at most $f^k_{N-1}(m+1)$ such indices, 
	in total we would have at most $\sum_{m=1}^rf^k_{N-1}(m+1)=f^k_N(r)-1$ indices in $I^k$, but we have at least $f^k_N(r)$ of them).
	Choose some such $m$; let $b$ be the first index such that $m_b=m$, and $e$ the last such index.
	We see that $m=r_b$.
	
	Let $c$ be the next element of $I^k$ after $b$ (of course $c\leq e$).
	Notice that $r_c\leq r_b+1=m+1$; this follows from Proposition~\ref{prop:pre-k} used for the run $\subrun{R}{b}{c}$.
	Thus,
	\begin{align*}
		|I^k\cap\prz{c}{e}|\geq f^k_{N-1}(m+1)\geq f^k_{N-1}(r_c)\,.
	\end{align*}
	
	We use the induction assumption for $I^k\cap\prz{c}{e}$.
	We obtain a $(k-1)$-advancing subset $J^{k-1}\subseteq I^k\cap\prz{c}{e}$ of size at least $N-1$.
	We take 
	\begin{align*}
		I^{k-1}:=\{i\in\prz{b}{\min J^{k-1}}\mid\subrun{R}{i}{\max J^{k-1}}\in\up^{k-1}\}\cup J^{k-1}\,.
	\end{align*}
	We easily see that $I^{k-1}$ is $(k-1)$-advancing (we add to $J^{k-1}$ exactly these indices for which the appropriate run is $(k-1)$-upper).
	Because $r_i\geq m_i=m=r_b$ for each $i\in I^k\cap\prz{b}{e}$ (hence, in particular, for each $i\in I^k\cap\prz{b}{\max J^{k-1}}$),
	Proposition~\ref{prop:size2pre} implies that $\subrun{R}{b}{\max J^{k-1}}$ is $(k-1)$-upper.
	Thus, $I^{k-1}$ in addition to the $N-1$ elements of $J^{k-1}$ contains at least one additional element $b$.

	Finally, we show that $\posl(\tops^{k-1}(R(b)))$ 
	is one of the $(k-1)$-stacks in $\posl(\tops^k(R(a)))$.
	We know that $r_i\geq m_i>m=r_b$ for each $i\in I^k\cap\prz{a}{b-1}$, so, due to Proposition~\ref{prop:size2pre}, run $\subrun{R}{i}{b}$ is not $(k-1)$-upper.
	This means that the topmost $(k-1)$-stack of $R(b)$ was not modified since $R(a)$.
	On the other hand $\subrun{R}{a}{b}$ is $k$-upper,
	thus, indeed, the topmost $(k-1)$-stack of $R(b)$ is one of the $(k-1)$-stacks in the topmost $k$-stack of $R(a)$ (while ignoring positions annotating the stacks).
\qed

While using Theorem~\ref{thm:types} we need to ensure that the replicated run reads only stars.
For this reason, fix a finite monoid $M$, and a morphism $\varphi\colon A^*\to M$, 
such that its value $\varphi(w)$ determines whether a word $w$ consists only of stars.
Our second auxiliary lemma is used to conclude the proof of Lemma~\ref{lem:schowane-malo-zmieniane}.

\begin{lem}\lab{lem:schowane-malo-zmieniane-pom2}
	Let $S$ be a nonempty $0$-upper run reading only stars, in which $S(|S|)$ has the same $(\calA,\phi)$-type and the same topmost stack symbol as $S(0)$ (where $\phi$ as above).
	Then $S$ can be extended into a run $S\circ T\circ U$ reading only stars, where $T$ and $U$ are nonempty $0$-upper runs, 
	and $U(|U|)$ has the same $(\calA,\phi)$-type and the same topmost stack symbol as $U(0)$.
	As a consequence, $S(0)$ is a milestone.
\end{lem}

\proof
	First, we observe that for each $r\in\Nat$ we can construct a composition $S_1\circ\dots\circ S_r$ of $r$ nonempty $0$-upper runs, reading only stars, in which $S_1=S$.
	For $r=1$ this is trivially true.
	Suppose that we have such a composition for some $r$.
	Then Theorem~\ref{thm:types} applied to this composition and to $S(|S|)$
	(where we use the fact that $S(|S|)$ has the same $(\calA,\phi)$-type and the same positionless topmost $0$-stack as $(S_1\circ\dots\circ S_r)(0)$)
	gives us a run that starts in $S(|S|)$ and is $(\calA,\phi)$-parallel to $S_1\circ\dots\circ S_r$.
	Recalling the definition of being $(\calA,\phi)$-parallel, we see that this run is a composition $S'_1\circ\dots\circ S'_r$ 
	of $r$ nonempty $0$-upper runs reading only stars.
	Together with $S$ at the beginning, they give a longer composition as required.
	
	Take such a composition for $r$ equal to the number of stack symbols in our alphabet $\Gamma$, times the number of $(\calA,\varphi)$-types, plus two.
	Then, by the pigeonhole principle, we can find two indices $i,j\in\prz{2}{r}$ with $i<j$ for which $S_j(|S_j|)$ has the same $(\calA,\phi)$-type and the same topmost stack symbol as $S_i(|S_i|)$.
	Skipping the part after $S_j$, we obtain a composition $S\circ T\circ U$ as required.
	
	We can repeat the same construction for $U$, and append two more nonempty $0$-upper runs, out of which the second has equal $(\calA,\phi)$-type and topmost stack symbol at its two ends.
	Continuing this forever, we obtain an infinite run reading only stars, divided into $0$-upper runs.
	Since it starts in $S(0)$, this configuration is a milestone.
\qed

\proof[Proof of Lemma~\ref{lem:schowane-malo-zmieniane}]
	Fix some $l$-stack $s^l_{0}$.
	Let $N_0$ be equal to the number of stack symbols in the alphabet, times the number of $(\calA,\phi)$-types, plus one, where again $\phi$ checks whether a word consists only of stars.
	For $k\in\prz{1}{l}$ we take $N_k=f^k_{N_{k-1}}(r_k)$, where $r_k$ is the maximal size of a $k$-stack that appears in $s^l_{0}$, 
	and $f^k_{N_{k-1}}$ is the function from Lemma~\ref{lem:schowane-malo-zmieniane-pom}.
	We define $\beta(\posl(s^{l}_{0})):=N_{l}$.
	
	Now take a run $R$ and $l$-stacks $s^l_i$ for $i\in\prz{1}{|R|}$, such that the assumptions of the lemma are satisfied.
	First, for each $k\in\prz{0}{l}$ we want to construct a $k$-advancing set $I^k$ of size at least $N_k$,
	such that $\posl(\tops^k(R(\min I^k)))$ is one of the $k$-stacks in $\posl(s^l_{0})$.

	As $I^l$ we take the set of those indices $i$ for which $s^l_i=\topp^l(R(i))$.
	It is immediate from the definitions that $I^l$ is $l$-advancing 
	(recall that $s^l_i=\hist(\subrun{R}{i}{j},s^l_j)$ for $i\leq j$).
	By assumption $|I^l|\geq\beta(\posl(s^l_{0}))=N_l$.
	Moreover, $s^l_{\min I^l}$ was not modified from the beginning of the run (as it was not the topmost $l$-stack),
	so this $l$-stack is positionless-equal to $s^l_{0}$.

	Then by induction on $l-k$, we construct $I^{k-1}$ out of $I^k$ using Lemma~\ref{lem:schowane-malo-zmieniane-pom}.
	Notice that the size of the topmost $k$-stack of $R(\min I^k)$ is at most $r_k$, so we can indeed obtain $I^{k-1}$ of size at least $N_{k-1}$.

	Finally, we have a $0$-advancing set $I^0$ such that $|I^0|\geq N_0$. 
	Observe that in $I^0$ we can find two indices $i,j$ with $i<j$ such that $R(j)$ has the same $(\calA,\phi)$-type and the same topmost stack symbol as $R(i)$,
	by the pigeonhole principle (recall the definition of $N_0$ from the first paragraph of the proof).
	Lemma~\ref{lem:schowane-malo-zmieniane-pom2} applied to $\subrun{R}{i}{j}$ proves that $R(i)$ is a milestone.
	By construction $i\in I^0\subseteq\dots\subseteq I^l$, so $s^l_i=\topp^l(R(i))$.
\qed

\section{Pumping Lemma}\lab{sec:pumping}

In this section we present a pumping lemma, which can be used to change the number of stars read in some place of a run, without changing too much the rest of the run.
For this section we fix an $n$-DPDA $\calA$ with input alphabet $A$ containing a $\star$ symbol.
We also fix a morphism $\phi\colon A^*\to M$ into a finite monoid.

\begin{figure*}
\begin{center}
\includegraphics[scale=0.6]{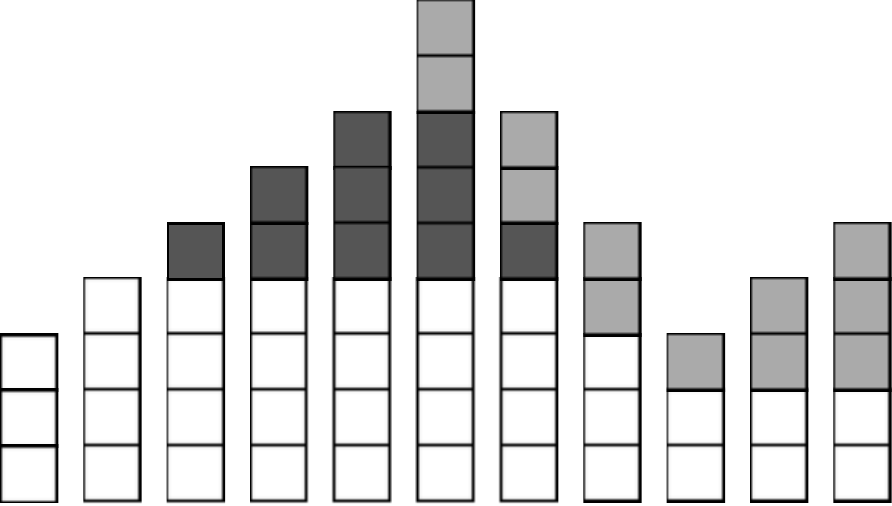}
\hspace{.8cm}
\includegraphics[scale=0.6]{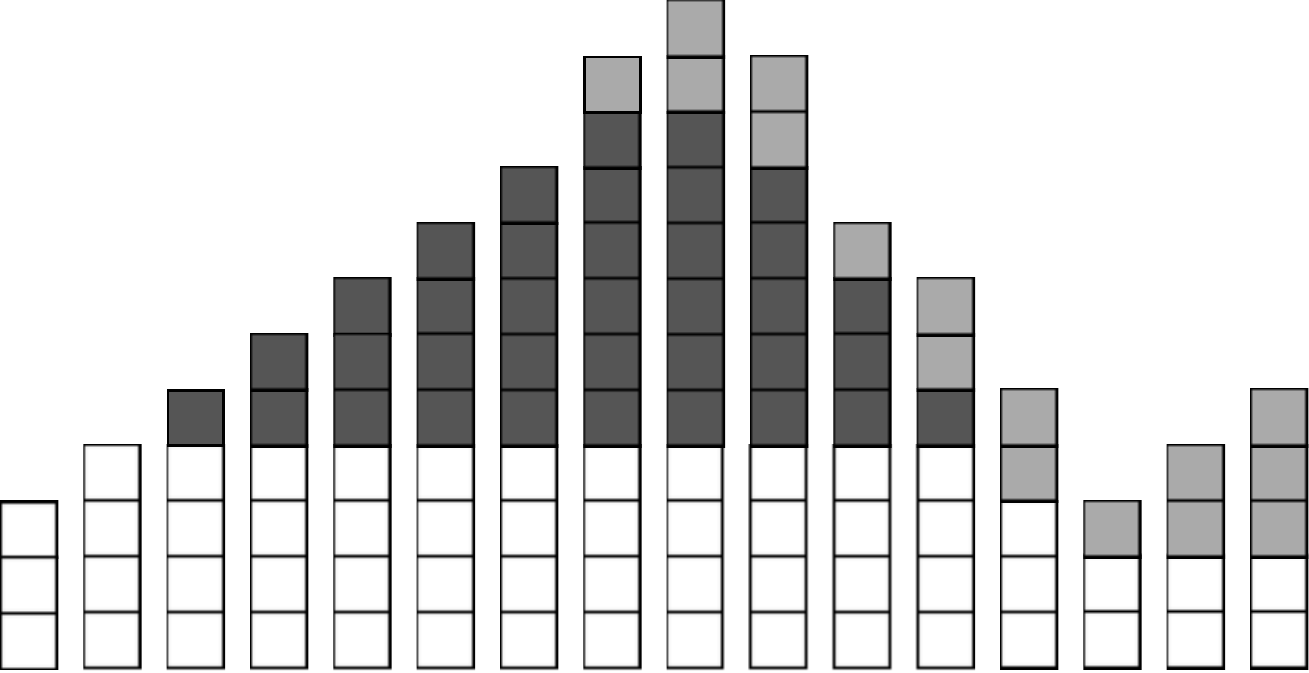}
\end{center}
\caption{An example configuration at the end of a run of a $2$-DPDA, and an analogous configuration after pumping. 
The $2$-stack grows from left to right. 
White symbols were already present in $R(0)$;
dark gray symbols were created while reading stars at the beginning of $R$;
light gray symbols were created later.}
\label{fig:pump1}
\end{figure*}

We start by an intuitive explanation of the pumping lemma.
In the situation that we consider, we have a milestone configuration from which we start runs that first read some number of stars, and later also other symbols.
One possibility is that most of these runs are $k$-upper (for some $k$), except maybe some runs reading a small number of stars at the beginning.
Then we are unable to use our pumping lemma, but we gain the knowledge that our run is $k$-upper.
The opposite situation is that there are runs from this configuration that are not $k$-upper and read arbitrarily many stars at the beginning; our pumping lemma talks about this situation.
Consider such a run $R$ of a $2$-DPDA, whose last configuration is depicted on the left of Figure~\ref{fig:pump1}.
It starts in a milestone, so its initial fragment that reads only stars is basically $0$-upper.
This means that the automaton builds on top of the stack of $R(0)$ (depicted in white), without modifying it;
also in the copies of the topmost $1$-stack the original part is not modified (the automaton can inspect this part, but then it has to be removed).
We consider a run that is not $0$-upper, so later, when we start reading other symbols than stars, the ``white part'' of the topmost $1$-stack is uncovered;
its content is the same as in $R(0)$.
By assumption there exists a run from the same configuration that reads more (arbitrarily many) stars at the beginning, and is not $0$-upper.
When it uncovers the ``white part'' of the topmost $1$-stack, this part is exactly the same as in the original run, so these runs can continue in the same way.
This is depicted on the right of the figure.

Next, we state our pumping lemma.
For uniformity of presentation, we refer there to $(-1)$-upper runs, with the assumption that no run is $(-1)$-upper.

\begin{thm}[Pumping lemma]\lab{thm:pumping}
	For each milestone configuration $c$ there exists a number $\pb(c)$ having the following property.
	Let $R\circ R'$ be a run starting in $c$, where $R$ is not $(k-1)$-upper and reads a word beginning with at least $\pb(c)$ stars, and $R'$ is $k$-upper.
	In such a situation, for each $l\in\Nat$ there exists a run $S\circ S'$ starting in $c$,
	and such that $\varphi(S)=\varphi(R)$, and $S$ reads a word beginning with at least $l$ stars, and $S'$ is $(k,\phi)$-parallel to $R'$.
\end{thm}

Let us mention that another pumping lemma for higher-order pushdown automata was presented in a former paper by the author \cite{parys-pumping}.
There are several differences between these two lemmas.
An advantage of the former lemma is that it gives a precise value for $\pb(c)$, in terms of the size of $c$.
Moreover, it works not only for deterministic PDA, but also for nondeterministic PDA in which the $\varepsilon$-closure of the configuration graph is finitely-branching.
On the other hand, the former pumping lemma is only given for $k=0$.
Additionally, it just says that the length of the word read by the run increases, not necessarily the number of stars at its beginning.
The former pumping lemma was later generalized to collapsible pushdown automata \cite{par-kar-col-pum}.

In the rest of the section we present a proof of Theorem~\ref{thm:pumping}.
Its essence is as described above: we consider the moment when the run ceases to be $(k-1)$-upper.
One possibility is as depicted in Figure~\ref{fig:pump1}: this happens during a $\pop^k$ operation; 
our topmost $k$-stack becomes positionless-equal to $\pop^k(\tops^k(R(0)))$.
This can also happen during a $\pop^r$ operation for some $r> k$.
Then we can obtain another topmost $k$-stack, but altogether we have only finitely many possibilities.
At least one of these possibilities happens for runs reading arbitrarily large number of stars at the beginning, by the pigeonhole principle; we can stick to this possibility.
Next, when we change the number of stars read at the beginning, we still land in a configuration having the same positionless topmost $k$-stack as in the original run, 
when the run ceases to be $(k-1)$-upper;
from this configuration we can mimic the rest of the original run.
When $k<n-1$, the type of the rest of the stack is important as well 
(the latter fragment of the run can perform returns visiting interiors of our stack; the existence of such returns is described by the type).
This is not a problem, since the type comes from a finite set, so we can assume that it is fixed as well.

The most difficult part of the proof is to show that indeed when the run ceases to be $(k-1)$-upper, there are only finitely many possible shapes for the topmost $k$-stack.
This is shown in Corollary~\ref{cor:pomp-skonczenie-wiele-klas}.
It is based on Lemma~\ref{lem:pomp-skonczenie-wiele-klas}, in which we analyze the situation just after reading the stars.

In order to state Lemma~\ref{lem:pomp-skonczenie-wiele-klas}, we need two definitions.
For a run $R$ starting in a configuration $c$, and for a $k$-stack $s^k$ in some configuration $R(i)$, where $k\in\prz{1}{n}$, 
we say that $s^k$ is \emph{$c$-clear} in $R(i)$ (with respect to $R$) when $\hist(\subrun{R}{0}{i},\topp^{k-1}(s^k))\neq\topp^{k-1}(c)$.
Moreover, for a configuration $c$, and for $k\in\prz{1}{n}$, let $\calS^k(c)$ be the smallest set of positionless $k$-stacks such that 
if $R$ is a run that starts in $c$ and reads only stars, and $s^k$ is a $k$-stack of $R(|R|)$ that is $c$-clear with respect to $R$, then $\posl(s^k)\in\calS^k(c)$.

\begin{lem}\lab{lem:pomp-skonczenie-wiele-klas}
	For each milestone configuration $c$, and for $k\in\prz{1}{n}$, the set $\calS^k(c)$ is finite.
\end{lem}

\proof
	Let $\calX(c)$ be the set containing all positionless $k$-stacks of $c$, 
	and additionally $\posl(\pop^k(\topp^k(c)))$; clearly $\calX(c)$ is finite.
	We claim that every positionless $k$-stack in $\calS^k(c)$ can be obtained from a positionless $k$-stack $s^k_{-}\in \calX(c)$ by applying at most $\beta(s^k_{-})$ 
	$\push$ and $\pop$ operations, where $\beta$ is the function from Lemma~\ref{lem:schowane-malo-zmieniane};
	this immediately implies that $\calS^k(c)$ is finite.
	
	Fix a run $R$ that starts in $c$ and reads only stars, and fix a $c$-clear $k$-stack $s^k$ of $R(|R|)$.
	Consider the smallest index $i$ for which the $k$-stack $\hist(\subrun{R}{i}{|R|},s^k)$ is $c$-clear in $R(i)$; denote $t^k=\hist(\subrun{R}{i}{|R|},s^k)$.
	We claim that $\posl(t^k)\in\calX(c)$.
	Indeed, either $i=0$ and $t^k$ is one of the $k$-stacks of $c$,
	or $\hist(\subrun{R}{i-1}{|R|},s^k)$ is not $c$-clear in $R(i-1)$.
	In the latter case, this $k$-stack becomes $c$-clear in the next configuration, so necessarily this is the topmost $k$-stack,
	and the operation between these configurations is $\pop^k$.
	We see that $\subrun{R}{0}{i-1}$ is $(k-1)$-upper, and $\subrun{R}{0}{i}$ is not $(k-1)$-upper.
	Proposition~\ref{prop:pomp-bez-zmian} implies that $\subrun{R}{0}{i}$ is a $k$-return, and Proposition~\ref{prop:return-jest-fajny} implies that 
	$t^k=\topp^k(R(i))\cong\pop^k(\topp^k(c))$; thus, $\posl(t^k)\in\calX(c)$.

	Observe that $t^k$ can be changed in $\subrun{R}{i}{|R|}$ only when it is the topmost $k$-stack.
	If there exist at most $\beta(\posl(t^k))$ indices $j\in\prz{i}{|R|}$ 
	such that $\hist(\subrun{R}{j}{|R|},s^k)=\topp^k(R(j))$,
	then $s^k$ can be obtained from $t^k$ by applying at most $\beta(\posl(t^k))$ $\push$ and $\pop$ operations, as we wanted to prove.

	It remains to prove that indeed there are at most $\beta(\posl(t^k))$ indices $j\in\prz{i}{|R|}$ such that $\hist(\subrun{R}{j}{|R|},s^k)=\topp^k(R(j))$.
	Suppose to the contrary that there are more than $\beta(\posl(t^k))$ such indices $j$.
	Then we can use Lemma~\ref{lem:schowane-malo-zmieniane} for $\subrun{R}{i}{|R|}$;
	it gives us an index $j$ such that the configuration $R(j)$ is a milestone and $\hist(\subrun{R}{j}{|R|},s^k)=\topp^k(R(j))$.
	Because both $c$ and $R(j)$ are milestones, we know that $\subrun{R}{0}{j}$ is $0$-upper, thanks to Lemma~\ref{lem:milestones-sequence}.
	One case is that $i=0$; then $\hist(R,s^k)\neq\topp^k(c)$ 
	(because, by the definition of $i$, $\hist(R,s^k)$ is $c$-clear in $c$),
	and we know that $\hist(\subrun{R}{j}{|R|},s^k)=\topp^k(R(j))$,
	so $\subrun{R}{0}{j}$ is not $k$-upper; in particular it cannot be $0$-upper.
	Otherwise, as already observed, $\subrun{R}{0}{i}$ is not $(k-1)$-upper and $\hist(\subrun{R}{i}{|R|},s^k)=\topp^k(R(i))$, which implies that $\subrun{R}{i}{j}$ is $k$-upper.
	But $\subrun{R}{0}{j}$ is $(k-1)$-upper, so this contradicts Proposition~\ref{prop:pre-jest-fajne} applied for $\subrun{R}{0}{i}\circ\subrun{R}{i}{j}$.
\qed

\begin{cor}\lab{cor:pomp-skonczenie-wiele-klas}
	For each milestone configuration $c$ there exists a finite set $\calS(c)$ of configurations having the following property.
	Let $k\in\prz{0}{n}$, let $R$ be a run starting in $c$, and let $r\in\prz{0}{|R|}$ be such that $\subrun{R}{0}{r}$ reads only stars.
	Suppose that $R$ is not $(k-1)$-upper, but for each $i\in\prz{r}{|R|-1}$ either $\subrun{R}{0}{i}$ is $(k-1)$-upper or $\subrun{R}{i}{|R|}$ is not $k$-upper.
	Then we can find a configuration $d\in \calS(c)$ having the same $(\calA,\phi)$-type and the same positionless topmost $k$-stack as $R(|R|)$.
\end{cor}

\proof
	There are only finitely many possible values of an $(\calA,\phi)$-type of a configuration.
	Thus, it is enough to show, for each $k$, that there are only finitely many possible positionless topmost $k$-stacks over all configurations $R(|R|)$ satisfying the assumptions.
	For $k=0$ this is trivial as a positionless $0$-stack contains just one symbol.
	Suppose that $k\geq 1$.
	We have two cases.
	
	First suppose that $\subrun{R}{i}{|R|}$ is $k$-upper for some $i\in\prz{r}{|R|-1}$; fix the greatest such index $i$.
	Then by assumption $\subrun{R}{0}{i}$ is $(k-1)$-upper, but $R$ is not.
	This is possible only when $i=|R|-1$ (thanks to Proposition~\ref{prop:pre-k} used for $\subrun{R}{i}{|R|}$).
	Proposition~\ref{prop:pomp-bez-zmian} says that $R$ is necessarily a $k$-return.
	Thus, $\topp^k(R(|R|))\cong\pop^k(\topp^k(R(0)))$ (cf.~Proposition~\ref{prop:return-jest-fajny}); the content of this $k$-stack is fixed.
	
	The other case is that $\subrun{R}{i}{|R|}$ is not $k$-upper for every $i\in\prz{r}{|R|-1}$.
	This means that the topmost $k$-stack of $R(|R|)$ is an unchanged copy of some $k$-stack of $R(r)$.
	As $R$ is not $(k-1)$-upper, this $k$-stack of $R(r)$ is $c$-clear; it thus belongs to the set $\calS^k(c)$, which is finite by Lemma~\ref{lem:pomp-skonczenie-wiele-klas}.
\qed

\proof[Proof of Theorem~\ref{thm:pumping}]
	Consider the infinite run $P$ starting at the milestone configuration $c$ and reading only stars.
	Consider first the degenerate case when in $P$ only finitely many stars are read.
	As $\pb(c)$ we take their number, plus one.
	Then the thesis is satisfied trivially, as there is no run that starts in $c$ and reads a word beginning with $\pb(c)$ stars.
	So for the rest of the proof suppose that $P$ reads infinitely many stars.
	
	Let $\calS(c)$ be the set from Corollary~\ref{cor:pomp-skonczenie-wiele-klas} (used for $c$).
	For each $i\geq 1$ we define the set $T_i\subseteq\prz{0}{n}\times\calS(c)\times M$ as follows.
	A triple $(j,d,m)$ belongs to $T_i$ if and only if there exists a run $R$ from $c$ 
	such that the word read by $R$ begins with (at least) $i$ stars, and $\phi(R)=m$, and $R(|R|)$ has the same $(\calA,\phi)$-type and the same positionless topmost $j$-stack as $d$.
	By definition $T_{i+1}\subseteq T_i$ (for each $i$), and there are only finitely many possible sets, so from some moment every $T_i$ is the same.
	As $\pb(c)$ we take a positive number such that $T_i=T_{\pb(c)}$ for all $i\geq\pb(c)$.
	
	Consider now a run $R\circ R'$ starting in $c$, where $R$ is not $(k-1)$-upper and reads a word beginning with at least $\pb(c)$ stars, and $R'$ is $k$-upper, for some $k\in\prz{0}{n}$.
	Consider also a number $l$.
	Our goal is to construct a run $S\circ S'$ starting in $c$ and such that $\varphi(S)=\varphi(R)$, and $S$ reads a word beginning with at least $l$ stars, and $S'$ is $(k,\phi)$-parallel to $R'$.
	Without loss of generality, we can assume that $l\geq\pb(c)$.
	Let $r$ be an index such that $\subrun{R}{0}{r}$ reads exactly $\pb(c)$ stars.
	Without loss of generality, we can assume that there is no $i\in\prz{r}{|R|-1}$ such that $\subrun{R}{0}{i}$ is not $(k-1)$-upper and $\subrun{R}{i}{|R|}$ is $k$-upper 
	(if such $i$ exists, we move the subrun $\subrun{R}{i}{|R|}$ to $R'$, that is, we use the pumping lemma for $\subrun{R}{0}{i}\circ(\subrun{R}{i}{|R|}\circ R')$, 
	and then in the resulting $S'$ we find the subrun $(k,\phi)$-parallel to $\subrun{R}{i}{|R|}$ and we move it back to $S$).
	
	We use Corollary~\ref{cor:pomp-skonczenie-wiele-klas} for $R$ and $r$; its assumptions are satisfied thanks to our ``without loss of generality'' assumption.
	We obtain some $d\in \calS(c)$ that has the same $(\calA,\phi)$-type and the same positionless topmost $k$-stack as $R(|R|)$.
	It means that $(k,d,\phi(R))\in T_{\pb(c)}$.
	Because $T_{\pb(c)}=T_l$, there exists a run $S$ from $c$ such that the word read by $S$ begins with (at least) $l$ stars, and $\phi(S)=\phi(R)$, and
	$S(|S|)$ has the same $(\calA,\phi)$-type and the same topmost $k$-stack as $R(|R|)$.
	
	Finally, we use Theorem~\ref{thm:types} for $R'$ in order to obtain an accepting run $S'$ that starts in $S(|S|)$ and is $(k,\phi)$-parallel to $R'$.
\qed

\section{Why $U$ Cannot Be Recognized?}\lab{sec:ostatni}

In this section we prove that the language $U$ cannot be recognized by a deterministic higher-order pushdown automaton.
Notice that our techniques presented in previous sections were quite general (not too much related to the $U$ language).
We believe that they can be useful for other purposes, for instance, to analyze behavior of some automata (in particular automata whose main objective is to count and compare 
the number of times a symbol appears on the input).

Of course our proof is by contradiction: suppose that for some $n$ we have an $(n-1)$-DPDA recognizing $U$.
We construct an $n$-DPDA $\calA$ that works as follows.
First it performs a $\push^n$ operation.
Then it simulates the $(n-1)$-DPDA (not using the $\push^n$ and $\pop^n$ operations).
When the $(n-1)$-DPDA is going to accept, $\calA$ performs a $\pop^n$ operation and afterwards accepts.
Clearly, $\calA$ recognizes $U$ as well (here we use the fact that no word in $U$ is a prefix of another word in $U$).
Such a normalization allows us to use Theorem~\ref{thm:stypes}, as in $\calA$ every accepting run is an $n$-return.

Fix a finite monoid $M$ and a morphism $\lambda\colon A^*\to M$ that checks whether a word 
is of the form $\sharp^*$ (some number of $\sharp$ symbols),
or of the form $\star^*]\star^*$ (a closing bracket surrounded by some number of stars), or of neither of these two forms.
This means that $\lambda(u)\neq \lambda(v)$ for all words $u,v$ being of different forms.
Let $N$ be the number of equivalence classes of the $(\calA,\lambda)$-sequence-equivalence relation, times the number of $(\calA,\lambda)$-types, plus one.
Consider the following words:
\begin{align*}
	&w_0=[\,,\\
	&w_{k+1}=w_k^N]^N[&&\mbox{for }k\in\prz{0}{n-1}\,,
\end{align*}
where the number in the superscript (in this case $N$) denotes the number of repetitions of a word.
For a word $w$, its \emph{pattern} is a word obtained from $w$ by removing its letters other than brackets (leaving only brackets).
Fix a morphism $\phi\colon A^*\to M$ such that from its value $\phi(w)$ we can deduce
\begin{itemize}
\item	whether the word $w$ contains the $\sharp$ symbol, and
\item	whether the pattern of $w$ is longer than $|w_n|$ (recall that $n$ is the order of $\calA$), and
\item	the exact value of the pattern of $w$, whenever this pattern is not longer than $|w_n|$.
\end{itemize}

We fix a run $R$, and an index $z(w)$ for each prefix $w$ of $w_n$, such that the following holds.
The run $R$ begins in the initial configuration.
Between $R(0)$ and $R(z(\varepsilon))$ only stars are read.
For each prefix $w$ of $w_n$, the configuration $R(z(w))$ is a milestone.
Just after $z(w)$, the run $R$ reads $\pb(R(z(w)))$ stars, where $\pb$ is the function from Theorem~\ref{thm:pumping} used for morphism $\phi$.
If $w=va$ (where $a$ is a single letter), the word read by $R$ between $R(z(v))$ and $R(z(w))$ consists of $a$ surrounded by some number of stars.
Of course such a run $R$ exists: we read stars until we reach a milestone (succeeds thanks to Corollary~\ref{cor:inf-milestone}),
then we read as many stars as required by the pumping lemma, then we read the next letter of $w_n$, and so on (because $\calA$ accepts $U$, it will never block).

It is important to analyze relations between configurations $R(z(v))$ for some prefixes $v$ of $w_n$.
In order to avoid complicated subscripts, for any prefixes $v,w$ of $w_n$ we denote $\Rz{v,w}:=\subrun{R}{z(v)}{z(w)}$.

By construction of $\calA$, for every prefix $v$ of $w_n$ the run $\Rz{v,w_n}$ is $(n-1)$-upper (as we never perform a $\pop^n$ operation
before reading some $\sharp$ symbol).
This contradicts the following key lemma (taken for $k=n-1$ and $u=\varepsilon$).

\begin{lem}
	Let $k\in\prz{-1}{n-1}$, and let $u$ be a word such that $uw_{k+1}$ is a prefix of $w_n$.
	Then there exist a prefix $v$ of $w_{k+1}$ such that $v\neq w_{k+1}$ and
	$\Rz{uv,uw_{k+1}}$ is not $k$-upper.
\end{lem}

\proof
	The proof is by induction on $k$.
	For $k=-1$ this is obvious, as no run is $(-1)$-upper (we take $v=\varepsilon$).

	Let now $k\geq 0$.
	Figure~\ref{fig2} may be helpful in finding different runs present in the proof below.
	Suppose that the thesis of the lemma does not hold.
	Then for each prefix $v$ of $w_{k+1}$ the run $\Rz{uv,uw_{k+1}}$ is $k$-upper.
	From this we get the following property $\heartsuit$.
	\begin{quote}
		Let $v'$ be a prefix of $w_{k+1}$, and $v$ a prefix of $v'$.
		Then $\Rz{uv,uv'}$ is $k$-upper.
	\end{quote}
	
	In the proof, we construct two sequences of accepting runs, with many extra stars inserted in two different places.
	Namely, in one sequence we insert extra stars before the last opening bracket that is not closed, and in the other sequence---after this bracket.
	In effect, the number of sharp symbols read by runs in one of these sequences should be unbounded, and in the other---bounded.
	The sequences are constructed in such a way that this violates Theorem~\ref{thm:stypes},
	which says that the number of sharp symbols read by runs in these sequences is either bounded in both sequences or unbounded in both sequences.

	Now we come to details.
	By the induction assumption (where $uw_k^{i-1}$ is taken as $u$), for each $i\in\prz{1}{N}$ there exists a prefix $v_i$ of $w_k$ such that
	$\Rz{uw_k^{i-1}v_i,uw_k^i}$ is not $(k-1)$-upper.
	As $\Rz{uw_k^i,uw_k^N}$ is $k$-upper (property $\heartsuit$), from Proposition~\ref{prop:pre-jest-fajne} we know that $\Rz{uw_k^{i-1}v_i,uw_k^N}$ cannot be $(k-1)$-upper as well.
	
	Now we are ready to use the pumping lemma (Theorem~\ref{thm:pumping}). 
	For each $i\in\prz{1}{N}$ we use it for $\Rz{uw_k^{i-1}v_i,uw_k^N}\circ\Rz{uw_k^N,uw_{k+1}}$.
	Recall from the definition of $R$
	that the word read by $\Rz{uw_k^{i-1}v_i,uw_k^N}$ begins with such a number of stars that the pumping lemma can be used.
	For each number $l$ we obtain a run $S_{i,l}\circ S_{i,l}'$, such that $\phi(S_{i,l})=\phi(\Rz{uw_k^{i-1}v_i,uw_k^N})$, and $S_{i,l}$ reads a word beginning with at least $l$ stars,
	and $S_{i,l}'$ is $(k,\phi)$-parallel to $\Rz{uw_k^N,uw_{k+1}}$; let $d_{i,l}=S_{i,l}(|S_{i,l}|)$.
	Notice that the run $\subrun{R}{0}{z(uw_k^{i-1}v_i)}\circ S_{i,l}$, starts in the initial configuration, ends in $d_{i,l}$, and reads a word having pattern $uw_k^N$.
	
	\begin{figure}
	\begin{center}
	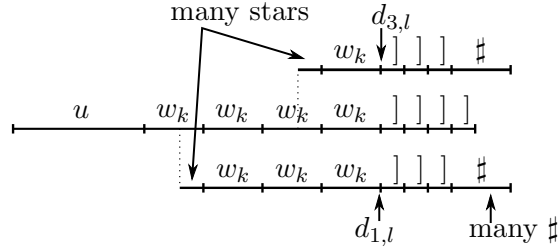
	\end{center}
	\caption{Illustration of runs appearing in the proof (where $N=4$, $x=1$, $y=3$). Recall that stars can appear between letters of the words.}
	\label{fig2}
	\end{figure}
	
	Because there are finitely many possible $(\calA,\lambda)$-types, 
	we can assume that $\type_{\calA,\lambda}(d_{i,l})=\type_{\calA,\lambda}(d_{i,j})$ for each $i\in\prz{1}{N}$ and each $l$ and $j$.
	Indeed, we can choose (for each $i$ separately) some value of $\type_{\calA,\lambda}(d_{i,l})$ that appears infinitely often, 
	and then we take the subsequence of only these $d_{i,l}$ that give this value.
	
	Since there are more possible indices $i\in\prz{1}{N}$ than the number of classes of the $(\calA,\lambda)$-sequence-equivalence relation, 
	times the number of $(\calA,\lambda)$-types, there have to exist two indices $x$, $y$ with $1\leq x<y\leq N$ such that 
	$\type_{\calA,\lambda}(d_{x,1})=\type_{\calA,\lambda}(d_{y,1})$, and the sequences
	$d_{x,1},d_{x,2},\dots$ and $d_{y,1},d_{y,2},\dots$ are $(\calA,\lambda)$-sequence-equivalent.
	From now we fix these two indices $x,y$.
	Furthermore, because $S_{i,l}'$ is $(k,\phi)$-parallel to $\Rz{uw_k^N,uw_{k+1}}$ for each $i\in\prz{1}{N}$ and each $l$,
	we know that the topmost $k$-stacks of all $d_{x,l}$ and of all $d_{y,l}$ are positionless-equal.
	
	Let $R'$ be a prefix of $S_{x,1}'$ that is $(k,\phi)$-parallel to $\Rz{uw_k^N,uw_k^N]^{N-x}}$.
	Notice that $R'$ consists of $N-x$ runs, each of which is $k$-upper and reads a word of the form $\star^*]\star^*$ (a closing bracket surrounded by some number of stars).
	Let also $R''$ be an $n$-return that starts in $R'(|R'|)$ and reads only $\sharp$ symbols 
	(because $\calA$ recognizes $U$, there is an accepting run $R''$ that starts in $R'(|R'|)$ and reads only $\sharp$ symbols;
	by construction of $\calA$, it is an $n$-return).
	
	Finally, we use Theorem~\ref{thm:stypes} for $\lambda$ (as $\phi$), sequences
	$d_{x,1},d_{x,2},\dots$ (as $c_1,c_2,\dots$) and 
	$d_{y,1},d_{y,2},\dots$ (as $d_1,d_2,\dots$), and for run the $R'\circ R''$.\footnote{
		Notice that we cannot use $\Rz{uw_k^N,uw_k^N]^{N-x}}$ instead of $R'$, because do not know anything about the $(\calA,\lambda)$-type of $R(z(uw_k^N)))$.}
	As noticed above (in particular because $R'(0)=d_{x,1}$), the configurations $R'(0)$, and $d_{x,l}$, and $d_{y,l}$ for each $l$ all have the same $(\calA,\lambda)$-types and positionless topmost $k$-stacks.
	Thus, the assumptions of the theorem are satisfied.
	For each $l$, we obtain runs $S_l=S_l'\circ S_l''$ (from $d_{x,l}$) and $T_l=T_l'\circ T_l''$ (from $d_{y,l}$).
	The word read by any of these runs contains $N-x$ closing brackets with some number of stars around them,
	and after them some number of $\sharp$ symbols.
	
	The runs $\subrun{R}{0}{z(uw_k^{x-1}v_x)}\circ S_{x,l}\circ S_l$ and $\subrun{R}{0}{z(uw_k^{y-1}v_y)}\circ S_{y,l}\circ T_l$ for each $l$ have pattern $uw_k^N]^{N-x}$.
	In this pattern the last opening bracket that is not closed is the last bracket of the $x$-th $w_k$ after $u$.
	Recall that configurations $d_{x,l}$ were obtained by pumping inside the $x$-th $w_k$, so before this bracket;
	for $l\to\infty$ the number of stars inserted there is unbounded.
	From the definition of the language $U$ it follows that the sequence $\sharp(S_1),\sharp(S_2),\dots$ has to be unbounded.
	On the other hand, configurations $d_{y,l}$ were obtained by pumping inside the $y$-th $w_k$, so 
	after the last opening bracket that was not closed (as $y>x$).
	For each $l$ the number of stars before this bracket is the same.
	From the definition of the language $U$ it follows that the sequence $\sharp(T_1),\sharp(T_2),\dots$ has to be constant, hence bounded.
	This contradicts the thesis of Theorem~\ref{thm:stypes}, which says that these sequences are either both bounded or both unbounded.
\qed

\section*{Acknowledgment}

We thank A.~Kartzow, and the anonymous reviewers of this paper and its conference version for their constructive comments.

\bibliographystyle{alpha}
\bibliography{bib}

\end{document}

%% file: pics/biegi.pdf_tex
%% Creator: Inkscape inkscape 0.92.2, www.inkscape.org
%% PDF/EPS/PS + LaTeX output extension by Johan Engelen, 2010
%% Accompanies image file 'biegi.pdf' (pdf, eps, ps)
%%
%% To include the image in your LaTeX document, write
%%   \input{<filename>.pdf_tex}
%%  instead of
%%   \includegraphics{<filename>.pdf}
%% To scale the image, write
%%   \def\svgwidth{<desired width>}
%%   \input{<filename>.pdf_tex}
%%  instead of
%%   \includegraphics[width=<desired width>]{<filename>.pdf}
%%
%% Images with a different path to the parent latex file can
%% be accessed with the `import' package (which may need to be
%% installed) using
%%   \usepackage{import}
%% in the preamble, and then including the image with
%%   \import{<path to file>}{<filename>.pdf_tex}
%% Alternatively, one can specify
%%   \graphicspath{{<path to file>/}}
%% 
%% For more information, please see info/svg-inkscape on CTAN:
%%   http://tug.ctan.org/tex-archive/info/svg-inkscape
%%
\begingroup%
  \makeatletter%
  \providecommand\color[2][]{%
    \errmessage{(Inkscape) Color is used for the text in Inkscape, but the package 'color.sty' is not loaded}%
    \renewcommand\color[2][]{}%
  }%
  \providecommand\transparent[1]{%
    \errmessage{(Inkscape) Transparency is used (non-zero) for the text in Inkscape, but the package 'transparent.sty' is not loaded}%
    \renewcommand\transparent[1]{}%
  }%
  \providecommand\rotatebox[2]{#2}%
  \ifx\svgwidth\undefined%
    \setlength{\unitlength}{205.84203828bp}%
    \ifx\svgscale\undefined%
      \relax%
    \else%
      \setlength{\unitlength}{\unitlength * \real{\svgscale}}%
    \fi%
  \else%
    \setlength{\unitlength}{\svgwidth}%
  \fi%
  \global\let\svgwidth\undefined%
  \global\let\svgscale\undefined%
  \makeatother%
  \begin{picture}(1,0.41828562)%
    \put(0,0){\includegraphics[width=\unitlength,page=1]{biegi.pdf}}%
    \put(0.11147049,0.22377001){\color[rgb]{0,0,0}\makebox(0,0)[lb]{\smash{$u$}}}%
    \put(0.19720895,0.3292296){\color[rgb]{0,0,0}\makebox(0,0)[lb]{\smash{ }}}%
    \put(0.27154781,0.38918031){\color[rgb]{0,0,0}\makebox(0,0)[lb]{\smash{ }}}%
    \put(0,0){\includegraphics[width=\unitlength,page=2]{biegi.pdf}}%
    \put(0.26404501,0.22490863){\color[rgb]{0,0,0}\makebox(0,0)[lb]{\smash{$w_k$}}}%
    \put(0.58855051,0.22490863){\color[rgb]{0,0,0}\makebox(0,0)[lb]{\smash{$w_k$}}}%
    \put(0.37221344,0.22490863){\color[rgb]{0,0,0}\makebox(0,0)[lb]{\smash{$w_k$}}}%
    \put(0.48311472,0.22377001){\color[rgb]{0,0,0}\makebox(0,0)[lb]{\smash{$w_k$}}}%
    \put(0.59014467,0.1144629){\color[rgb]{0,0,0}\makebox(0,0)[lb]{\smash{$w_k$}}}%
    \put(0.48197617,0.1144629){\color[rgb]{0,0,0}\makebox(0,0)[lb]{\smash{$w_k$}}}%
    \put(0.58855051,0.33307713){\color[rgb]{0,0,0}\makebox(0,0)[lb]{\smash{$w_k$}}}%
    \put(0.3738076,0.1144629){\color[rgb]{0,0,0}\makebox(0,0)[lb]{\smash{$w_k$}}}%
    \put(0.73998642,0.22490863){\color[rgb]{0,0,0}\makebox(0,0)[lb]{\smash{$]$}}}%
    \put(0.73998642,0.33307713){\color[rgb]{0,0,0}\makebox(0,0)[lb]{\smash{$]$}}}%
    \put(0.78325382,0.33307713){\color[rgb]{0,0,0}\makebox(0,0)[lb]{\smash{$]$}}}%
    \put(0.69671902,0.33307713){\color[rgb]{0,0,0}\makebox(0,0)[lb]{\smash{$]$}}}%
    \put(0.69671902,0.22490863){\color[rgb]{0,0,0}\makebox(0,0)[lb]{\smash{$]$}}}%
    \put(0.82652122,0.22490863){\color[rgb]{0,0,0}\makebox(0,0)[lb]{\smash{$]$}}}%
    \put(0.69831318,0.1144629){\color[rgb]{0,0,0}\makebox(0,0)[lb]{\smash{$]$}}}%
    \put(0.74158058,0.1144629){\color[rgb]{0,0,0}\makebox(0,0)[lb]{\smash{$]$}}}%
    \put(0.78484798,0.1144629){\color[rgb]{0,0,0}\makebox(0,0)[lb]{\smash{$]$}}}%
    \put(0.78325382,0.22490863){\color[rgb]{0,0,0}\makebox(0,0)[lb]{\smash{$]$}}}%
    \put(0,0){\includegraphics[width=\unitlength,page=3]{biegi.pdf}}%
    \put(0.84815492,0.33307713){\color[rgb]{0,0,0}\makebox(0,0)[lb]{\smash{$\sharp$}}}%
    \put(0.84974908,0.1144629){\color[rgb]{0,0,0}\makebox(0,0)[lb]{\smash{$\sharp$}}}%
    \put(0,0){\includegraphics[width=\unitlength,page=4]{biegi.pdf}}%
    \put(0.29055202,0.40293566){\color[rgb]{0,0,0}\makebox(0,0)[lb]{\smash{many stars}}}%
    \put(0,0){\includegraphics[width=\unitlength,page=5]{biegi.pdf}}%
    \put(0.62738465,0.00515579){\color[rgb]{0,0,0}\makebox(0,0)[lb]{\smash{$d_{1,l}$}}}%
    \put(0.65800608,0.39501783){\color[rgb]{0,0,0}\makebox(0,0)[lb]{\smash{$d_{3,l}$}}}%
    \put(0.83654102,0.00515579){\color[rgb]{0,0,0}\makebox(0,0)[lb]{\smash{many $\sharp$}}}%
    \put(0,0){\includegraphics[width=\unitlength,page=6]{biegi.pdf}}%
  \end{picture}%
\endgroup%